\documentclass[a4paper,USenglish,cleveref,thm-restate]{lipics-v2021}

\title{Subgame Perfection in Graph Games with \texorpdfstring{$\omega$}{ω}-Recognizable Preference Relations}
\titlerunning{SPEs in Graph Games with $\omega$-Recognizable Preference Relations}

\author{Véronique Bruyère}{Université de Mons (UMONS), Belgium \and \url{https://informatique-umons.be/bruyere-veronique/}}{veronique.bruyere@umons.ac.be}{https://orcid.org/0000-0002-9680-9140}{}

\author{Christophe Grandmont}{Université de Mons (UMONS), Belgium \and Université libre de Bruxelles (ULB), Belgium \and \url{https://chrisgdt.github.io/}}{christophe.grandmont@umons.ac.be}{https://orcid.org/0009-0009-4573-0123}{}

\author{Noémie Meunier}{Université de Mons (UMONS)}{noemie.meunier@gmail.com}{}{}

\author{Jean-François Raskin}{Université libre de Bruxelles (ULB), Belgium \and \url{https://verif.ulb.ac.be/jfr/}}{jean-francois.raskin@ulb.be}{https://orcid.org/0000-0002-3673-1097}{Supported by Fondation ULB (\url{https://www.fondationulb.be/en/}) and the Thelam Fondation.}

\authorrunning{V.\ Bruyère, C.\ Grandmont, N.\ Meunier, and J.-F.\ Raskin}



\nolinenumbers
\pdfoutput=1
\hideLIPIcs

\funding{This work has been supported by the Fonds de la Recherche Scientifique – FNRS under Grant n° T.0023.22 (PDR Rational).}

\Copyright{Véronique Bruyère, Christophe Grandmont, Noémie Meunier, and Jean-François Raskin}

\usepackage{multicol}
\usepackage{tikz}
\usepackage{graphics}
\usepackage{centernot}
\usepackage{pifont}
\crefname{enumi}{Condition}{Conditions}

\newtheorem*{problems*}{Problems}

\mathchardef\mhyphen="2D

\newcommand{\<}{\langle}
\renewcommand{\>}{\rangle}
\renewcommand{\|}{\upharpoonright}

\newcommand{\Reals}{\mathbb{R}}

\newcommand{\Q}{\mathbb{Q}}

\newcommand{\ssetminus}{\! \setminus \!}
\def\rest#1#2{#1_{\restriction#2}}

\newcommand{\val}{\mathsf{Val}}

\newcommand{\aut}[1]{\ensuremath{\mathcal{#1}}}

\newcommand{\lang}{\mathcal{L}}

\newcommand{\DBW}{\text{DBW}}
\newcommand{\DBWs}{\text{DBWs}}

\newcommand{\DPW}{\text{DPW}}

\newcommand{\DPWs}{\text{DPWs}}

\newcommand{\APW}{\text{APW}}

\newcommand{\ARW}{\text{ARW}}
\newcommand{\ARWs}{\text{ARWs}}

\newcommand{\DFA}{\text{DFA}}
\newcommand{\DFAs}{\text{DFAs}}

\newcommand{\R}{\ensuremath{\mathrel{\ltimes}}}
\newcommand{\notR}{\ensuremath{\mathrel{\centernot\ltimes}}}

\newcommand{\arena}{A}

\newcommand{\game}{\mathcal{G}}

\newcommand{\players}{\mathcal{P}}

\newcommand{\playercirc}{\text{\large$\circ$\normalsize}}
\newcommand{\playersquare}{\text{\scriptsize$\square$\normalsize}}

\newcommand{\plays}{\mathsf{Plays}}
\newcommand{\hist}{\mathsf{Hist}}

\newcommand{\infocc}{\mathsf{Inf}}
\newcommand{\finocc}{\mathsf{Fin}}

\newcommand{\bsigma}{\bar{\sigma}}

\newcommand{\outcome}[1]{\<#1 \>}
\newcommand{\outcomefrom}[2]{\outcome{#1}_{#2}}

\newcommand{\Rabin}[1]{\mathsf{Rabin}(#1)}

\newcommand{\prover}{\ensuremath{\mathbb{P}}}

\newcommand{\challenger}{\ensuremath{\mathbb{C}}}

\newcommand{\bigO}{\mathcal{O}}

\newcommand{\nl}{$\mathsf{NL}$}

\newcommand{\pspace}{$\mathsf{PSPACE}$}
\newcommand{\pspaceHard}{\pspace{}-hard}
\newcommand{\pspaceComplete}{\pspace{}-complete}

\newcommand{\exptime}{$\mathsf{EXPTIME}$}
\newcommand{\exptimeHard}{\exptime{}-hard}
\newcommand{\exptimeComplete}{\exptime{}-complete}

\usetikzlibrary{
  automata,
  arrows,
  positioning,
  shapes.geometric,
  calc
}
\tikzset{
edge with arrows/.style = {
    ->,
    >=stealth,
    shorten >=1pt,
},
directed/.style = {
    edge with arrows,
    node distance=2.3cm,
    on grid,
    semithick,
    double distance=1.5pt,
},
automaton/.style = {
    directed,
    auto,
    initial text={},
    pin distance = 1ex,
    every pin edge/.style = {
        draw=none
    },
    every state/.style={
      minimum size=1.5mm,
      inner sep=1pt,
    }
},
system/.style = {
    state,
    circle,
    minimum size=0mm,
    inner sep=5pt,
},
system2/.style={
    state,
    ellipse,
    minimum size=0mm,
    inner sep=2pt,
},
environment/.style = {
    state,
    rectangle,
    minimum size=0mm,
    inner sep=8pt,
},
environment2/.style = {
    state,
    diamond,
    minimum size=0mm,
    inner sep=4pt,
},
}

\begin{document}

\maketitle

\begin{abstract}
    This paper investigates the constrained existence problem for subgame perfect equilibria (SPEs) in multiplayer graph games. In the proposed framework, each player has a preference relation over the set of plays, assumed to be $\omega$-recognizable. Equivalently, he has a preference relation over a finite set of payoffs, and the set of plays with the same payoff is $\omega$-regular, for each payoff. This generic framework avoids the need to focus on specific payoff functions. We show that the constrained SPE existence problem is \exptimeComplete{}, as well as for Nash equilibria (NEs).
\end{abstract}

\newcommand{\pay}{\mathsf{pay}}
\newcommand{\outset}{P}
\newcommand{\outelement}{p}
\newcommand{\boutelement}{\bar{p}}
\newcommand{\doutelement}{d}
\newcommand{\G}{G}
\newcommand{\X}{RP}
\newcommand{\OK}{\ding{51}}
\newcommand{\KO}{\ding{55}}

\section{Introduction}
\label{section:intro}

\emph{Games on graphs} form an important mathematical framework for the verification and synthesis of reactive systems~\cite{BloemCJ18,Games-on-Graphs,lncs2500}. Classical zero-sum games model antagonism between two players, but analyzing multi-agent systems requires non-zero-sum games, where each agent acts rationally to pursue its own objective~\cite{BrenguierCHPRRS16,Bruyere17}. This situation is classically captured by solution concepts such as \emph{Nash equilibria} (NEs)~\cite{Nash50}. However, in sequential games, like in games on graphs, NEs are known to suffer from the issue of \emph{non-credible threats}. This limitation motivates the study of \emph{subgame perfect equilibrium} (SPE)~\cite{OsborneRubinstein1994}, a refinement of NE that requires the strategy profile to be an equilibrium in every subgame, thereby ensuring that players' threats and promises remain rational throughout the entire game, even after the deviation of some of the players.

In games played on graphs, objectives are usually classified into two main families~\cite{BloemCJ18}. First, with \emph{Boolean qualitative objectives} (e.g., parity objectives or, more broadly, $\omega$-regular objectives), infinite plays $x \in V^{\omega}$, where $V$ is the vertex set of the graph, are partitioned as winning or losing. Second, \emph{quantitative objectives} are defined through a payoff function (e.g., mean-payoff or discounted-sum) that assigns to each infinite play $v \in V^{\omega}$ a real-valued payoff.

To reduce the proliferation of algorithmic techniques specialized to particular objective classes, recent research has proposed \emph{unifying frameworks}, either by introducing relations that combine existing objectives~\cite{BouyerBMU15,Gutierrez-ijcai2017-nash-lexico} or by adopting automata-based preference relations~\cite{BruyereFiliotGrandmontRaskin26-arxiv,BruyereGrandmontRaskin25-mfcs}. The latter approach induces, for each player, a comparison between any two plays $x, y \in V^{\omega}$: an automaton reads the two sequences $x$ and $y$ synchronously and accepts the pair $(x, y)$ precisely when $y$ is preferred to $x$. Such $\omega$-automatic relations, as well as the strict subclass of $\omega$-recognizable relations are studied in detail in~\cite{Bergstrasser-Ganardi-2023,LodingSpinrath2019}.
In this paper, we continue the study of non-zero-sum multiplayer games, aiming for generic results in the setting of automata-based preference relations over the set of plays $V^{\omega}$. Unlike the previous works~\cite{BruyereFiliotGrandmontRaskin26-arxiv,BruyereGrandmontRaskin25-mfcs}, we consider the \emph{richer} notion of SPEs rather than NEs. As in those studies, however, we adopt a \emph{generic} setting in which players' objectives are given by preference relations that are $\omega$-recognizable. More precisely, for each player~$i$, we consider payoff functions $\pay_i: V^\omega \rightarrow P_i$ where $(i)$ $\outset_i$ is a finite set of payoffs, $(ii)$ for all $\outelement \in \outset_i$, the set of plays $X_\outelement$ with payoff $\outelement$ is $\omega$-regular, and $(iii)$ $\outset_i$ is equipped with a preference relation $\mathord{\R_i} \subseteq \outset_i \times \outset_i$. No condition is imposed on $\R_i$, such as being a total order. This setting is provably equivalent to one in which the induced preference relation over $V^{\omega}$ is $\omega$-recognizable. In practice, each set $X_\outelement$ is recognized independently by a deterministic parity automaton, and player preferences are modeled as an explicit relation over these finitely many sets. This framework naturally captures scenarios where players rank a finite set of distinct outcomes, accommodating multi-criteria trade-offs or generic dynamics, without requiring the preference structure to form an order.

\subparagraph*{Contributions}
We conduct a computational analysis of the \emph{constrained SPE existence problem} in this generic class of games. Our main result, \cref{theorem:main-result}, establishes the \exptimeComplete{}ness of this problem. To prove the \exptime{}-membership, we adapt a reduction to zero-sum two-player games that was previously investigated, in a restricted form, in Meunier's PhD thesis~\cite{meunier16}. The original construction assumed a strict total order on the values attained by the payoff functions and the use of deterministic Büchi automata, and no hardness argument was given. The framework addressed here necessitates a generalization of this construction and requires considering richer Emerson-Lei conditions in the associated zero-sum game. In this game, Prover plays out a candidate equilibrium move by move, while Challenger, acting on behalf of all the players of the original game, accepts Prover's move or forces a player deviation. Prover wins this game if and only if an SPE exists. For the \exptimeHard{}ness, we propose a reduction from the membership problem for linear-space alternating Turing machines, a canonical \exptimeComplete{} problem~\cite{ChandraKozenStockmeyer1981}. Given a machine $M$ and an input word $w$, we build a two-player game where Eve and Adam jointly simulate a run of $M$: on each round, Eve declares the current configuration, a player declares the next transition (Eve for existential states, Adam for universal states), and then both players get a chance to flag the other's previous declarations as invalid. An unflagged play reaching the accepting sink corresponds to an accepting run of $M$. Preferences are designed so that lying is never advantageous -- wrongly flagging a mistake is punished just as much as missing a real one -- so correct assessment of the alternating run is the only behavior compatible with equilibrium.
The constructed game has an SPE reaching the accepting outcome if and only if $M$ accepts $w$, giving \exptimeHard{}ness.
By adapting the proof given for SPEs, we also get that the constrained NE existence problem is \exptimeComplete{}. When the number of payoffs is fixed for each player, we keep \exptimeComplete{}ness for SPEs, while we get \pspaceComplete{}ness for NEs.

\subparagraph*{Related Work}
The formal analysis of non-zero-sum games on graphs has been the subject of sustained investigation, and an extensive body of work addresses NEs for both qualitative objectives~\cite{Bouyer-Brenguier-Markey-2010,BouyerBMU15,BrihayePS13,Gutierrez-MPRSW21,Gutierrez-ijcai2017-nash-lexico} and quantitative objectives~\cite{AlmagorKupfermanPerelli2018,UmmelsWojtczak2011-meanpayoff}. In parallel, the existence and computational complexity of SPEs have been studied across a range of frameworks, notably for $\omega$-regular objectives~\cite{BriceRB22-spe-parity,BrihayeBGRB19,Gradel-Ummels-08} as well as for quantitative objectives~\cite{BriceRB22-spe-meanpayoff,BrihayeBPG13,meunier16}.

More recent contributions have considered broader frameworks for these solution concepts~\cite{BouyerBMU15,KlimosLST12,LeRoux-Pauly-Equilibria,Rajasekaran-Bansal-Vardi-2023}. In particular, games equipped with $\omega$-automatic preference relations~\cite{BruyereFiliotGrandmontRaskin26-arxiv,BruyereGrandmontRaskin25-mfcs} provide an alternative viewpoint for the algorithmic analysis of zero-sum properties and NEs; however, the study of SPEs in such general settings has not yet been developed. The present work is most closely related to games with $\omega$-recognizable preference relations, for which NEs always exist when the relations are preorders~\cite{BruyereGrandmontRaskin25-mfcs}.
In that setting, an $\omega$-recognizable preorder induces a finite lattice of equivalence classes, interpretable as payoff values. Here, we instead adopt this graphical representation of payoffs as a starting point and allow general (unconstrained) preference relations, e.g., without imposing any underlying order structure.

The expressivity of automata to compare pairs of words is studied in~\cite{BansalChaudhuriVardi2022} for several classical objectives. Games with incomplete preferences over temporal objectives are also explored to synthesize NEs using preference automata in~\cite{DBLP:conf/aaai/Kulkarni0T25}. Note that the notion of preference between infinite plays can be viewed as a hyperproperty~\cite{DBLP:journals/jcs/ClarksonS10} relating pairs of traces, in the spirit of the temporal hyperproperty framework of~\cite{DBLP:journals/siglog/Finkbeiner23}. However, obtaining \exptime{}-optimal algorithms requires a direct, tailored treatment rather than an application of this general framework.

\subparagraph*{Structure of the Paper} \Cref{section:preliminaries} presents the required preliminaries and specifies the formal framework and problem studied in this work. \Cref{section:exptime-membership} provides a formal definition of the Prover--Challenger game, proves its correctness as an encoding of the constrained SPE existence problem, and describes an exponential-time decision procedure to solve it. \Cref{section:exptime-hardness} establishes \exptimeHard{}ness for the constrained SPE existence problem with a reduction from the membership problem for linear-space alternating Turing machines. The paper ends with a conclusion and future work.

\section{Preliminaries}
\label{section:preliminaries}

In this section, we recall the concepts of game on graph and of subgame perfect equilibrium, and state our main result.

\subparagraph*{Arenas and Games} An \emph{arena} is a tuple $\arena = (V,E,\players,(V_i)_{i\in\players})$ where $V$ is a finite set of vertices, $E \subseteq V \times V$ is a set of edges, $\players$ is a finite set of players, and $(V_i)_{i\in\players}$ is a partition of $V$, with $V_i$ the set of vertices owned by player~$i$. We assume, w.l.o.g., that each $v \in V$ has at least one successor, i.e., there exists $v' \in V$ such that $(v,v') \in E$. A \emph{play} $\pi\in V^\omega$ (resp.\ a \emph{history} $h\in V^*$) is an infinite (resp.\ finite) sequence of vertices $\pi_0\pi_1\dots$ such that $(\pi_k,\pi_{k+1})\in E$ for all $k$. The set of all plays is denoted $\plays$, and we write $\plays(v)$ for the set of plays starting with vertex $v$. Similarly, we use notation $\hist$ and $\hist(v)$ for the histories.

A \emph{payoff function} $\pay_i$ for player~$i$ is of the form $\pay_i: V^\omega \rightarrow \outset_i$ such that $\outset_i$ is a set equipped with a relation $\mathord{\R_i} \subseteq \outset_i \times \outset_i$, called \emph{preference relation} for player~$i$. We say that player~$i$ \emph{prefers} $\pi'$ to $\pi$ if $\pay_i(\pi) \R_i \pay_i(\pi')$. The set $V^\omega$ is partitioned into the subsets $\pay_i^{-1}(\outelement)$, $\outelement \in \outset_i$. A \emph{game} $\game = (\arena, (\pay_i)_{i\in\players})$ is composed of an arena $\arena$ and of payoff functions $\pay_i: V^\omega \rightarrow \outset_i$, $i \in \players$.

\begin{example} \label{ex:payoff}
    Among games that are classically studied, we have games using \emph{Boolean} payoff functions $\pay_i: V^\omega \rightarrow \outset_i$ such that $\outset_i = \{0,1\}$ and $0 \R_i 1$. The set $\pay_i^{-1}(1)$ is an $\omega$-regular set, which is the \emph{objective} of player~$i$~\cite{BloemCJ18}. More generally, each player can have an $m$-tuple of $\omega$-regular objectives and a payoff function $\pay_i: V^\omega \rightarrow \{0,1\}^m$ that indicates which objectives in the $m$-tuple are satisfied. The preference relation $\R_i$ on $\{0,1\}^m$ can be the lexicographic order, the counting order ($x \R_i y$ if and only if the number of $1$ in $x$ is less than the number of $1$ in $y$), see \cite[Section 2.5.2.]{BouyerBMU15} for additional examples. Other classical examples are \emph{quantitative} payoff functions $\pay_i: V^\omega \rightarrow \outset_i$ such that $\outset_i$ is the set $\Reals \cup \{\pm \infty\}$ equipped with $\R_i$ equal to the usual order on reals. These payoff functions are usually defined from a weight function $w_i: V \times V \rightarrow \Q$ (often restricted to $E$), and well-known examples are lim-inf, mean-payoff and discounted sum~\cite{BloemCJ18,lncs2500}.
\end{example}

\subparagraph*{Strategies, Nash and Subgame Perfect Equilibria} Let $\arena$ be an arena. A \emph{strategy} $\sigma_i:V^*V_i\rightarrow V$ for player~$i$ maps any history $hv \in V^*V_i$ to a successor of vertex $v$. It is \emph{positional} if $\sigma_i(hv) = \sigma_i(v)$ for all histories $h$. A play $\pi = \pi_0\pi_1 \dots$ is \emph{consistent} with a strategy $\sigma_i$ if $\pi_{k+1} = \sigma_i(\pi_0 \dots \pi_k)$ for all $k$ such that $\pi_k \in V_i$. Consistency is naturally extended to histories. A tuple of strategies $\bsigma = (\sigma_i)_{i\in\players}$ with $\sigma_i$ a strategy for each player~$i$ is called a \emph{strategy profile}. The play $\pi$ starting from an initial vertex $v$ and consistent with each $\sigma_i$ is denoted by $\outcomefrom{\bsigma}{v}$ and is called its \emph{outcome}.

Given a game $\game$ and a vertex $v \in V$, a \emph{Nash equilibrium} (NE) from $v$ is a strategy profile $\bsigma$ such that for all players $i$ and all their strategies $\tau_i$, we have $\pay_i(\outcomefrom{\bsigma}{v}) \not\R_i \pay_i(\outcomefrom{(\tau_i,\bsigma_{-i})}{v})$, where $\bsigma_{-i}$ denotes the strategy profile $(\sigma_j)_{j\in\players \setminus \{i\}}$ \cite{Nash50}. So, NEs are strategy profiles where no single player has an incentive to unilaterally deviate from his strategy. When there exists $\tau_i$ such that $\pay_i(\outcomefrom{\bsigma}{v}) \R_i \pay_i(\outcomefrom{(\tau_i,\bsigma_{-i})}{v})$, we say that $(\tau_i,\bsigma_{-i})$ is a \emph{profitable deviation} for player $i$.

We extend NEs to subgames as follows. Given a game $\game = (\arena, (\pay_i)_{i\in\players})$ and a history $h \in \hist$, we denote by $\rest{\game}{h} = (\arena,(\pay_{\|h,i})_{i \in \players})$ the \emph{subgame of $\game$ from $h$}, where for each player~$i$, we define $\pay_{\|h,i}$ by $\pay_{\|h,i}(\pi) = \pay_{i}(h \pi)$ for all $h\pi \in V^\omega$. Given a strategy $\sigma_i$ of player~$i$, we define the strategy $\sigma_{i\|h}$ in the subgame $\rest{\game}{h}$ from $\sigma_i$ by $\sigma_{i\|h}(h') = \sigma_i(hh')$ for all histories $hh' \in V^*V_i$. A \emph{subgame perfect equilibrium} (SPE) in $\game$ from an initial vertex $v_0$ is defined as a strategy profile $\bsigma$ such that for all histories $hv \in \hist(v_0)$, $\bsigma_{\|h}$ is an NE from $v$ in the subgame $\rest{\game}{h}$ \cite{OsborneRubinstein1994}.

\begin{example}
We illustrate both concepts of NE and SPE with a simple game $\game$ from~\cite{Bruyere21}, depicted in \cref{fig:NE}. There are four plays starting from $v_0$ and the payoffs $(\outelement_\playercirc,\outelement_\playersquare)$ of players~$\playercirc$ and~$\playersquare$ are indicated below each play. Their preference relations are the usual order on integers. We consider the strategy profile $\bsigma$ composed of the positional strategies: $v_1 \rightarrow v_4$, $v_2 \rightarrow v_5$ for player~$\playercirc$, and $v_0 \rightarrow v_2$ for player~$\playersquare$. It is an NE from $v_0$ with outcome $v_0v_2(v_5)^\omega$. For instance, if player $\playersquare$ uses the deviating strategy $v_0 \rightarrow v_1$, the resulting outcome is $v_0v_1(v_4)^\omega$ with lower payoff~$0$. However, $\bsigma$ is not an SPE: in the subgame with initial vertex~$v_1$, the outcome of $\bsigma$ is $v_1(v_4)^\omega$ and player~$\playercirc$ gets a higher payoff by deviating with $v_1 \rightarrow v_3$. However, the strategy profile $\bsigma'$ with $v_1 \rightarrow v_3$, $v_2 \rightarrow v_5$ for player~$\playercirc$ and $v_0 \rightarrow v_1$ for player~$\playersquare$ is both an NE and an SPE from $v_0$.
\end{example}

\subparagraph*{Studied Problems and Results} The \emph{constrained SPE (resp.\ NE) existence problem} asks, given a vertex $v$ and $\outelement_i \in \outset_i$ for each player $i \in \players$, whether there exists an SPE (resp.\ NE) $\bsigma$ from $v$ such that for all $i \in \players$, $\outelement_i = \pay_i(\outcomefrom{\bsigma}{v}) ~~\text{or}~~ \outelement_i \R_i \pay_i(\outcomefrom{\bsigma}{v})$. In this paper, we study this problem for games $\game = (\arena,(\pay_i)_{i\in\players})$ satisfying the following two \emph{conditions} for each player~$i$: $(i)$ the set $\outset_i$ is finite; $(ii)$ for each $\outelement \in \outset_i$, the set $\pay_i^{-1}(\outelement)$ is $\omega$-regular. No condition is imposed on $\R_i$. The relation $\R_i$ can be seen as a directed graph $\G_i$ whose set of vertices is $\outset_i$ and $(\outelement,\outelement')$ is an edge if and only if $\outelement \R_i \outelement'$. The games $\game$ satisfying both conditions are called \emph{regular-payoff} (\X{} games).

\begin{example}
    We come back to games with (tuples of) $\omega$-regular objectives as given in \Cref{ex:payoff}. In case of a single objective, those games satisfy the two conditions, and $\R_i$ is described by a graph $\G_i$ with the two vertices $0,1$ and the sole edge $(0,1)$. Games with lim-inf objectives also satisfy the conditions, as they use a set of payoffs in $\{0,1,\dots,p\}$, $p \in \mathbb{N}$, equipped with the usual order on integers~\cite{BansalChaudhuriVardi2022}. Another example of \X{} games are the games with $m$-tuples of ordered $\omega$-regular objectives in \cite[Section 2.5.2.]{BouyerBMU15}: several kinds of preference relations $\R_i$ are studied, that are all supposed to be a strict partial order (see Figure 6 in \cite{BouyerBMU15} for examples of graph $\G_i$).
\end{example}

In games with $\omega$-automatic preference relations introduced in \cite{BruyereGrandmontRaskin25-mfcs}, no payoff function $\pay_i$ is given, but rather a preference relation $R_i$ that directly compares plays: $x R_i y$ indicates that $y$ is preferred to $x$ by player~$i$. Each $R_i$ is supposed to be \emph{$\omega$-automatic}, that is, given by an automaton that reads pairs of vertices by advancing synchronously on the two plays $x,y$~\cite{Bergstrasser-Ganardi-2023,LodingSpinrath2019}. A strict subclass of $\omega$-automatic relations is the class of \emph{$\omega$-recognizable} relations, where $\mathord{R_i} = \bigcup_{i=1}^{k} X_i \times Y_i$, with $X_i, Y_i \subseteq V^\omega$ $\omega$-regular sets. \X{} games are exactly games with $\omega$-recognizable preferences relations. Indeed, given an $\omega$-automatic relation $R_i$, consider the equivalence relation $\equiv_i$ defined by: $x \equiv_i y$ if $\forall z \in V^\omega$, $x R_i z \Leftrightarrow y R_i z$ and $z R_i x \Leftrightarrow z R_i y$. Then, $R_i$ is $\omega$-recognizable if and only if $\equiv_i$ has finite index \cite{LodingSpinrath2019}. Each equivalence class of $\equiv_i$ is $\omega$-regular and corresponds to a payoff.

It is proved in \cite[Section 7]{BruyereGrandmontRaskin25-mfcs} that any \X{} game such that each $\R_i$ is a strict partial order\footnote{The equivalent hypothesis in \cite{BruyereGrandmontRaskin25-mfcs} is that the game has $\omega$-recognizable relations $R_i$ that are preorders. In this case, $x \equiv_i y$ if and only if $x R_i y$ and $y R_i x$.},
always has an NE. This existence result has its limitations: when no condition is imposed on the relations $\R_i$, there are simple \X{} games with no NE. Take the one-player game $\game = (\arena, \pay)$ whose arena is depicted in \cref{fig:game-no-ne}, $\pay: V^\omega \rightarrow \{a,b\}$ is such that $\pay^{-1}(x) = \{d x^\omega\}$, and the preference relation is the \emph{cyclic} relation $\R = \{(a,b), (b,a)\}$. There is no NE in $\game$ from $d$, as any play $d x^\omega$ has a preferable deviation $d y^\omega$ by choosing $y$ such that $x \R y$.

\begin{figure}[t]
    \centering
    \begin{minipage}[c]{0.32\textwidth}
        \centering
        \begin{tikzpicture}[->,>=stealth, shorten >=1pt, auto,
            state/.style={draw, minimum size=0.6cm, inner sep=0pt},scale=.7]
            \node[state, rectangle] (v0) at (0,2.4){$v_0$};
            \node[state, circle] (v1) at (-1,1.2){$v_1$};
            \node[state, circle] (v2) at (1,1.2){$v_2$};
            \node[state, circle] (v3) at (-1.8,0){$v_3$};
            \node[state, circle] (v4) at (-0.5,0){$v_4$};
            \node[state, circle] (v5) at (0.5,0){$v_5$};
            \node[state, circle] (v6) at (1.8,0){$v_6$};

            \draw[-stealth] (v0) edge (v1);
            \draw[-stealth] (v0) edge (v2);
            \draw[-stealth] (v1) edge (v3);
            \draw[-stealth] (v1) edge (v4);
            \draw[-stealth] (v2) edge (v5);
            \draw[-stealth] (v2) edge (v6);

            \draw[-stealth] (v3) edge [loop below] node {\footnotesize$(2,2)$} (v3);
            \draw[-stealth] (v4) edge [loop below] node {\footnotesize$(0,0)$} (v4);
            \draw[-stealth] (v5) edge [loop below] node {\footnotesize$(1,1)$} (v5);
            \draw[-stealth] (v6) edge [loop below] node {\footnotesize$(1,1)$} (v6);
        \end{tikzpicture}

        \captionof{figure}{A simple game with two NEs and one SPE \cite{Bruyere21}.}
        \label{fig:NE}
    \end{minipage}
    \hfill
    \begin{minipage}[c]{0.29\textwidth}
        \centering
        \begin{tikzpicture}[
            automaton,
            system/.append style={minimum size=0.6cm, inner sep=0pt},
            environment/.append style={minimum size=0.6cm, inner sep=0pt},scale=.8]
            \node[system] (v0) {$d$};
            \node[system,left=1.1cm of v0] (a) {$a$};
            \node[system,right=1.1cm of v0] (b) {$b$};

            \path[->]
                (a) edge[loop left] (a)
                (b) edge[loop right] (b)
                (v0) edge (a)
                (v0) edge (b);
        \end{tikzpicture}

        \captionof{figure}{A \X{} one-player game without NE from $d$, using a cyclic relation $\R = \{(a,b),(b,a)\}$.}
        \label{fig:game-no-ne}
    \end{minipage}
    \hfill
    \begin{minipage}[c]{0.36\textwidth}
        \centering
        \begin{tikzpicture}[
            automaton,
            system/.append style={minimum size=0.6cm, inner sep=0pt},
            environment/.append style={minimum size=0.6cm, inner sep=0pt},scale=.8]
            \node[system] (c) {$c$};
            \node[system,right=1cm of c] (a) {$a$};
            \node[environment,right=1cm of a] (b) {$b$};
            \node[environment,right=1cm of b] (d) {$d$};

            \path[->]
                (c) edge[loop left] (c)
                (a) edge[bend left=20] (b)
                (b) edge[bend left=20] (a)
                (d) edge[loop right] (d)
                (a) edge (c)
                (b) edge (d);
        \end{tikzpicture}

        \captionof{figure}{A \X{} two-player game without SPE from vertex $a$, using two total orders $\R_\playercirc$ and $\R_\playersquare$~\cite{Solan-Vieille-2003}.}
        \label{fig:game-no-spe}
    \end{minipage}
\end{figure}

The situation becomes more complex when we move to SPEs: the NE existence result of~\cite[Section 7]{BruyereGrandmontRaskin25-mfcs} is no longer valid for SPEs. Indeed, a two-player game with no SPE is described in~\cite{Solan-Vieille-2003}, where both players have three payoffs totally ordered. Its arena is depicted in \cref{fig:game-no-spe} and the players use the same payoff function $\pay: V^\omega \rightarrow \{c,d,\ell\}$ such that $\pay^{-1}(c)$ (resp.\ $\pay^{-1}(d)$, $\pay^{-1}(\ell)$) collects all plays with suffix $c^\omega$ (resp.\ $d^\omega$, $(ab)^\omega$). Player~\playercirc{} uses the total order $\ell \R_\playercirc c \R_\playercirc d$, while Player~\playersquare{} controls square vertices and uses the total order $c \R_\playersquare d \R_\playersquare \ell$.

Our main result is the following. Its proof is detailed in Sections~\ref{section:exptime-membership} and~\ref{section:exptime-hardness} below.

\begin{theorem}[restate=speexptimecompletetwoplayers,name=]
\label{theorem:main-result}
  The constrained SPE existence problem is \exptimeComplete{} for \X{} games, already for two players.
\end{theorem}

By adapting the proof, we get the same complexity result for NEs (see Appendix~\ref{app:exptime-ne-constrained}). When the number of payoffs is fixed for each player, the \exptimeHard{}ness remains for SPEs, while we get \pspaceComplete{}ness for NEs. The proofs are detailed in Appendices \ref{app:spe-exptime-hard-fixed-payoffs} and \ref{app:ne-recognizable-pspace}.

\begin{theorem}[restate=speexptimecompletethreepayoffs,name=]
\label{theorem:main-second-result}
  The constrained SPE existence problem is \exptimeComplete{} for \X{} games, already for three payoffs per player.
\end{theorem}

\section{EXPTIME-Membership}
\label{section:exptime-membership}

In this section, we use the concept of Prover game introduced in~\cite{meunier16} to prove the \exptime{}-membership of the constrained SPE existence problem for \X{} games. The Prover game in this thesis has never been published and was tailored to relations $\R_i$ that are total strict orders, where each set $\pay_i^{-1}(\outelement)$ is accepted by some deterministic Büchi automaton. We extend its usage by allowing for arbitrary relations $\R_i$ and deterministic parity automata (\DPWs{}).

\begin{proposition}[restate=exptimemembershipprover,name=]
\label{prop:membership}
    The constrained SPE existence problem for \X{} games is in \exptime{}.
\end{proposition}

Let $\game = (\arena, (\pay_i)_{i \in \players})$ be an \X{} game. Our goal is to construct a zero-sum Emerson-Lei game whose winning strategies correspond to SPEs in $\game$. This game is called the \emph{Prover Game} $\game_{\prover,\challenger}$ and is played between two players: \emph{Prover} $\prover$ and \emph{Challenger} $\challenger$. Prover's objective is to demonstrate that an SPE exists by showing that a specific tuple of payoffs is eventually reached by the play. To do that, \prover{} \emph{announces} these payoffs and provides the current SPE outcome vertex by vertex; \challenger{} acts as the collective environment of $n$ players, scrutinizing Prover's claims. Whenever \prover{} suggests a move, \challenger{} can either \emph{accept} it or \emph{deviate}, i.e., making the current vertex's owner deviate. In case of a deviation, \prover{} announces new payoffs where the deviating player does not have a payoff that is preferred to each of the previous payoffs announced consecutively. Prover wins when the play eventually stabilizes (finitely many deviations) and matches the announced payoffs. When multiple players deviate infinitely often, \prover{} also wins. And if player $i$ is the sole player deviating infinitely often, \prover{} wins if \challenger{} fails to achieve a payoff for player $i$ preferred to the payoffs announced along deviations. As $\game$ is an \X{} game, for each payoff $\outelement$ of $\outset_i$, $i \in \players$, we have a \DPW{} accepting $\pay_i^{-1}(\outelement)$. We add states of those automata as observers inside the states of the Prover game.

In the sequel, we denote by $n$ the number of players, by $\ell_i$ the size of each set $\outset_i$, and  by $\aut{A}^i_\outelement$ a \DPW{} that accepts $\pay_i^{-1}(\outelement)$, for all $\outelement \in \outset_i$ and $i \in \players$. The parity objective of $\aut{A}^i_\outelement$ can be encoded as an Emerson-Lei objective\footnote{An Emerson-Lei objective is described by a Boolean formula $\phi$ over variables $\infocc(T)$ and $\finocc(T)$ for some targets $T \subseteq V$, where $\infocc(T)$ (resp. $\finocc(T)$) is evaluated to true along a play $\pi$ if and only if $T$ is visited infinitely (resp.\ finitely) often by $\pi$~\cite{Emerson-Lei-Streett-ptime-complete}.} with Boolean formula denoted~$\varphi_\outelement^i$.

\subsection{Prover Game Definition}
The Prover Game $\game_{\prover,\challenger}$ is defined as follows. We first describe its state space. Let $\boutelement = (\outelement_1,\dots,\outelement_n) \in \outset_1 \times \dots \times \outset_n$ denote an $n$-tuple of payoffs, representing the payoffs currently announced by Prover for all players. Given $i \in \players$, let $D \subseteq \outset_i$ be the \emph{deviation history} that collects the previous payoffs announced consecutively for $i$, and that is initially equal to $\varnothing$. Let $\bar{q}^i = (q^i_1,\dots,q^i_{\ell_i})$ denote the current tuple of states for player $i$ such that $q^i_j$ is a state of $\aut{A}^i_{\outelement_j}$, for all $\outelement_j \in \outset_i$.

\subparagraph*{Arena}
The state space $S$ of $\game_{\prover,\challenger}$ consists of three distinct types of states $s \in S = S_\prover \cup S_\challenger$, where $S_P \cup S_D = S_\prover$ are the states of $\prover$ and $S_\challenger$ are the states of $\challenger$. Each state has components $\boutelement$, $D$, and $\bar{q}^i$, $i \in \players$, as described above, as well as additional components that we describe now.

\begin{itemize}
    \item \textbf{Prover States ($S_P \subseteq S_\prover$):} The states of $S_P$ are of the form $s = \left(v, \boutelement, i, D, \bar{q}^1, \dots, \bar{q}^n \right)$ where $v \in V$ is the current vertex of $\arena$, $i$ is the last player who deviated.

    \item \textbf{Challenger States ($S_\challenger$):} The states of $S_\challenger$ are of the form $s = \left((v,v'), \boutelement, i, D, \bar{q}^1, \dots, \bar{q}^n \right)$ where $(v,v')$ is an edge of $\arena$ suggested by \prover, and \challenger{} will either accept this move or deviate.

    \item \textbf{Deviation States ($S_D \subseteq S_\prover$):} The states of $S_D$ are of the form $s = \left(v, \boutelement, j, D, i, \bar{q}^1, \dots, \bar{q}^n \right)$ where \challenger, acting on behalf of player $j$, has rejected \prover's suggestion to move to $v'$ and deviated to vertex $v \neq v'$. The state indicates the active deviator $j$, and recalls the previous deviator $i$.
\end{itemize}

\noindent There are three types of transitions in the Prover game $\game_{\prover,\challenger}$.
\begin{itemize}
    \item \textbf{Suggestion ($S_P \to S_\challenger$):} From a Prover state $(v, \boutelement, i, D, \bar{q}^1, \dots, \bar{q}^n)$, \prover{} suggests an edge $(v,v') \in E$. The Prover game moves to Challenger state $((v,v'), \boutelement, i, D, \bar{q}^1, \dots, \bar{q}^n)$.

    \item \textbf{Response ($S_\challenger \to S_P$ or $S_\challenger \to S_D$):} From a Challenger state $((v,v'), \boutelement, i, D, \bar{q}^1, \dots, \bar{q}^n)$, \challenger{} evaluates the suggestion:
    \begin{itemize}
        \item \emph{Accept:} The edge $(v,v')$ is accepted. The game moves to $(v', \boutelement, i, D, \bar{q}'^1, \dots, \bar{q}'^n) \in S_P$.
        \item \emph{Deviate:} Let $j$ be such that $v \in V_j$. Challenger deviates to a vertex $u$  
        such that $(v,u) \in E$. The game moves to the Deviation state $(u, \boutelement, j, D, i, \bar{q}'^1, \dots, \bar{q}'^n)$, memorizing the new deviator ($j$) and the previous one ($i$).
    \end{itemize}
    In both cases, all \DPWs{} update to $\bar{q}'^i = \bar\delta_i(\bar q^i,v)$, for all $i \in \{1,\ldots,n\}$, where $\bar\delta_i$ denotes the tuple of transition functions for the \DPWs{} $\aut{A}^i_{\outelement}$, $\outelement \in \outset_i$, used by player $i$.

    \item \textbf{Deviation Resolution ($S_D \to S_\challenger$):} From a Deviation state $(v, \boutelement, j,D,i, \bar{q}^1, \dots, \bar{q}^n)$, \prover{} suggests an edge $(v,v') \in E$ (as from a Prover state). This action resolves the deviation by announcing a new payoff profile $\boutelement'$ and updating the deviation history to $D'$. The game moves to the Challenger state $((v,v'), \boutelement', j, D', \bar{q}^1, \dots, \bar{q}^n)$, such that
    \begin{itemize}
        \item If $j = i$ (the same player deviated again), the deviation history accumulates: $D' = D \cup \{\outelement_j\}$, where $\outelement_j$ is the $j$-component of $\boutelement$.
        \item If $j \neq i$ (a new player deviated), the deviation history resets: $D' = \{\outelement_j\}$.
    \end{itemize}
 Moreover, for the newly announced $\boutelement' = (\outelement'_1,\dots,\outelement'_n)$, we require that $d \notR_j \outelement'_j$, $\forall d \in D'$.
\end{itemize}

\subparagraph*{Prover's Objective}
To formally express the Emerson-Lei objective for Prover in $\game_{\prover,\challenger}$, we first define specific subsets of states to be tracked by the $\infocc$ and $\finocc$ operators during a play. Let $i,j$ be two players and $\outelement$ be a payoff in $\outset_1 \cup \dots \cup \outset_n$:
\begin{itemize}
    \item $S_D$: The entire set of Deviation states.
    \item $S_D(j,i)$: The set of Deviation states where player $j$ is the last deviator and player $i$ is the previous one (i.e., states of the form $(\dots, j, D, i,\dots)$ for some deviation history $D$).
    \item $S_{\text{ann}}(\outelement)$: The set of all states where $\outelement$ is a component of an announced payoff profile $\boutelement$ (i.e., states of the form $(\dots, \boutelement, \dots)$ for some $\boutelement$ and $i$ such that $\outelement = \outelement_i$).
    \item $S_{\text{dev}}(\outelement)$: The set of all states where $\outelement$ belongs to a deviation history $D$ (i.e., states of the form $(\dots, D, \dots)$ for some $D \ni \outelement$).
\end{itemize}
Note that the target sets of formulas $\varphi^i_\outelement$ used for each $\aut{A}^i_\outelement$ are reused, naturally extended to the Prover game state space (by abuse of notation, we keep the initial notation $\varphi^i_\outelement$).

The Emerson-Lei objective for $\prover$ in the Prover game $\game_{\prover,\challenger}$ is described by the Boolean formula $\varphi_\prover$, defined as the disjunction of three scenarios:

\begin{enumerate}
    \item \label{spe:prover-cond1} \textbf{Stable Outcome:} There are finitely many deviations, and the play eventually satisfies the formulas of the announced payoff profile:
    \[
    \finocc(S_D) ~\wedge~ \bigwedge_{i \in \players} \bigwedge_{\outelement \in \outset_i} \left( \infocc(S_{\text{ann}}(\outelement)) \Longrightarrow \varphi_\outelement^i \right)
    \]

    \item \label{spe:prover-cond3} \textbf{Deviator Instability:} Two or more players deviate infinitely often:
    \[
    \bigvee_{i,j \in \players, i \neq j} \infocc(S_D(j,i))
    \]

    \item \label{spe:prover-cond2} \textbf{Infinite Deviation:} Exactly one player $j$ deviates infinitely often, and their deviations do not yield a payoff preferred to what is stored in the deviation history $D$:
    \[
    \bigvee_{j\in\players} \left( \infocc(S_D(j,j)) \wedge \bigwedge_{d \in \outset_j} \left( \infocc(S_{\text{dev}}(d)) \Longrightarrow \bigvee_{\outelement' \in \outset_j, d \notR_j \outelement'} \varphi_{\outelement'}^j \right) \right)
    \]
\end{enumerate}

\subsection{Correspondence with SPEs}

We now can show that SPEs in $\game$ correspond to winning strategies for Prover in $\game_{\prover,\challenger}$.

\begin{proposition}[restate=specorrespondance,name=]
\label{prop:SPEcorrespondence}
    Given a vertex $v_0 \in V$ and a payoff profile $\boutelement_0 \in \outset_1 \times \dots \times \outset_n$, there exists an SPE $\bsigma$ whose outcome $\outcomefrom{\bsigma}{v_0}$ has payoff profile $\boutelement_0$ if and only if Prover has a winning strategy in the Prover game $\game_{\prover,\challenger}$ with Emerson-Lei objective $\varphi_\prover$ from the state $s_0 = (v_0,\boutelement_0,i,\varnothing,\bar q_0^1,\dots,\bar q_0^n)$ such that $i$ is any player and each $\bar q_0^j$ is the tuple of initial states of $\aut{A}_\outelement^j$, $\outelement \in \outset_j$.
\end{proposition}

\begin{proof}
    Suppose that there exists an SPE $\bsigma = (\sigma_i)_{i \in \players}$ in $\game$. We are going to construct a strategy $\tau_{\prover}$ for Prover that simulates $\bsigma$ and that will be proved to be winning in the Prover game. The strategy $\tau_{\prover}$ is defined as follows. Consider any history in $\game_{\prover,\challenger}$ starting from $s_0$ and ending in a state $s = (v,\dots) \in S_{\prover}$, and let $hv$ be its projection onto the set $V$ of $\game$ obtained by projecting onto their first component each of its states belonging to $S_\prover$. Then, Prover's strategy $\tau_{\prover}$ suggests the edge $(v, v') \in E$, as dictated by $\sigma_i(hv) = v'$ such that $v \in V_i$. Hence,
    \begin{itemize}
        \item if $s = (v, \boutelement, i, D, \bar{q}^1, \dots, \bar{q}^n) \in S_{P}$, then \prover{} moves to $s' = ((v,v'), \boutelement, i, D, \bar{q}^1, \dots, \bar{q}^n)$;
        \item if $s = (v, \boutelement, j, D, i, \bar{q}^1 \dots, \bar{q}^n) \in S_{D}$, then \prover{} moves to $s' = ((v,v'), \boutelement', j, D', \bar{q}^1, \dots, \bar{q}^n)$, where the newly announced $\boutelement' = (\outelement'_1,\dots,\outelement'_n)$ is the payoff profile of the outcome $h\outcomefrom{\bsigma_{\|h}}{v}$, that is, $\outelement'_i = \pay_i(h\outcomefrom{\bsigma_{\|h}}{v})$ for all $i \in \players$.
        Note that in this latter case, $\tau_\prover$ respects the transition constraint $d \notR_j \outelement'_j$ for all $d \in D'$ of the Prover game $\game_{\prover,\challenger}$. Indeed, recall that $D'$ accumulates the $j$-components of all announced payoff profiles since player $j$'s first deviation in the current uninterrupted sequence of deviations. By definition of $\tau_\prover$, these payoff profiles are constructed from $g\outcomefrom{\bsigma_{\|g}}{u}$, for some particular prefixes $gu$ of $h$. Because $\bsigma$ is an SPE, player $j$'s deviation to $v$ cannot yield an outcome $h\outcomefrom{\bsigma_{\|h}}{v}$ preferable (with respect to $\R_j$) to any such $g\outcomefrom{\bsigma_{\|g}}{u}$. Therefore, the move proposed by Prover respects the constraint $d \notR_j \outelement'_j$ for all $d \in D'$.
    \end{itemize}

    Let $\tau_{\challenger}$ be any Challenger strategy. We analyze the resulting play $\rho = \outcomefrom{(\tau_\prover,\tau_\challenger)}{s_0}$ and show that it satisfies $\varphi_\prover$, that is, one of Formulas \ref{spe:prover-cond1}--\ref{spe:prover-cond2} defined above:
    \begin{enumerate}
        \item \emph{Finitely many deviations:} Along $\rho$, \challenger{} eventually stops deviating, thus always accepting the move $(v,v')$ proposed by \prover{}. Therefore, $\rho$ has transitions that eventually always alternate between $S_P$ and $S_\challenger$, showing that $\finocc(S_D)$ holds. Moreover, the payoff profile announced in the states of $\rho$ stabilizes at some $\boutelement = (\outelement_1,\dots,\outelement_n)$. By construction of $\tau_\prover$, the projection of $\rho$ onto $V$ is equal to some play $\pi = h\outcomefrom{\bsigma_{\|h}}{v}$ of $\game$ with payoff profile equal to $\boutelement$. Hence, $\infocc(S_{\text{ann}}(\outelement))$ holds exactly for each $\outelement \in \{\outelement_1,\dots,\outelement_n\}$. As the payoff profile of $\pi$ equals $\boutelement$, each formula $\varphi_\outelement^i$ holds, for all $\outelement = \outelement_i\in \{\outelement_1,\dots,\outelement_n\}$. Consequently, Formula \ref{spe:prover-cond1} is satisfied.

        \item \emph{Multiple players deviate infinitely often:} Along $\rho$, at least two players $i,j$, $i \neq j$, infinitely often deviate. Hence, the play $\rho$ visits $S_D(j,i)$ infinitely often, trivially satisfying Formula \ref{spe:prover-cond3}.
        \item \emph{Exactly one player deviates infinitely often:} Eventually, \challenger{} makes exactly one player $j$ deviate infinitely often. Hence, $\rho$ visits infinitely often Deviation states all of the form $(\dots, j,D,j, \dots)$ where the deviation history increases until stabilizing at some $D$ (since $\outset_j$ is finite). Therefore, $\infocc(S_D(j,j))$ holds exactly for the pair $j,j$ and $\infocc(S_{\text{dev}}(d))$ holds exactly for all $d \in D$. Let us consider the payoff profile $\boutelement'$ of the projection $\pi'$ of $\rho$ onto $V$, and let $\outelement' \in \outset_j$ be the $j$-component of $\boutelement'$. We want to show that for all $d \in D$, $d \notR_j \outelement'$. In this way, Formula \ref{spe:prover-cond2} will be satisfied. Consider $d \in D$. By construction of $\tau_\prover$, there exists a prefix of $\rho$ whose projection onto $V$ is equal to $hv$ and the play $\pi = h\outcomefrom{\bsigma_{\|h}}{v}$ has a payoff profile whose $j$-component is equal to $d$. Moreover, in the subgame $\game_{\|h}$, $\pi'$ is a play starting at $v$ and deviating from $\pi$ by player~$j$. Since $\bsigma$ is an SPE, $\pi'$ cannot be preferred to $\pi$ for player~$j$, that is, $d \notR_j \outelement'$.

    \end{enumerate}

    Conversely, if Prover has a winning strategy $\tau_\prover$ from $s_0$ in $\game_{\prover,\challenger}$, we construct a strategy profile $\bsigma = (\sigma_i)_{i \in \players}$ from $v_0$ in $\game$ that mimics $\tau_\prover$, and we prove that it is an SPE in Appendix~\ref{app:ProverSecondPart}.
\end{proof}

Thanks to the previous proposition, the \exptime-membership of the (constrained) SPE existence problem reduces to solving the Prover game with Emerson-Lei objective $\varphi_\prover$. This is described in Appendix~\ref{app:prover-exptime-solve}. Note that the (constrained) SPE existence problem remains in \exptime{}, when the \DPWs{} $\aut{A}^i_{\outelement}$ are replaced by deterministic Emerson-Lei automata in the proof.

\section{EXPTIME-Hardness}
\label{section:exptime-hardness}

We prove the \exptime{}-hardness of the constrained SPE existence problem via a reduction from the membership problem for linear space alternating Turing machines, being \exptimeHard{}~\cite{ChandraKozenStockmeyer1981}.

\begin{proposition}[restate=constrainedspeexptimehard,name=]
\label{prop:SPE-exptime-hard-fixed-players}
    The constrained SPE existence problem for \X{} games is \exptimeHard{}, already for two players.
\end{proposition}

An \emph{alternating Turing machine} (ATM) is a tuple $M = \langle Q, q_0, g, \Sigma_i, \Sigma_t, \delta, F \rangle$ where $Q = \{q_0,\dots,q_r,q_{\text{acc}},q_{\text{rej}}\}$ is a finite set of states, where $q_0$ is the initial state, $q_{\text{acc}}$ is the accepting state, and $q_{\text{rej}}$ is the rejecting state; $Q_\exists,Q_\forall \subseteq \{q_0,\dots,q_r\}$ form a partition of $Q$ with existential ($\exists$) and universal ($\forall$) states; $\Sigma_i = \{a,b\}$ is an input alphabet and $\Sigma_t = \{a,b,\#\}$ is a tape alphabet where $\#$ is the blank symbol; and $\delta: Q \setminus \{q_{\text{acc}},q_{\text{rej}}\}  \times \Sigma_t \times Q \times \Sigma_t \times \{L,R\}$ is a transition relation where $L,R$ indicate the head movement (Left, Right) on the tape.

A configuration is a triple $(q,\ell,x)$ where $q$ is the current state, $\ell$ the tape head position, and $x \in (\Sigma_t)^{n}$ the $n$ symbols on the tape. A computation tree of an ATM $M$ on an input word $w \in \Sigma_i^*$ is a finite tree whose nodes are configurations, the root is the initial configuration $\mathfrak{c}_0 = (q_0,1,w)$, the leaves are configurations containing either $q_{\text{acc}}$ or $q_{\text{rej}}$, and the successor nodes of any configuration respect the transition relation. A word $w$ is accepted by $M$ if $\exists$ has a strategy that specifies, for each configuration containing an existential state, which successor to choose in the computation tree, in a way that all leaf nodes are accepting configurations. We say that $M$ is a \emph{linear space} ATM if the space used by $M$ on any $w \in \Sigma_i^*$ is bounded by $|w|$. The \emph{membership problem} is to decide whether a given word $w$ is accepted by a given linear space ATM $M$.

\subparagraph*{Reduction}

Given a linear space ATM $M$ and an input word $w = w_1 \ldots w_n$ of length~$n$, we construct an \X{} game $\game$ with two players to simulate the computation of $M$ on $w$. The set of players is $\players = \{\text{Eve}, \text{Adam}\}$. The game simulates the ATM through a sequence of three phases: $(i)$ a \emph{configuration phase}, where Eve declares the current configuration $\mathfrak{c} = (q,\ell,x)$, then $(ii)$ a \emph{transition phase}, where Eve (resp.\ Adam) chooses any tuple $\mathfrak{t} = (q',s',m) \in Q \times \Sigma_t \times \{L,R\}$ if $q \in Q_\exists$ (resp.\ $q \in Q_{\forall}$), acting as the transition declaration $(q,x_\ell,q',s',m)$, and $(iii)$ a \emph{referee phase}, where Eve (resp.\ Adam), observing Adam's choices (resp.\ Eve's choices), selects either to validate the configuration and transition phases, hence going back to phase $(i)$, or to flag the play as invalid due to a mistake, terminating the simulation. The referee phase ensures that the declared configurations are consistent with one another, starting from the initial configuration $\mathfrak{c}_0 = (q_0,1,w)$, and that each transition $(q,x_\ell,q',s',m)$ belongs to the transition relation $\delta$ of $M$. The simulation stops during a referee phase when a player flags the play as invalid, or when the chosen transition leads to $q_{\text{acc}}$ or $q_{\text{rej}}$ and Adam and Eve both validate the play. Payoffs and references relations will be designed so that lying is never advantageous -- wrongly flagging a mistake is punished just as much as missing a real one -- so correct assessment of the alternating run is the only behavior compatible with equilibrium.

\subparagraph*{Arena and Plays}
The arena of $\game$ is structured as a cycle starting at a vertex $v_{\text{go}}$, corresponding to the generation of successive configurations, transitions, and referee phases (see \cref{fig:exptime-hardness-arena}).

\begin{figure}
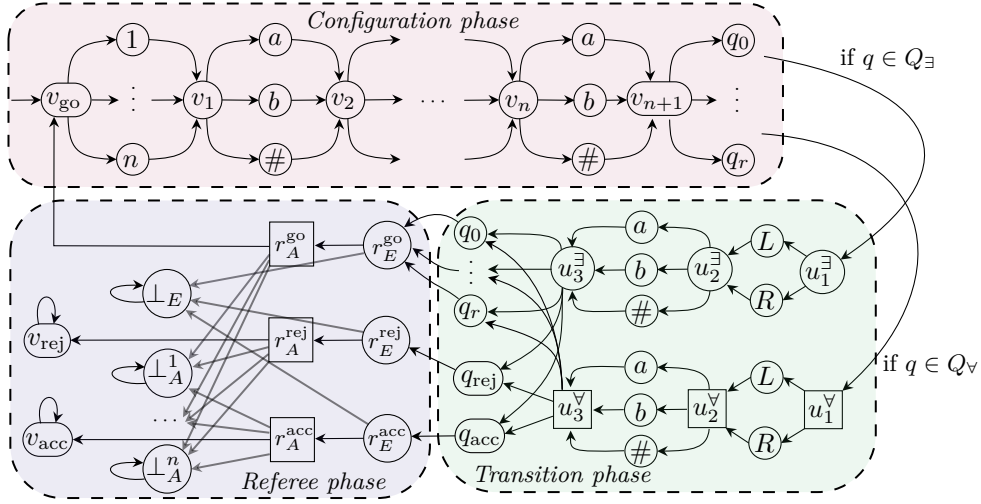

    \centering

    \caption{The game $\game$ used in the reduction, where circles are owned by Eve and squares by Adam.}
    \label{fig:exptime-hardness-arena}
\end{figure}

\begin{enumerate}
    \item \textbf{Configuration phase:} Eve declares a configuration by choosing edges from the vertices $v_{\text{go}}, v_1, v_2 \dots, v_{n+1}$. At vertex $v_{\text{go}}$, Eve moves to vertex $v_1$ by choosing an intermediate vertex $\ell \in \{1,\ldots,n\}$, indicating that the tape head should be at the $\ell$-th cell. From $v_k$, $k \in \{1,\ldots,n\}$, Eve moves to $v_{k+1}$ by choosing an intermediate vertex labeled with a tape symbol $s \in \Sigma_t$, indicating that the $k$-th tape cell currently contains $s$.
    From $v_{n+1}$, Eve chooses an intermediate vertex labeled $q \in Q \setminus \{q_{\text{acc}}, q_{\text{rej}}\}$, and then move to $u_1^{\exists}$ if $q \in Q_\exists$, and to $u_1^{\forall}$ if $q \in Q_\forall$.
    \item \textbf{Transition phase:} From $u_1^{\exists}$ (resp.\ $u_1^{\forall}$), Eve (resp.\ Adam) chooses a sequence $m \in \{L,R\}$, $s' \in \Sigma_t$, and $q' \in Q$, corresponding to the transition that should be taken, respectively the movement of the tape head, the symbol to write at the tape head position, and the new state of $M$. If $q' = q_{\text{acc}}$ (resp.\ $q' = q_{\text{rej}}$), the ATM goes to $r_{E}^{\text{acc}}$ (resp.\ $r_{E}^{\text{rej}}$), otherwise it goes to $r_{E}^{\text{go}}$.
    \item \textbf{Referee Phase:} From vertex $r_E^\alpha$, with $\alpha \in \{\text{acc},\text{rej},\text{go}\}$, Eve has two choices: either move to $r_A^{\alpha}$, or move to a sink vertex $\bot_{E}$, to flag the play as invalid due to Adam. From $r_{A}^{\alpha}$, Adam has $n+1$ choices: either move to $v_{\alpha}$, or move to a sink vertex $\bot_A^k$ for $k \in \{1,\ldots,n\}$, to flag the play as invalid due to Eve. The difference between each $\bot_A^k$ is discussed later.
\end{enumerate}
The arena of $\game$ has a size polynomial in the size $M$ and the length $n$ of the input word $w$.

A play $\pi$ in $\plays(v_{\text{go}})$ is \emph{valid} if it correctly simulates a \emph{path} of the computation tree on $w$, i.e.:
\begin{itemize}
\item In the first configuration phase, the declared configuration is the initial configuration $\mathfrak{c}_0$ of $M$;
\item For each declared configuration $\mathfrak{c} = (q,\ell,x)$ of $\pi$, the transition $\mathfrak{t} = (m,s',q')$ declared just after $\mathfrak{c}$ in the transition phase is an existing transition of $M$, i.e., $(q,x_\ell,q',s',m) \in \delta$;
\item After each such transition $\mathfrak{t}$ of $\pi$, in the referee phase, $\pi$ successively moves to $r_E^\alpha$, $r_A^\alpha$, and finally to $v_\alpha$, with $\alpha \in \{\text{acc},\text{rej},\text{go}\}$. In case $\pi$ moves to $v_{\text{go}}$, the next declared configuration $\mathfrak{c}'$ is the next configuration of the path of the computation tree, obtained by applying $\mathfrak{t}$ on $\mathfrak{c}$.
\end{itemize}
This concept is extended to histories: a history is $\emph{valid}$ if it is prefix of some valid play.

During the referee phase, Eve (resp.\ Adam) decides to validate or not the transition declaration of Adam (resp.\ the configuration and transition declarations of Eve) by either going to $v_\alpha$, with $\alpha \in \{\text{go, acc, rej}\}$ or going to $\bot_E$ (resp.\ $\bot_A^k$, $k \in \{1,\dots, n\}$). Therefore, we will define the notion of $A$-validity for plays that are seen valid for Eve. We will similarly define the notion of $(E,k)$-validity for Adam, a more complex notion as Adam has to check both configuration and transition declarations of Eve, locally at the $k$-th tape cell.
\begin{itemize}
    \item A play $\pi$ is \emph{$A$-valid} if for any declared configuration $\mathfrak{c} = (q,\ell,x)$ of $\pi$, if the transition $\mathfrak{t} = (m,s',q')$ just after $\mathfrak{c}$ was declared by Adam, then $(q,x_\ell,q',s',m) \in \delta$.

    \item Let $k \in \{1,\dots,k\}$. For any declared configuration $\mathfrak{c} = (q,\ell,x)$ of a play $\pi$, we focus on its \emph{$k$-projection} $(q,\ell, x_k)$ such that $x$ is projected on its $k$-th symbol. A play $\pi$ is \emph{$(E,k)$-valid} if:
    \begin{enumerate}
        \item\label{item:E-k-validity-transi} For any configuration $\mathfrak{c} = (q,\ell,x)$ of $\pi$ such that $\ell = k$, if the transition $\mathfrak{t} = (m,s',q')$ just after $\mathfrak{c}$ was declared by Eve, then $(q,x_k,q',s',m) \in \delta$.
        \item\label{item:E-k-validity-first-config} The $k$-projection of the first configuration of $\pi$ is equal to $(q_0,1,w_1)$ if $k = 1$, and to $(\cdot,\cdot,w_k)$ if $k \neq 1$.
        \item\label{item:E-k-validity-k-projection-config} For any two successive configurations $\mathfrak{c} = (q,\ell,x)$ and $\mathfrak{c}'$ of $\pi$ separated by the transition $\mathfrak{t} = (m,s',q')$, we define $(q',\ell',t)$ such that $\ell' = \begin{cases}
        \ell - 1 & \text{if } m = L,\\
        \ell + 1 & \text{if } m = R,
        \end{cases}$ and $t = \begin{cases}
        s' & \text{if } \ell = k,\\
        s  & \text{if } \ell \neq k.
        \end{cases}$ Then the $k$-projection of $\mathfrak{c}'$ is equal to $(q',\ell',t)$ if $\ell' = k$, and to $(\cdot,\cdot,t)$ if $\ell' \neq k$.
    \end{enumerate}
\end{itemize}

A history is $A$-valid (resp. $(E,k)$-valid) if it is prefix of an $A$-valid (resp. $(E,k)$-valid) play.
Note that a play or history that is $A$-valid and $(E,k)$-valid for all $k$, is valid. A play or history which is not valid (resp.\ not $A$-valid, not $(E,k)$-valid) is called \emph{invalid} (resp.\ \emph{$A$-invalid}, \emph{$(E,k)$-invalid}).

\subparagraph*{Preference relations}
We now want to show that $\game$ is an \X{} game. We first introduce the following partition of the set $\plays(v_{\text{go}})$ of $\game$:
\begin{itemize}
    \item The set $C_{\text{acc}}$ (resp.\ $C_{\text{rej}}$) of all plays eventually looping in $v_{\text{acc}}$ (resp.\ $v_{\text{rej}}$);
    \item The set $C_{\text{go}}$ of all plays visiting $v_{\text{go}}$ infinitely often;
    \item The set of all plays eventually looping in $\bot_{E}$ is partitioned into two subsets: the subset $C_{\bot_{E}}^{\text{\OK}}$ of $A$-invalid plays that correctly led Eve to move to $\bot_{E}$, and the subset $C_{\bot_{E}}^{\text{\KO}}$ of $A$-valid plays that incorrectly led Eve to move to $\bot_{E}$.
    \item Similarly, for each $k \in \{1, \dots, n\}$, the set of all plays eventually looping in $\bot_{A}^k$ is partitioned into the subset $C_{\bot_{A}^k}^{\text{\OK}}$ of $(E,k)$-invalid plays and the subset $C_{\bot_{A}^k}^{\text{\KO}}$ of $(E,k)$-valid plays.
\end{itemize}
We finally denote by $\overline{\plays(v_{\text{go}})}$ the set $V^\omega \setminus \plays(v_{\text{go}})$, where $V$ denotes the set of vertices of $\game$.

By notation abuse, we directly define the preference relations $\R_E$ and $\R_A$ over a partition of $V^\omega$ by using the previously defined sets. On the left side of \cref{fig:exptime-hard-relations}, $\R_E$ is a strict partial order given by its graph $\G_E$: an edge $X \rightarrow Y$ means that for all $(x,y) \in X \times Y$, $\pay_E(x) \R_E \pay_E(y)$ (edges obtained by transitivity are not indicated). On the right side of \cref{fig:exptime-hard-relations}, the graph $\G_A$ of the strict partial order $\R_A$ is depicted, such that $\pay_A(x) \R_A \pay_A(y)$ if and only if $\pay_E(y) \R_E \pay_E(x)$.

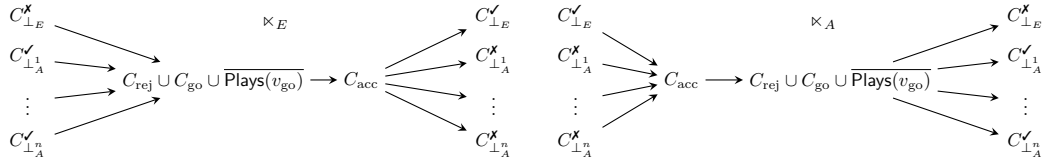
\begin{figure}[h]
    \centering
    \begin{minipage}[c]{0.49\textwidth}
    \begin{center}
    \begin{tikzpicture}[scale=0.7, transform shape, ->, >=stealth]
        \tikzset{io/.style={minimum width=1cm, minimum height=0.5cm}}

        \node at (4.67, 1.1) {$\R_E$};

        \node[io] (leftE) at (0, 1.2) {$C_{\bot_{E}}^{\text{\KO}}$};
        \node[io] (leftA1) at (0, 0.4) {$C_{\bot_{A}^1}^{\text{\OK}}$};
        \node[io] (leftDots) at (0, -0.4) {$\vdots$};
        \node[io] (leftAn) at (0, -1.2) {$C_{\bot_{A}^n}^{\text{\OK}}$};

        \node (crej) at (3.5, 0) {$C_{\text{rej}} \cup C_{\text{go}} \cup \overline{\plays(v_{\text{go}})}$};
        \node (cacc) at (6.3, 0) {$C_{\text{acc}}$};

        \node[io] (rightE) at (8.8, 1.2) {$C_{\bot_{E}}^{\text{\OK}}$};
        \node[io] (rightA1) at (8.8, 0.4) {$C_{\bot_{A}^1}^{\text{\KO}}$};
        \node[io] (rightDots) at (8.8, -0.4) {$\vdots$};
        \node[io] (rightAn) at (8.8, -1.2) {$C_{\bot_{A}^n}^{\text{\KO}}$};

        \draw (leftE) -- (crej);
        \draw (leftA1) -- (crej);
        \draw (leftDots) -- (crej);
        \draw (leftAn) -- (crej);

        \draw (crej) -- (cacc);

        \draw (cacc) -- (rightE);
        \draw (cacc) -- (rightA1);
        \draw (cacc) -- (rightDots);
        \draw (cacc) -- (rightAn);
    \end{tikzpicture}
    \end{center}
    \end{minipage}
    \hfill
    \begin{minipage}[c]{0.49\textwidth}
    \begin{center}
    \begin{tikzpicture}[scale=0.7, transform shape, ->, >=stealth]
        \tikzset{io/.style={minimum width=1cm, minimum height=0.5cm}}

        \node at (4.67, 1.1) {$\R_A$};

        \node[io] (leftE) at (0, 1.2) {$C_{\bot_{E}}^{\text{\OK}}$};
        \node[io] (leftA1) at (0, 0.4) {$C_{\bot_{A}^1}^{\text{\KO}}$};
        \node[io] (leftDots) at (0, -0.4) {$\vdots$};
        \node[io] (leftAn) at (0, -1.2) {$C_{\bot_{A}^n}^{\text{\KO}}$};

        \node (cacc) at (2, 0) {$C_{\text{acc}}$};
        \node (crej) at (5, 0) {$C_{\text{rej}} \cup C_{\text{go}} \cup \overline{\plays(v_{\text{go}})}$};

        \node[io] (rightE) at (8.5, 1.2) {$C_{\bot_{E}}^{\text{\KO}}$};
        \node[io] (rightA1) at (8.5, 0.4) {$C_{\bot_{A}^1}^{\text{\OK}}$};
        \node[io] (rightDots) at (8.5, -0.4) {$\vdots$};
        \node[io] (rightAn) at (8.5, -1.2) {$C_{\bot_{A}^n}^{\text{\OK}}$};

        \draw (leftE) -- (cacc);
        \draw (leftA1) -- (cacc);
        \draw (leftDots) -- (cacc);
        \draw (leftAn) -- (cacc);

        \draw (cacc) -- (crej);

        \draw (crej) -- (rightE);
        \draw (crej) -- (rightA1);
        \draw (crej) -- (rightDots);
        \draw (crej) -- (rightAn);
    \end{tikzpicture}
    \end{center}
    \end{minipage}
    \caption{The two relations $\R_E$ and $\R_A$ that are strict partial orders.}
    \label{fig:exptime-hard-relations}
\end{figure}

To prove that $\game$ is an \X{} game, it remains to show that the sets of the partition used by $\R_E$ and $\R_A$ are all $\omega$-regular. It is enough to provide an accepting deterministic Büchi automaton (\DBW{}) of polynomial size for each set $C_{\text{acc}}$, $C_{\text{rej}}$, $C_{\text{go}}$, $C_{\bot_{E}}^{\text{\OK}}$, $C_{\bot_{E}}^{\text{\KO}}$, $C_{\bot_{A}^k}^{\text{\OK}}$, $C_{\bot_{A}^k}^{\text{\KO}}$, and $\overline{\plays(v_{\text{go}})}$, as \DBWs{} are closed under union with a polynomial size construction. This is easily done for $C_{\text{acc}}$, $C_{\text{rej}}$, $C_{\text{go}}$, and $\overline{\plays(v_{\text{go}})}$. The automata for the remaining sets are detailed in Appendix~\ref{app:exptime-hard-dbw-construction}.

\subparagraph*{Correctness of the reduction}

We are now ready to prove the correctness of the reduction: the Turing machine $M$ accepts the input word $w$ if and only if there exists an SPE from $v_{\text{go}}$ in $\game$ satisfying the constraint $C_{\text{acc}}$ for both players. In view of the preferences relations $\R_E$ and $\R_A$ (see \Cref{fig:exptime-hard-relations}), this is equivalent to the existence of an SPE whose outcome \emph{belongs to $C_{\text{acc}}$}.

To this end, we need to define the notion of \emph{canonical strategy}. Consider the finite computation tree of $M$ on $w$. At each node of the tree, either $\exists$ wins, or $\forall$ wins; and $M$ accepts $w$ if and only if $\exists$ wins at its root. Therefore, there exist two positional strategies $\sigma_{\exists}$ and $\sigma_{\forall}$, winning from their respective winning nodes, and arbitrary from their respective loosing nodes. From $\sigma_{\exists}$, we define the following canonical strategy $\sigma_E^{\text{can}}$ in $\game$. Let $h \in \hist(v_{\text{go}})$. We have three cases:
\begin{itemize}
\item \emph{$h$ is valid}: If $h$ ends in $v_{\text{go}}$ (resp.\ $u_1^\exists$), then $\sigma_E^{\text{can}}$ mimics $\sigma_\exists$ from $h$ by declaring the expected configuration (resp.\ declaring the expected transition). If it ends in $r_E^{\alpha}$, then $\sigma_E^{\text{can}}(h) = r_A^{\alpha}$.
\item \emph{$h$ is invalid but $A$-valid}: $\sigma_E^{\text{can}}$ can no longer mimic $\sigma_{\exists}$ and thus declares an arbitrary configuration (resp.\ transition) if $h$ ends in $v_{\text{go}}$ (resp.\ $u_1^\exists$). If $h$ ends in $r_E^{\alpha}$, then $\sigma_E^{\text{can}}(h) = r_A^{\alpha}$.
\item \emph{$h$ is $A$-invalid}: we proceed as in the previous item, except that $\sigma_E^{\text{can}}(h) = \bot_E$ if $h$ ends in $r_E^{\alpha}$.
\end{itemize}
The strategy $\sigma_E^{\text{can}}$ is a natural retaliating strategy  in case Adam picks an edge leading to an invalid history. We similarly define the canonical strategy $\sigma_A^{\text{can}}$ from $\sigma_\forall$, with one small difference. Suppose that $h$ ends in $r_A^\alpha$. If $h$ is valid or $(E,k)$-valid for all $k$, then $\sigma_A^{\text{can}}(h) = v_\alpha$. Otherwise, consider the shortest prefix $h'$ of $h$ such that $h'$ is $(E,k)$-invalid for a unique $k$, then $\sigma_A^{\text{can}}(h) = \bot_A^k$.

\begin{proof}
Suppose that $M$ accepts $w$. We will prove that there exists an SPE in $\game$ from $v_{\text{go}}$ with outcome in $C_{\text{acc}}$. From the computation tree of $M$ on $w$, we consider the strategies $\sigma_\exists$ and $\sigma_\forall$ as described above, as well as their canonical strategies $\sigma_E^{\text{can}}$ and $\sigma_A^{\text{can}}$. Let $\bsigma = (\sigma_E^{\text{can}},\sigma_A^{\text{can}})$ and denote by $\rho$ the outcome $\outcomefrom{\bsigma}{v_{\text{go}}}$.
As $\bsigma$ produces a valid outcome by construction, and $\exists$ is winning from the root of the computation tree, we get that $\rho \in C_{\text{acc}}$. Let us show that $\bsigma$ is an SPE, that is, for each history $hv \in \hist(v_{\text{go}})$, $\bsigma_{\|h}$ is an NE from $v$ in the subgame $\rest{\game}{h}$. We can assume that $v$ is not a sink vertex, since otherwise any strategy profile is trivially an NE in this subgame.
\begin{itemize}
    \item If $hv$ is valid, $h \outcomefrom{\bsigma_{\|h}}{v}$ is valid by construction. Moreover, it belongs to $C_{\text{acc}}$ or $C_{\text{rej}}$ according to the winner in the computation tree at the node corresponding to the longest prefix of $h$ ending in $u_1^\exists$ or $u_1^\forall$. Consider a deviating strategy $\sigma'_E$ of Eve in $\rest{\game}{h}$. If $h \outcomefrom{(\sigma'_E,\sigma_{A|h}^{\text{can}})}{v}$ is valid, it is non-profitable for Eve with respect to $\R_E$, as the winner plays optimally. If it is invalid, then it is $(E,k)$-invalid for some $k$, and thus belongs to $C_{\bot_A^k}^{\text{\OK}}$ by definition of $\sigma_A^{\text{can}}$. It is again non-profitable for Eve. The same arguments can be repeated for a deviating strategy of Adam.
    \item If $hv$ is $A$-valid and $(E,k)$-invalid for some $k$, for the $k$ such that the shortest invalid prefix $h'$ of $h$ is $(E,k)$-invalid uniquely for $k$, then $h \outcomefrom{\bsigma_{\|h}}{v}$ belongs to $C_{\bot_A^k}^{\text{\OK}}$ by construction. So, Adam has a maximal payoff with respect to $\R_A$, and has no incentive to deviate. Consider a deviating strategy $\sigma'_E$ of Eve in $\rest{\game}{h}$. Her choices will have no impact, except from $r_E^\alpha$, with $\alpha \in \{\text{acc},\text{rej},\text{go}\}$. If she moves to $r_A^\alpha$, then the resulting outcome $h \outcomefrom{(\sigma'_E,\sigma_{A|h}^{\text{can}})}{v}$ belongs to $C_{\bot_A^k}^{\text{\OK}}$, and thus is non-profitable. If she moves to $\bot_E$, the resulting outcome belongs to $C_{\bot_E}^\text{\KO}$ as Adam played with valid choices, resulting again in a non-profitable deviation.
    \item If $hv$ is $A$-invalid, then $h \outcomefrom{\bsigma_{\|h}}{v}$ belongs to $C_{\bot_E}^{\text{\OK}}$ by construction, leading to a maximal payoff for Eve. All deviating strategies of Adam lead to $r_E^\alpha$, with $\alpha \in \{\text{acc},\text{rej},\text{go}\}$, from which Eve moves to $\bot_E$. The resulting outcome remains in $C_{\bot_E}^{\text{\OK}}$ and is non-profitable for Adam.
\end{itemize}

Conversely, assume that $w$ is not accepted by $M$ and there exists an SPE $\bsigma = (\sigma_E,\sigma_A)$ with outcome $\rho \in C_{\text{acc}}$. From the computation tree of $M$ on $w$, we extract a positional strategy $\tau_{\forall}$ of $\forall$, winning at its root, and construct its canonical strategy $\sigma_A^{\text{can}}$, as described above. We consider the subgame of $\game$ equal to itself, and show that $\bsigma$ is not an NE from $v_{\text{go}}$, which is a contradiction. First, if $\rho$ is invalid, from the shortest invalid prefix of $\rho$, Eve (or Adam) would have an incentive to move to $\bot_E$ (or some $\bot_A^k$) to get a profitable deviation in $C_{\bot_E}^{\text{\OK}}$ (or $C_{\bot_A^k}^{\text{\OK}}$). Second, if $\rho$ is valid, let us consider the deviating strategy $\sigma_A^{\text{can}}$ of Adam and the resulting play $\rho' = \outcomefrom{\sigma_E,\sigma_A^{\text{can}}}{v_{\text{go}}}$. If $\rho'$ is valid, then $\rho' \in C_{\text{rej}}$ by definition of $\sigma_A^{\text{can}}$, which is profitable for Adam (as $\rho \in C_{\text{acc}}$). If $\rho'$ is not valid, again by definition of $\sigma_A^{\text{can}}$, it is due to Eve and $\rho' \in C_{\bot_A^k}^{\text{\OK}}$ for some $k$, which is profitable for Adam. In both cases, we get the announced contradiction.
\end{proof}

\section{Conclusion}
\label{section:conclusion-future-work}

In this work, we have studied the contrained SPE existence problem for \X{} games, that is, for games where each player has a preference relation that is $\omega$-recognizable. This generic class of games encompasses several games with specific payoff functions and offers a general approach to solving equilibrium decision problems. Our main result states that this problem is \exptimeComplete{}, and that the hardness already holds for two players. The proof can be easily adapted to get the same complexity result for NEs. When the number of payoffs is fixed for each player, the contrained SPE existence problem remains \exptimeComplete{}, whereas it becomes \pspaceComplete{} for NEs. 

As future work, we will consider the extension of our analysis to games equipped with $\omega$-automatic preference relations. The decidability status of the constrained SPE existence problem in this broader framework is, to the best of our knowledge, open. We nevertheless conjecture that this problem becomes undecidable in the presence of $\omega$-automatic preference relations. This conjecture is informed by a recent undecidability result for $\omega$-automatic relations in the rational synthesis setting~\cite{BruyereFiliotGrandmontRaskin26-arxiv}. There, undecidability arises from the need to compare separated subtrees in trees that encode strategies. An analogous difficulty appears to occur for SPEs when preferences are specified by $\omega$-automatic relations. In contrast, for NEs, when $\omega$-automatic preferences are given by deterministic parity automata, the constrained NE existence problem is shown to be \pspaceComplete{} in the same paper~\cite{BruyereFiliotGrandmontRaskin26-arxiv}.

\bibliography{bibliography}

\newpage\appendix

\section{Proof of \texorpdfstring{\cref{prop:SPEcorrespondence}}{Proposition~\ref{prop:SPEcorrespondence}}}
\label{app:ProverSecondPart}

We here provide the proof of the second part of \cref{prop:SPEcorrespondence}.

\specorrespondance*

\begin{proof}[Proof of the second part of \cref{prop:SPEcorrespondence}]
    Suppose that Prover has a winning strategy $\tau_\prover$ from $s_0$ in $\game_{\prover,\challenger}$. From $\tau_\prover$, we construct a strategy profile $\bsigma = (\sigma_i)_{i \in \players}$ from $v_0$ in $\game$ as follows. Let $hv$ be any history in $\game$ with $v \in V_i$. There exists a unique history $g$ in $\game_{\prover,\challenger}$ ending in a state of $S_\prover$ and consistent with $\tau_\prover$ such that its projection onto $V$ is equal to $hv$. If $\tau_\prover(g) = ((v,v'),\dots)$, then we define $\sigma_i(hv)= v'$. Let us show that $\bsigma$ is an SPE in $\game$. Consider any history $hv$, the subgame $\game_{\|h}$, and the outcome $\pi = \outcomefrom{\bsigma_{\|h}}{v}$ starting at vertex $v$. Let $j$ be a player, $\sigma'_j$ be a deviating strategy, and $\pi' = \outcomefrom{(\bsigma_{-j\|h},\sigma'_j)}{v}$. Our goal is to show that $\pi'$ is not a profitable deviation for $j$.

    We denote by $\rho$ the unique play in the Prover game such that its projection onto $V$ is equal to $h\pi$. Along the suffix of $\rho$ that simulates $\pi$, \challenger{} always accepts the moves suggested by \prover, and the payoff profile announced by \prover{} is a fixed $\boutelement = (\outelement_1,\dots,\outelement_n)$. Hence, as $\tau_\prover$ is winning, Formula~\ref{spe:prover-cond1} is satisfied by $\rho$, meaning that the payoff profile of $h\pi$ is equal to $\boutelement$. Similarly to $\rho$, there exists a unique play $\rho'$ in $\game_{\prover,\challenger}$ whose projection onto $V$ is equal to $h\pi'$. Let $\tau_\challenger$ be the strategy of \challenger{} such that $\outcomefrom{\tau_\prover,\tau_{\challenger}}{s_0} = \rho'$. It simulates the strategy $\sigma'_j$ in the subgame $\game_{\|h}$. There are two cases:
    \begin{enumerate}
        \item \emph{Finitely many deviations by player $j$}: Along $\rho'$, as \challenger{} eventually stops deviating, we repeat for $\rho'$ the argument made above for $\rho$. By Formula \ref{spe:prover-cond1}, \challenger{} eventually accepts the payoff profile $\boutelement'$ announced by \prover{} and the play $h\pi'$ has payoff profile $\boutelement'$. Moreover, the deviation history eventually stabilizes at some $D$ that contains the component $\outelement_j$ of $\boutelement$. By construction of $\game_{\prover,\challenger}$, at the last update of the deviation history (leading to $D$), \prover{} is forced to announce a payoff $\outelement'_j$ in $\boutelement'$ such that $d \notR_j \outelement'_j$, $\forall d \in D$. In particular, $\outelement_j \notR_j \outelement'_j$. It follows that $\pay_j(h\pi) \notR_j \pay_j(h\pi')$.

        \item \emph{Infinitely many deviations by player $j$}: the play $\rho'$ infinitely often visits states of the form $(\dots,j,D,j,\dots)$ for some fixed deviation history $D$ that contains $\outelement_j$. Let $\boutelement' = (\outelement'_1,\dots,\outelement'_n)$ be the payoff profile of $h \pi'$. Then $\varphi_{\outelement'}^i$ holds for all $\outelement' = \outelement'_i$, $i \in \{1, \dots, n\}$. As $\tau_\prover$ is winning and $\rho'$ visits $S_D(j,j)$ infinitely often, Formula \ref{spe:prover-cond2} is satisfied. Thus, $\rho'$ satisfies $\bigvee_{\outelement' \in \outset_j, d \notR_j \outelement'} \varphi_{\outelement'}^j$ for all $d \in D$, and in particular, $\rho'$ satisfies $\outelement_j \notR_j \outelement'_j$ (as $D \ni \outelement_j$). Therefore, $\pay_j(h\pi) \notR_j \pay_j(h\pi')$.
    \end{enumerate}
    Consequently, $\bsigma$ is an SPE.
\end{proof}

\section{Proof of \texorpdfstring{\cref{prop:membership}}{Proposition \ref{prop:membership}}}
\label{app:prover-exptime-solve}

We here complete the proof of the \exptime-membership of the constrained SPE existence problem.

\exptimemembershipprover*

By \Cref{prop:SPEcorrespondence}, we have to solve the Prover game with Emerson-Lei objective $\varphi_\prover$. We will use the next two results.

\begin{lemma}\label{prop:ParityToBB}
    A parity objective with priorities in $\{0, \ldots, d\}$ (with $d$ even) can be encoded into a Emerson-Lei objective with $d+1$ targets and a Boolean formula of size $O(d^2)$.
\end{lemma}

\begin{theorem}[\cite{BruyereFRT24}] \label{thm:BBgames}
    Solving Emerson-Lei games is polynomial in the number $|V|$ of vertices, exponential in the number $k$ of targets, and linear in the size $|\phi|$ of the Boolean formula.
\end{theorem}

\begin{proof}[Proof of \cref{prop:membership}]
Given a vertex $v$ of $\game$ and payoffs $d_i \in \outset_i$ for each player $i \in \players$, we have to decide whether there exists an SPE $\bsigma$ such that for all $i \in \players$, $d_i = \pay_i(\outcomefrom{\bsigma}{v})$ or $d_i \R_i \pay_i(\outcomefrom{\bsigma}{v})$. By \Cref{prop:SPEcorrespondence}, this is equivalent to decide whether there exists some payoff profile $\boutelement$ such that $d_i = \outelement_i$ or $d_i \R_i \outelement_i$, for all $i \in \players$, for which Prover has a winning strategy in $\game_{\prover,\challenger}$ with objective $\varphi_\prover$ from the state $s_0 = (v,\boutelement,i,\varnothing,\bar q_0^1,\dots,\bar q_0^n)$ (with $i$ an arbitrary player).

Let us evaluate the complexity of this decision procedure, depending on the following parameters:
\begin{itemize}
    \item $|V|$ and $|E|$: respectively the number of vertices and the number of edges of $\game$,
    \item $n$: number of players,
    \item $\ell$: $\max\{ \ell_i \mid i \in \{1,\dots,n\}\}$,
    \item $Q$: $\max \{|Q^i_\outelement| \mid i \in \{1,\dots,n\}, \outelement\in \outset_i\}$ where $Q^i_\outelement$ is the set of states of $\aut{A}^i_\outelement$,
    \item $d$: $\max \{d^i_\outelement \mid i \in \{1,\dots,n\}, \outelement\in \outset_i\}$ where $\{0, \dots, d^i_\outelement\}$ is the set of priorities of $\aut{A}^i_\outelement$,
\end{itemize}

By \Cref{prop:ParityToBB}, the formulas $\varphi^i_\outelement$ all use a number of target sets in $\bigO(d)$ and their sizes $|\varphi^i_\outelement|$ are in $\bigO(d^2)$. In view of Theorem~\ref{thm:BBgames}, we need to carry out several computations.
\begin{itemize}
    \item \textbf{State space of $\game_{\prover,\challenger}$:} The state space $S$ consists of the union of $S_P$, $S_\challenger$, and $S_D$. A state is defined by a current vertex or edge, an announced payoff profile $\boutelement = (\outelement_1,\dots,\outelement_n)$, an active player $i$, a deviation history $D \subseteq \outset_i$, a newly deviating player, and $n$ tuples of automata states $\bar{q}^1, \dots, \bar{q}^n$. We get $|S|$ in $\bigO\left( |E| \cdot \ell^n \cdot n \cdot 2^{\ell} \cdot n \cdot Q^{\ell\cdot n} \right)$.

    \item \textbf{Number $k$ of target sets in $\varphi_\prover$:} The formula $\varphi_\prover$ tracks the target sets $S_D$, $S_D(j,i)$ for all $j,i \in \players$, $S_{\text{ann}}(\outelement)$ and $S_{\text{dev}}(\outelement)$ for all $\outelement \in \outset_i$, $i \in \players$, and the target sets embedded within each sub-formulas $\varphi^i_\outelement$. We get $k$ in $\bigO\left(1 + n^2 + 2 \cdot n \cdot \ell + n \cdot \ell \cdot d \right) = \bigO\left(n^2 + n \cdot \ell \cdot d \right)$.

    \item \textbf{Size of formula $\varphi_\prover$:} Recall that $\varphi_\prover$ is the disjunction of Formulas \ref{spe:prover-cond1}--\ref{spe:prover-cond2}. The size of Formula \ref{spe:prover-cond1} is dominated by the size of Formula \ref{spe:prover-cond2} whose size is in $\bigO(n \cdot (2 + \ell \cdot(2 + \ell \cdot d^2)) )$. By adding the size of Formula \ref{spe:prover-cond3}, we get $|\varphi_\prover|$ in $\bigO\left( n \cdot \ell^2 \cdot d^2 + n^2 \right)$.
\end{itemize}

By \Cref{thm:BBgames}, constructing and solving the Prover game with condition $\varphi_\prover$ is polynomial in the number $|S|$ of its states, exponential in the number $k$ of its target sets, and linear in the size $|\varphi_\prover|$, that is, polynomial in $|E|$, and exponential in $\log(Q)$, $n$, $\ell$, and $d$. Recall that we have to solve such a Prover game for all $\boutelement$ such that $d_i = \outelement_i$ or $d_i \R_i \outelement_i$ for each given $d_i$.
The previous complexity has to be multiplied by $\ell^n$, thus leading to an exponential total complexity.
\end{proof}

Note that the SPE existence problem in \X{} games is in \exptime{}, as a direct corollary of the previous proof.

\section{Construction of \DBWs{} in the Proof of \texorpdfstring{\cref{prop:SPE-exptime-hard-fixed-players}}{Proposition~\ref{prop:SPE-exptime-hard-fixed-players}}}
\label{app:exptime-hard-dbw-construction}

In this section, we describe the \DBWs{} accepting the set $C^{\text{\OK}}_{\bot_{E}}$, $C^{\text{\KO}}_{\bot_{E}}$, $C^{\text{\OK}}_{\bot_{A}^k}$, and $C^{\text{\KO}}_{\bot_{A}^k}$ respectively, as a part of the proof of \cref{prop:SPE-exptime-hard-fixed-players}.

\constrainedspeexptimehard*

Note that for these automata focussed on $A$-valid and $(E,k)$-valid plays, in addition to verifying that the declared transition exists, we must also verify that the tape head does not move to $0$ or $n+1$ (recall that we use a reduction from the membership problem for linear space ATMs). The check of the latter condition can be avoided by slightly modifying the arena of \Cref{fig:exptime-hardness-arena}: We extend the vertices of the configuration phase with the information $1$ (resp.\ $n$, $\mathsf{other}$) to remember that the declared position of the tape head is $1$ (resp.\ $n$, $\ell \in \{2,\dots,n-1\}$). In the transition phase, the allowed movements of the tape head from $u_1^\exists$ and $u_1^\forall$ are then uniquely $R$ (resp.\ uniquely $L$, both $R, L$). In the sequel, we suppose working with the modified arena.

We start by constructing the $\aut{A}^{\text{\OK}}_{\bot_{E}}$ accepting $C_{\bot_{E}}^{\text{\OK}}$, and then explain how to modify it to obtain the \DBW{} $\aut{A}^{\text{\KO}}_{\bot_{E}}$ for $C_{\bot_{E}}^{\text{\KO}}$. \Cref{fig:exptime-hard-dbw-E-OK} describes a part of this \DBW{} whose initial state is $q_\dagger$ and the unique final state $q_\sharp$ is double-circled. Two symbols on an edge is a shortcut to read the first symbol, reach a new intermediate state, and then read the second symbol. A star $*$ on a loop means that it loops for all remaining options to keep the automaton deterministic. For simplicity, we suppose that the words accepted by the automaton belong to $\plays(v_{\text{go}})$.\footnote{Simply take the intersection with the \DBW{} for $\plays(v_{\text{go}})$.}

The automaton $\aut{A}^{\text{\OK}}_{\bot_{E}}$ reads the current declared configuration $\mathfrak{c} = (q,\ell,x)$ until it stores its fragment $(q,\ell,s)$ such that $s = x_\ell$. Then it stores the information of the transition phase up to reading $u_3^\exists \cdot q'$ or $u_3^\forall \cdot q'$. If Adam was playing and the transition he declared does not exist, i.e., $u_3^\forall$ was read and $(q,s,q',s',m) \not\in \delta$, then the automaton enters a gadget by reading symbols until it loops in the final state while reading $\bot_E$.\footnote{\label{footnote:continue-playing}Eve invalidates the wrong actions of Adam by directly going to $\bot_E$, or continues playing before going to $\bot_E$.} Otherwise, it reads $r_E^\alpha$ and then $r_A^{\text{go}}$, before going back to the initial state and repeating the same scenario for the next declared configuration.

The automaton $\aut{A}^{\text{\KO}}_{\bot_{E}}$ is similar to $\aut{A}^{\text{\OK}}_{\bot_{E}}$ except that its final state is $q'_\sharp$ (instead of $q_\sharp$). Note that both automata have polynomial size.

\begin{figure}
    \centering
    \begin{tikzpicture}[x=0.75pt,y=0.75pt,yscale=-1]
        \draw (204.97,125.82) -- (128.5,125.82);
        \draw [shift={(207.97,125.82)}, rotate = 180,fill={rgb, 255:red, 0;green,0;blue,0},line width=0.08,draw opacity=0] (5.36,-2.57) -- (0,0) -- (5.36,2.57) -- (3.56,0) -- cycle;
        \draw   (300.51,125.49) .. controls (300.51,120.41) and (304.63,116.3) .. (309.7,116.3) -- (324.14,116.3) .. controls (329.22,116.3) and (333.33,120.41) .. (333.33,125.49) -- (333.33,125.49) .. controls (333.33,130.57) and (329.22,134.68) .. (324.14,134.68) -- (309.7,134.68) .. controls (304.63,134.68) and (300.51,130.57) .. (300.51,125.49) -- cycle ;
        \draw   (413.64,125.49) .. controls (413.64,120.41) and (417.75,116.3) .. (422.83,116.3) -- (450.62,116.3) .. controls (455.7,116.3) and (459.81,120.41) .. (459.81,125.49) -- (459.81,125.49) .. controls (459.81,130.57) and (455.7,134.68) .. (450.62,134.68) -- (422.83,134.68) .. controls (417.75,134.68) and (413.64,130.57) .. (413.64,125.49) -- cycle ;
        \draw (437.01,184.7) -- (437.15,135.16);
        \draw [shift={(437,187.7)}, rotate = 270.16,fill={rgb, 255:red, 0;green,0;blue,0},line width=0.08,draw opacity=0] (5.36,-2.57) -- (0,0) -- (5.36,2.57) -- (3.56,0) -- cycle;
        \draw (403.62,196.82) -- (353.11,196.82);
        \draw [shift={(350.11,196.82)}, rotate = 360,fill={rgb, 255:red, 0;green,0;blue,0},line width=0.08,draw opacity=0] (5.36,-2.57) -- (0,0) -- (5.36,2.57) -- (3.56,0) -- cycle;
        \draw   (403.62,196.82) .. controls (403.62,191.75) and (407.73,187.63) .. (412.81,187.63) -- (459.89,187.63) .. controls (464.97,187.63) and (469.09,191.75) .. (469.09,196.82) -- (469.09,196.82) .. controls (469.09,201.9) and (464.97,206.01) .. (459.89,206.01) -- (412.81,206.01) .. controls (407.73,206.01) and (403.62,201.9) .. (403.62,196.82) -- cycle ;
        \draw   (268.59,196.82) .. controls (268.59,191.75) and (272.7,187.63) .. (277.78,187.63) -- (341.67,187.63) .. controls (346.74,187.63) and (350.86,191.75) .. (350.86,196.82) -- (350.86,196.82) .. controls (350.86,201.9) and (346.74,206.01) .. (341.67,206.01) -- (277.78,206.01) .. controls (272.7,206.01) and (268.59,201.9) .. (268.59,196.82) -- cycle ;
        \draw (297,125.82) -- (226.53,125.82);
        \draw [shift={(300,125.82)}, rotate = 180,fill={rgb, 255:red, 0;green,0;blue,0},line width=0.08,draw opacity=0] (5.36,-2.57) -- (0,0) -- (5.36,2.57) -- (3.56,0) -- cycle;
        \draw (410.8,125.49) -- (333.33,125.49);
        \draw [shift={(413.8,125.49)}, rotate = 180,fill={rgb, 255:red, 0;green,0;blue,0},line width=0.08,draw opacity=0] (5.36,-2.57) -- (0,0) -- (5.36,2.57) -- (3.56,0) -- cycle;
        \draw (213.66,116.75) .. controls (207.11,94.3) and (227.43,92.79) .. (221.96,113.69);
        \draw [shift={(221.14,116.41)}, rotate = 288.7,fill={rgb, 255:red, 0;green,0;blue,0},line width=0.08,draw opacity=0] (5.36,-2.57) -- (0,0) -- (5.36,2.57) -- (3.56,0) -- cycle;
        \draw (107.51,126.14) -- (94.04,126.14);
        \draw [shift={(110.51,126.14)}, rotate = 180,fill={rgb, 255:red, 0;green,0;blue,0},line width=0.08,draw opacity=0] (5.36,-2.57) -- (0,0) -- (5.36,2.57) -- (3.56,0) -- cycle;
        \draw (312.99,116.75) .. controls (306.44,94.3) and (326.76,92.79) .. (321.29,113.69);
        \draw [shift={(320.47,116.41)}, rotate = 288.7,fill={rgb, 255:red, 0;green,0;blue,0},line width=0.08,draw opacity=0] (5.36,-2.57) -- (0,0) -- (5.36,2.57) -- (3.56,0) -- cycle;
        \draw (140.22,196.82) -- (119.5,196.7) -- (119.62,178.52);
        \draw [shift={(119.64,175.52)}, rotate = 90.39,fill={rgb, 255:red, 0;green,0;blue,0},line width=0.08,draw opacity=0] (5.36,-2.57) -- (0,0) -- (5.36,2.57) -- (3.56,0) -- cycle;
        \draw (268.28,196.82) -- (161.78,196.82);
        \draw [shift={(158.78,196.82)}, rotate = 360,fill={rgb, 255:red, 0;green,0;blue,0},line width=0.08,draw opacity=0] (5.36,-2.57) -- (0,0) -- (5.36,2.57) -- (3.56,0) -- cycle;
        \draw (119.64,156.96) -- (119.77,138.42);
        \draw [shift={(119.79,135.42)}, rotate = 90.39,fill={rgb, 255:red, 0;green,0;blue,0},line width=0.08,draw opacity=0] (5.36,-2.57) -- (0,0) -- (5.36,2.57) -- (3.56,0) -- cycle;
        \draw (110.95,166.16) -- (72.45,166.16);
        \draw [shift={(69.45,166.16)}, rotate = 360,fill={rgb, 255:red, 0;green,0;blue,0},line width=0.08,draw opacity=0] (5.36,-2.57) -- (0,0) -- (5.36,2.57) -- (3.56,0) -- cycle;
        \draw   (48.37,166.16) .. controls (48.37,160.34) and (53.09,155.62) .. (58.91,155.62) .. controls (64.73,155.62) and (69.45,160.34) .. (69.45,166.16) .. controls (69.45,171.98) and (64.73,176.69) .. (58.91,176.69) .. controls (53.09,176.69) and (48.37,171.98) .. (48.37,166.16) -- cycle ;
        \draw (54.32,156.75) .. controls (51.53,141.43) and (64.46,141.02) .. (62.45,153.6);
        \draw [shift={(61.81,156.41)}, rotate = 286.4,fill={rgb, 255:red, 0;green,0;blue,0},line width=0.08,draw opacity=0] (5.36,-2.57) -- (0,0) -- (5.36,2.57) -- (3.56,0) -- cycle;
        \draw   (207.97,125.82) .. controls (207.97,120.7) and (212.12,116.54) .. (217.25,116.54) .. controls (222.37,116.54) and (226.53,120.7) .. (226.53,125.82) .. controls (226.53,130.95) and (222.37,135.1) .. (217.25,135.1) .. controls (212.12,135.1) and (207.97,130.95) .. (207.97,125.82) -- cycle ;
        \draw   (110.51,126.14) .. controls (110.51,121.02) and (114.66,116.86) .. (119.79,116.86) .. controls (124.91,116.86) and (129.07,121.02) .. (129.07,126.14) .. controls (129.07,131.27) and (124.91,135.42) .. (119.79,135.42) .. controls (114.66,135.42) and (110.51,131.27) .. (110.51,126.14) -- cycle ;
        \draw   (140.22,196.82) .. controls (140.22,191.7) and (144.38,187.54) .. (149.5,187.54) .. controls (154.63,187.54) and (158.78,191.7) .. (158.78,196.82) .. controls (158.78,201.95) and (154.63,206.1) .. (149.5,206.1) .. controls (144.38,206.1) and (140.22,201.95) .. (140.22,196.82) -- cycle ;
        \draw   (110.36,166.24) .. controls (110.36,161.12) and (114.52,156.96) .. (119.64,156.96) .. controls (124.77,156.96) and (128.92,161.12) .. (128.92,166.24) .. controls (128.92,171.37) and (124.77,175.52) .. (119.64,175.52) .. controls (114.52,175.52) and (110.36,171.37) .. (110.36,166.24) -- cycle ;
        \draw (309.66,237.97) -- (309.79,206.42);
        \draw [shift={(309.64,240.96)}, rotate = 270.24,fill={rgb, 255:red, 0;green,0;blue,0},line width=0.08,draw opacity=0] (5.36,-2.57) -- (0,0) -- (5.36,2.57) -- (3.56,0) -- cycle;
        \draw   (369.85,250.69) .. controls (369.85,245.57) and (374.01,241.41) .. (379.13,241.41) .. controls (384.25,241.41) and (388.41,245.57) .. (388.41,250.69) .. controls (388.41,255.82) and (384.25,259.97) .. (379.13,259.97) .. controls (374.01,259.97) and (369.85,255.82) .. (369.85,250.69) -- cycle ;
        \draw   (367.18,250.69) .. controls (367.18,244.09) and (372.53,238.75) .. (379.13,238.75) .. controls (385.73,238.75) and (391.08,244.09) .. (391.08,250.69) .. controls (391.08,257.29) and (385.73,262.64) .. (379.13,262.64) .. controls (372.53,262.64) and (367.18,257.29) .. (367.18,250.69) -- cycle ;
        \draw (389.73,256.79) .. controls (415.56,262.3) and (415.48,244.46) .. (393.07,246.89);
        \draw [shift={(390.15,247.31)}, rotate = 350.19,fill={rgb, 255:red, 0;green,0;blue,0},line width=0.08,draw opacity=0] (5.36,-2.57) -- (0,0) -- (5.36,2.57) -- (3.56,0) -- cycle;
        \draw   (300.36,250.24) .. controls (300.36,245.12) and (304.52,240.96) .. (309.64,240.96) .. controls (314.77,240.96) and (318.92,245.12) .. (318.92,250.24) .. controls (318.92,255.37) and (314.77,259.52) .. (309.64,259.52) .. controls (304.52,259.52) and (300.36,255.37) .. (300.36,250.24) -- cycle ;
        \draw (364.42,250.24) -- (318.92,250.24);
        \draw [shift={(367.42,250.24)}, rotate = 180,fill={rgb, 255:red, 0;green,0;blue,0},line width=0.08,draw opacity=0] (5.36,-2.57) -- (0,0) -- (5.36,2.57) -- (3.56,0) -- cycle;
        \draw (301.33,254.37) .. controls (282.01,258.73) and (284.25,245.03) .. (297.45,247.05);
        \draw [shift={(300.36,247.71)}, rotate = 196.59,fill={rgb, 255:red, 0;green,0;blue,0},line width=0.08,draw opacity=0] (5.36,-2.57) -- (0,0) -- (5.36,2.57) -- (3.56,0) -- cycle;

        \draw (148.33,110.61) node [anchor=north west,inner sep=0.75pt,font=\small,align=left] {$v_{\text{go}} \cdot \ell $};
        \draw (212.67,120.27) node [anchor=north west,inner sep=0.75pt,font=\normalsize,align=left] {$\ell $};
        \draw (413.46,116.94) node [anchor=north west,inner sep=0.75pt,font=\normalsize,align=left] {$(q,\ell ,s)$};
        \draw (403.19,188.22) node [anchor=north west,inner sep=0.75pt,font=\normalsize,align=left] {$(q,\ell ,s,m)$};
        \draw (439.82,160.37) node [anchor=north west,inner sep=0.75pt,font=\small,align=left] {$u_{1}^{\forall } \cdot m$};
        \draw (268.79,188.27) node [anchor=north west,inner sep=0.75pt,font=\normalsize,align=left] {$(q,\ell ,s,m,s')$};
        \draw (361.22,177.7) node [anchor=north west,inner sep=0.75pt,font=\small,align=left] {$u_{2}^{\forall } \cdot s'$};
        \draw (439.86,142.33) node [anchor=north west,inner sep=0.75pt,font=\small,align=left] {$u_{1}^{\exists } \cdot m$};
        \draw (360.79,160.07) node [anchor=north west,inner sep=0.75pt,font=\small,align=left] {$u_{2}^{\exists } \cdot s'$};
        \draw (299.79,117.61) node [anchor=north west,inner sep=0.75pt,font=\normalsize,align=left] {$(\ell ,s)$};
        \draw (247,114.27) node [anchor=north west,inner sep=0.75pt,font=\small,align=left] {$v_{\ell } \cdot s$};
        \draw (193.64,70.33) node [anchor=north west,inner sep=0.75pt,font=\footnotesize,align=left] {\begin{minipage}[lt]{37.22pt}\setlength\topsep{0pt}
        \begin{center}
        $v_{k} \cdot t$\\for $k\ \neq \ell $
        \end{center}

        \end{minipage}};
        \draw (348.33,111.94) node [anchor=north west,inner sep=0.75pt,font=\small,align=left] {$v_{n+1} \cdot s$};
        \draw (292.97,70.33) node [anchor=north west,inner sep=0.75pt,font=\footnotesize,align=left] {\begin{minipage}[lt]{35.71pt}\setlength\topsep{0pt}
        \begin{center}
        $v_{k} \cdot t$\\for $k\leq n$
        \end{center}

        \end{minipage}};
        \draw (167.02,200.7) node [anchor=north west,inner sep=0.75pt,font=\footnotesize,align=left] {\begin{minipage}[lt]{68.85pt}\setlength\topsep{0pt}
        \begin{center}
        $u_{3}^{\forall } \cdot q'$ with\\$(q,s,q',s',m) \in \delta $
        \end{center}

        \end{minipage}};
        \draw (189.1,175.9) node [anchor=north west,inner sep=0.75pt,font=\small,align=left] {$u_{3}^{\exists } \cdot q'$ or};
        \draw (100.17,225.3) node [anchor=north west,inner sep=0.75pt,font=\small,align=left] {$ $};
        \draw (121.64,180.52) node [anchor=north west,inner sep=0.75pt,font=\small,align=left] {$r_{E}^{\alpha }$};
        \draw (121.43,139.93) node [anchor=north west,inner sep=0.75pt,font=\small,align=left] {$r_{A}^{\text{go}}$};
        \draw (80.67,150.45) node [anchor=north west,inner sep=0.75pt,font=\normalsize,align=left] {$\bot _{E}$};
        \draw (51.8,157.07) node [anchor=north west,inner sep=0.75pt,font=\normalsize,align=left] {$q'_{\sharp }$};
        \draw (49.67,127.78) node [anchor=north west,inner sep=0.75pt,font=\normalsize,align=left] {$\bot _{E}$};
        \draw (306.5,256.97) node [anchor=north west,inner sep=0.75pt,font=\small,align=left] {$ $};
        \draw (313.02,212.27) node [anchor=north west,inner sep=0.75pt,font=\footnotesize,align=left] {$u_{3}^{\forall } \cdot q'$ with $(q,s,q',s',m) \notin \delta $};
        \draw (372.13,245.73) node [anchor=north west,inner sep=0.75pt,font=\normalsize,align=left] {$q_{\sharp }$};
        \draw (409.33,244.12) node [anchor=north west,inner sep=0.75pt,font=\normalsize,align=left] {$\bot _{E}$};
        \draw (277.55,246.97) node [anchor=north west,inner sep=0.75pt,font=\small,align=left] {$*$};
        \draw (330,235.12) node [anchor=north west,inner sep=0.75pt,font=\normalsize,align=left] {$\bot _{E}$};
        \draw (113.53,120.93) node [anchor=north west,inner sep=0.75pt,font=\normalsize,align=left] {$q_{\dagger }$};
    \end{tikzpicture}
    \caption{The \DBW{} accepting $C_{\bot_E}^{\text{\OK}}$.}
    \label{fig:exptime-hard-dbw-E-OK}
\end{figure}

\begin{figure}
    \centering
    \begin{tikzpicture}[x=0.75pt,y=0.75pt,yscale=-1]
        \draw   (169.45,128.82) .. controls (169.45,123.75) and (173.56,119.63) .. (178.64,119.63) -- (206.43,119.63) .. controls (211.51,119.63) and (215.62,123.75) .. (215.62,128.82) -- (215.62,128.82) .. controls (215.62,133.9) and (211.51,138.01) .. (206.43,138.01) -- (178.64,138.01) .. controls (173.56,138.01) and (169.45,133.9) .. (169.45,128.82) -- cycle ;
        \draw (187.32,119.97) .. controls (181.25,93.62) and (201.77,90.57) .. (195.47,117.09);
        \draw [shift={(194.81,119.63)}, rotate = 285.77,fill={rgb, 255:red, 0;green,0;blue,0},line width=0.08,draw opacity=0] (5.36,-2.57) -- (0,0) -- (5.36,2.57) -- (3.56,0) -- cycle;
        \draw (270.09,128.82) -- (215.62,128.82);
        \draw [shift={(273.09,128.82)}, rotate = 180,fill={rgb, 255:red, 0;green,0;blue,0},line width=0.08,draw opacity=0] (5.36,-2.57) -- (0,0) -- (5.36,2.57) -- (3.56,0) -- cycle;
        \draw (384.06,128.82) -- (338.56,128.82);
        \draw [shift={(387.06,128.82)}, rotate = 180,fill={rgb, 255:red, 0;green,0;blue,0},line width=0.08,draw opacity=0] (5.36,-2.57) -- (0,0) -- (5.36,2.57) -- (3.56,0) -- cycle;
        \draw (510.23,128.82) -- (469.33,128.82);
        \draw [shift={(513.23,128.82)}, rotate = 180,fill={rgb, 255:red, 0;green,0;blue,0},line width=0.08,draw opacity=0] (5.36,-2.57) -- (0,0) -- (5.36,2.57) -- (3.56,0) -- cycle;
        \draw   (273.09,128.82) .. controls (273.09,123.75) and (277.21,119.63) .. (282.28,119.63) -- (329.37,119.63) .. controls (334.45,119.63) and (338.56,123.75) .. (338.56,128.82) -- (338.56,128.82) .. controls (338.56,133.9) and (334.45,138.01) .. (329.37,138.01) -- (282.28,138.01) .. controls (277.21,138.01) and (273.09,133.9) .. (273.09,128.82) -- cycle ;
        \draw   (387.06,128.82) .. controls (387.06,123.75) and (391.18,119.63) .. (396.26,119.63) -- (460.14,119.63) .. controls (465.22,119.63) and (469.33,123.75) .. (469.33,128.82) -- (469.33,128.82) .. controls (469.33,133.9) and (465.22,138.01) .. (460.14,138.01) -- (396.26,138.01) .. controls (391.18,138.01) and (387.06,133.9) .. (387.06,128.82) -- cycle ;
        \draw   (513.23,128.82) .. controls (513.23,123.75) and (517.34,119.63) .. (522.42,119.63) -- (560.45,119.63) .. controls (565.53,119.63) and (569.64,123.75) .. (569.64,128.82) -- (569.64,128.82) .. controls (569.64,133.9) and (565.53,138.01) .. (560.45,138.01) -- (522.42,138.01) .. controls (517.34,138.01) and (513.23,133.9) .. (513.23,128.82) -- cycle ;
        \draw   (80.52,128.43) .. controls (80.52,123.3) and (84.67,119.15) .. (89.8,119.15) .. controls (94.92,119.15) and (99.08,123.3) .. (99.08,128.43) .. controls (99.08,133.55) and (94.92,137.71) .. (89.8,137.71) .. controls (84.67,137.71) and (80.52,133.55) .. (80.52,128.43) -- cycle ;
        \draw   (183.07,199.74) .. controls (183.07,194.62) and (187.23,190.46) .. (192.35,190.46) .. controls (197.47,190.46) and (201.63,194.62) .. (201.63,199.74) .. controls (201.63,204.87) and (197.47,209.02) .. (192.35,209.02) .. controls (187.23,209.02) and (183.07,204.87) .. (183.07,199.74) -- cycle ;
        \draw (193.29,187.41) -- (193.45,137.95);
        \draw [shift={(193.28,190.41)}, rotate = 270.18,fill={rgb, 255:red, 0;green,0;blue,0},line width=0.08,draw opacity=0] (5.36,-2.57) -- (0,0) -- (5.36,2.57) -- (3.56,0) -- cycle;
        \draw (429.92,138.6) -- (429.92,199.8) -- (204.63,199.74);
        \draw [shift={(201.63,199.74)}, rotate = 0.01,fill={rgb, 255:red, 0;green,0;blue,0},line width=0.08,draw opacity=0] (5.36,-2.57) -- (0,0) -- (5.36,2.57) -- (3.56,0) -- cycle;
        \draw (102.46,128.81) -- (169.45,128.82);
        \draw [shift={(99.46,128.81)}, rotate = 0.01,fill={rgb, 255:red, 0;green,0;blue,0},line width=0.08,draw opacity=0] (5.36,-2.57) -- (0,0) -- (5.36,2.57) -- (3.56,0) -- cycle;
        \draw   (111.52,199.83) .. controls (111.52,194.7) and (115.67,190.55) .. (120.8,190.55) .. controls (125.92,190.55) and (130.08,194.7) .. (130.08,199.83) .. controls (130.08,204.95) and (125.92,209.11) .. (120.8,209.11) .. controls (115.67,209.11) and (111.52,204.95) .. (111.52,199.83) -- cycle ;
        \draw   (108.85,199.83) .. controls (108.85,193.23) and (114.2,187.88) .. (120.8,187.88) .. controls (127.39,187.88) and (132.74,193.23) .. (132.74,199.83) .. controls (132.74,206.42) and (127.39,211.77) .. (120.8,211.77) .. controls (114.2,211.77) and (108.85,206.42) .. (108.85,199.83) -- cycle ;
        \draw (95.14,117.29) .. controls (96.31,107.22) and (93.79,101.91) .. (90.89,100.96) .. controls (86.93,99.66) and (82.25,106.45) .. (85.19,120.35);
        \draw [shift={(94.68,120.42)}, rotate = 279.94,fill={rgb, 255:red, 0;green,0;blue,0},line width=0.08,draw opacity=0] (5.36,-2.57) -- (0,0) -- (5.36,2.57) -- (3.56,0) -- cycle;
        \draw (136.08,199.73) -- (183.07,199.74);
        \draw [shift={(133.08,199.73)}, rotate = 0.02,fill={rgb, 255:red, 0;green,0;blue,0},line width=0.08,draw opacity=0] (5.36,-2.57) -- (0,0) -- (5.36,2.57) -- (3.56,0) -- cycle;
        \draw (188.11,208.12) .. controls (182.55,228.15) and (199.96,228.79) .. (198.06,211.18);
        \draw [shift={(197.6,208.25)}, rotate = 78.28,fill={rgb, 255:red, 0;green,0;blue,0},line width=0.08,draw opacity=0] (5.36,-2.57) -- (0,0) -- (5.36,2.57) -- (3.56,0) -- cycle;
        \draw (125.94,186.49) .. controls (127.11,176.42) and (124.59,171.11) .. (121.69,170.16) .. controls (117.73,168.86) and (113.05,175.65) .. (115.99,189.55);
        \draw [shift={(125.48,189.62)}, rotate = 279.94,fill={rgb, 255:red, 0;green,0;blue,0},line width=0.08,draw opacity=0] (5.36,-2.57) -- (0,0) -- (5.36,2.57) -- (3.56,0) -- cycle;

        \draw (169.27,120.27) node [anchor=north west,inner sep=0.75pt,font=\normalsize,align=left] {$(q,\ell ,s)$};
        \draw (272.67,120.22) node [anchor=north west,inner sep=0.75pt,font=\normalsize,align=left] {$(q,\ell ,s,m)$};
        \draw (223.9,113.03) node [anchor=north west,inner sep=0.75pt,font=\small,align=left] {$u_{1}^{\forall } \cdot m$};
        \draw (387.27,120.27) node [anchor=north west,inner sep=0.75pt,font=\normalsize,align=left] {$(q,\ell ,s,m,s')$};
        \draw (345.7,111.03) node [anchor=north west,inner sep=0.75pt,font=\small,align=left] {$u_{2}^{\forall } \cdot s'$};
        \draw (188.3,85.83) node [anchor=north west,inner sep=0.75pt,font=\small,align=left] {$*$};
        \draw (514.87,120.47) node [anchor=north west,inner sep=0.75pt,font=\normalsize,align=left] {$(q',\ell ',t)$};
        \draw (223.67,95) node [anchor=north west,inner sep=0.75pt,font=\small,align=left] {$u_{1}^{\exists } \cdot m$};
        \draw (345.6,93.4) node [anchor=north west,inner sep=0.75pt,font=\small,align=left] {$u_{2}^{\exists } \cdot s'$};
        \draw (473.1,110.63) node [anchor=north west,inner sep=0.75pt,font=\small,align=left] {$u_{3}^{\forall } \cdot q'$};
        \draw (82.8,118.62) node [anchor=north west,inner sep=0.75pt,font=\normalsize,align=left] {$q'_{\sharp }$};
        \draw (197.7,141.63) node [anchor=north west,inner sep=0.75pt,font=\small,align=left] {$v_{\text{go}} \cdot \ell '$ for $\ell '\neq \ell $ (if $\ell =k$)};
        \draw (197.5,157.23) node [anchor=north west,inner sep=0.75pt,font=\small,align=left] {$v_{k} \cdot s'$ for $s'\neq s$};
        \draw (197.9,172.83) node [anchor=north west,inner sep=0.75pt,font=\small,align=left] {$v_{n+1} \cdot q'$ for $q'\neq q$ (if $\ell =k$)};
        \draw (433.9,156.03) node [anchor=north west,inner sep=0.75pt,font=\small,align=left] {$u_{3}^{\exists } \cdot q'$ such that $\ell =k$\\and $(q,s,q',s',m) \notin \delta $};
        \draw (103.6,109.72) node [anchor=north west,inner sep=0.75pt,font=\footnotesize,align=left] {$r_{E}^{\alpha } \cdot r_{A}^{\alpha } \cdot \bot _{A}^{k}$};
        \draw (81.9,83.88) node [anchor=north west,inner sep=0.75pt,font=\small,align=left] {$\bot _{A}^{k}$};
        \draw (113.8,194.02) node [anchor=north west,inner sep=0.75pt,font=\normalsize,align=left] {$q_{\sharp }$};
        \draw (429.03,75.23) node [anchor=north west,inner sep=0.75pt,font=\small,align=left] {$u_{3}^{\exists } \cdot q'$ such that $\ell \neq k$\\or $(q,s,q',s',m) \in \delta $};
        \draw (150.33,178.33) node [anchor=north west,inner sep=0.75pt,font=\normalsize,align=left] {$\bot _{A}^{k}$};
        \draw (111.7,152.08) node [anchor=north west,inner sep=0.75pt,font=\small,align=left] {$\bot _{A}^{k}$};
        \draw (188.9,223.83) node [anchor=north west,inner sep=0.75pt,font=\small,align=left] {$*$};
    \end{tikzpicture}
    \caption{The \DBW{} accepting $C_{\bot_A^k}^{\text{\OK}}$, where $\ell' =
\begin{cases}
\ell + 1 & \text{if } m = R,\\
\ell - 1 & \text{if } m = L,
\end{cases}$, $t =
\begin{cases}
s' & \text{if } \ell = k,\\
s  & \text{if } \ell \neq k.
\end{cases}$}
    \label{fig:exptime-hard-dbw-A_k-OK}
\end{figure}
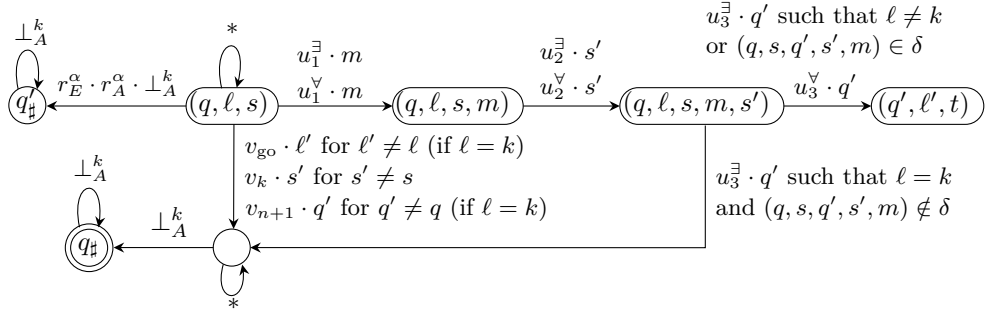

The \DBW{} $\aut{A}_{\bot_{A}^k}^{\text{\OK}}$ that accepts $C_{\bot_{A}^k}^{\text{\OK}}$ is partially depicted in \Cref{fig:exptime-hard-dbw-A_k-OK}. Again for simplicity, we suppose that the words accepted by the automaton belong to $\plays(v_{\text{go}})$. The construction is slightly different from that of $\aut{A}^{\text{\OK}}_{\bot_{E}}$: the states no longer store the fragments of the read declared configurations; instead, they indicate their \emph{expected} $k$-projections. Therefore, the initial state is now the expected $(q_0,1,w_k)$. If $(q,\ell,s)$ is the expected current $k$-projection, the loop labeled $*$ on this state means that the automaton reads symbol by symbol the current declared configuration~$\mathfrak{c}$, until entering the transition phase if Adam considers $\mathfrak{c}$ correct with respect to $(q,\ell,s)$ (see items~\ref{item:E-k-validity-first-config} and~\ref{item:E-k-validity-k-projection-config} of $(E,k)$-validity). Otherwise, $\aut{A}_{\bot_{A}^k}^{\text{\OK}}$ moves to a gadget where the word is accepted if it eventually reads $\bot_A^k$.\footnote{\cref{footnote:continue-playing} can be repeated here for Adam.} If Adam considers $\mathfrak{c}$ correct, the automaton then stores the information of the transition phase up to reading $u_3^\exists \cdot q'$ or $u_3^\forall \cdot q'$. If Eve was playing, $\ell = k$, and the transition she declared did not exist, then it enters the same gadget as before. Otherwise, it moves to the next expected $k$-projection $(q',\ell',t)$ as described in item~\ref{item:E-k-validity-k-projection-config}  of $(E,k)$-validity, and the scenario is repeated for $(q',\ell',t)$.

The \DBW{} $\aut{A}_{\bot_{A}^k}^{\text{\KO}}$ for $C_{\bot_{A}^k}^{\text{\KO}}$ is similar to $\aut{A}_{\bot_{A}^k}^{\text{\OK}}$ except that its unique final state is $q'_\sharp$. Again, both automata have polynomial size. We are now able to explain why Adam needs several sink vertices $\bot^k_A$. Assume that he only uses one $\bot_A$ and one $C_{\bot_A}^{\text{\OK}}$ as for Eve. Then he would need to store the full content of the tape to treat all $k \in \{1,\dots,n\}$, leading to an exponential state space for $\aut{A}_{\bot_{A}}^{\text{\OK}}$. And if two sets $C_{\bot_A^k}^{\text{\OK}}$ and $C_{\bot_A^{k'}}^{\text{\OK}}$ with $k \neq k'$ share the same $\bot_A$, then a play $\pi$ both $(E,k)$-invalid and $(E,k')$-invalid would belong to their intersection, while a partition is required.

\section{Constrained NE Existence Problem is EXPTIME-Complete}
\label{app:exptime-ne-constrained}

In this section, we prove that the constrained NE existence problem is \exptimeComplete{} for \X{} games. 

\begin{theorem}\label{theorem:ne-exptime-complete}
    The constrained NE existence problem is \exptimeComplete{} for \X{} games. The hardness holds even for two players.
\end{theorem}

The \exptimeHard{}ness is obtained with the same reduction used for the constrained SPE existence problem (see \Cref{section:exptime-hardness}). Indeed, we proved that if the ATM $M$ accepts the input word $w$, then there exists an SPE (thus an NE) in $\game$ from $v_{\text{go}}$; and if $M$ does not accept $w$, then there is no NE in $\game$ from $v_{\text{go}}$. Therefore, we here focus on the \exptime{}-membership:

\begin{proposition}\label{prop:prover-game-ne-exptime}
    The constrained NE existence problem is in \exptime{} for \X{} games.
\end{proposition}

We use a Prover game similar to the one in \cref{section:exptime-membership}, but much simpler. It is again a two-player zero-sum Emerson-Lei game, where the state space and objective are simplified; we explain the changes made below. The game starts at $(v_0,\boutelement_0,\varnothing)$, meaning that we are at vertex $v_0$, $\boutelement_0$ is the announced payoff profile, and no player deviated yet. At each step from $(v,\boutelement,\varnothing)$, Prover \prover{} suggests an edge $(v,v') \in E$, from which Challenger \challenger{} can either accept, leading to $(v',\boutelement,\varnothing)$, or make player~$i$ such that $v \in V_i$ deviate to $(v,v'') \in E$, leading to $(v'',\outelement_i,i)$. From a state of the form $(v,\outelement,i)$, \prover{} now acts for the coalition $-i$ retaliating against player~$i$, while \challenger{} acts for player~$i$, trying to achieve a payoff preferred to $\outelement$. The component equal to $\varnothing$ or some player~$i$ is called the \emph{deviating component}. With each state of the Prover game, for all $i \in \players$, we also associate tuples of states $\bar{q}^i = (q^i_1,\dots,q^i_{\ell_i})$ denoting the current state vector for player $i$ such that $q^i_j$ is a state of $\aut{A}^i_{\outelement_j}$, for all $\outelement_j \in \outset_i$.

\subparagraph*{Arena}
The state space $S$ of $\game_{\prover,\challenger}$ consists of the following vertices, connected by some transitions:
\begin{itemize}
    \item $s = \left(v, \boutelement, \varnothing, \bar{q}^1, \dots, \bar{q}^n \right)$, owned by $\prover$, where $v \in V$ is the current vertex and $\boutelement$ is the announced payoff profile. From $s$, Prover can move to $\left((v,v'), \boutelement, \varnothing, \bar{q}^1, \dots, \bar{q}^n \right)$, with $(v,v') \in E$;
    \item $s = \left((v,v'), \boutelement, \varnothing, \bar{q}^1, \dots, \bar{q}^n \right)$, owned by $\challenger$, where $(v,v') \in E$. From $s$, Challenger can move to $\left(v', \boutelement, \varnothing, \bar{q}'^1, \dots, \bar{q}'^n \right)$ or $\left(v'', \outelement_i, i, \bar{q}'^i\right)$, with $v''$ such that $(v,v'') \in E$, $i \in \players$ such that $v \in V_i$, and where all automata update to $\bar{q}'^j = \bar\delta_j(\bar q^j,v)$, for all $j \in \players$;
    \item $s = \left(v, \outelement_i, i, \bar{q}^i \right)$, owned by $\challenger$ if $v \in V_i$, and by $\prover$ if $v \not\in V_i$. The player owning $s$ can move to $\left(v', \outelement_i, i, \bar{q}'^i \right)$ such that $(v,v') \in E$, and where $\bar{q}^i$ is updated as above.
\end{itemize}

\subparagraph*{Prover's Objective}
To formally express the Emerson-Lei objective for Prover in $\game_{\prover,\challenger}$, we first define specific subsets of states to be tracked by the $\infocc$ and $\finocc$ operators during a play. Let $i$ be a player and $\outelement$ be a payoff in $\outset_1 \cup \dots \cup \outset_n$:
\begin{itemize}
    \item $S(\varnothing)$: The set of all states where the deviating component is $\varnothing$.
    \item $S(i)$: The set of all states where the deviating component is $i \in \players$.
    \item $S_{\text{ann}}(\outelement)$: The set of all states where $\outelement$ is a component of an announced payoff profile $\boutelement$ (i.e., states of the form $(\dots, \boutelement, \varnothing, \dots)$ for some $\boutelement$ and $i$ such that $\outelement = \outelement_i$, or states of the form $(\dots,\outelement,i,\dots)$ for some $i$).
\end{itemize}
Note that the target sets of formulas $\varphi^i_\outelement$ used for each $\aut{A}^i_\outelement$ are reused, naturally extended to the Prover game state space (by abuse of notation, we keep the initial notation $\varphi^i_\outelement$).

The Emerson-Lei objective for Prover $\prover$ in the Prover game $\game_{\prover,\challenger}$ is described by the Boolean formula $\varphi_\prover$, defined as the disjunction of two scenarios:
\begin{enumerate}
    \item \label{ne:prover-cond1} \textbf{Stable Outcome:} No player deviates, and the play satisfies the formulas of the announced payoff profile.
    \[
    \infocc(S(\varnothing)) ~\wedge~ \bigwedge_{i \in \players} \bigwedge_{\outelement \in \outset_i} \left( \infocc(S_{\text{ann}}(\outelement)) \Longrightarrow \varphi_\outelement^i \right)
    \]

    \item \label{ne:prover-cond2} \textbf{Deviation:} A player $j$ deviates, and the deviation does not yield a payoff preferred to the announced payoff.
    \[
    \bigvee_{j\in\players} \left( \infocc(S(j)) \wedge \bigwedge_{\doutelement \in \outset_j} \left( \infocc(S_{\text{ann}}(\doutelement)) \Longrightarrow \bigvee_{\outelement' \in \outset_j, \doutelement \notR_j \outelement'} \varphi_{\outelement'}^j \right) \right)
    \]
\end{enumerate}

\subparagraph*{Correspondence with NEs} We can now show that NEs in $\game$ correspond to winning strategies for Prover in $\game_{\prover,\challenger}$. We then prove the \exptime{}-membership stated in \cref{prop:prover-game-ne-exptime}.

\begin{proposition} \label{prop:NEcorrespondence}
    Given a vertex $v_0 \in V$ and a payoff profile $\boutelement_0 \in \outset_1 \times \dots \times \outset_n$, there exists an NE $\bsigma$ whose outcome $\outcomefrom{\bsigma}{v_0}$ has payoff profile $\boutelement_0$ if and only if Prover has a winning strategy in the Prover game $\game_{\prover,\challenger}$ with Emerson-Lei objective $\varphi_\prover$ from the state $s_0 = (v_0,\boutelement_0,\varnothing,\bar q_0^1,\dots,\bar q_0^n)$ such that each $\bar q_0^j$ is the tuple of initial states of $\aut{A}_\outelement^j$, $\outelement \in \outset_j$.
\end{proposition}

\begin{proof}
    Suppose that there exists an NE $\bsigma = (\sigma_i)_{i \in \players}$ from $v_0$ in $\game$ and denote by $\pi$ the outcome $\outcomefrom{\bsigma}{v}$ and by $\boutelement_0$ the payoff profile of $\pi$. Let us construct a strategy $\tau_{\prover}$ for Prover that simulates $\bsigma$ and that will be proved to be winning. Note first that any play or history in $\game_{\prover,\challenger}$ can be lifted to a play or history in $\game$ by projection to the set $V$, skipping each state of the form $((v,v'),\dots)$. The strategy $\tau_{\prover}$ is defined as follows: consider any history in $\game_{\prover,\challenger}$ starting from $s_0$, and let $hv$ be its projection onto the set $V$ of $\game$. We have that $\sigma_i(hv) = v'$ such that $v \in V_i$. Then, if the history ends in a state $s = (v,\boutelement,\varnothing,\dots)$, Prover's strategy $\tau_{\prover}$ moves to $((v,v'),\boutelement,\varnothing,\dots)$. Otherwise, the history ends in a state $s = (v,\outelement,j,\dots)$ for $j \neq i$, and $\tau_{\prover}$ moves to $(v',\outelement,j,\dots)$.

    Let $\tau_{\challenger}$ be any Challenger strategy. We analyze the resulting play $\rho = \outcomefrom{(\tau_\prover,\tau_\challenger)}{s_0}$ and show that it satisfies $\varphi_\prover$, that is, one of the two formulas~\ref{ne:prover-cond1} and~\ref{ne:prover-cond2} defined above is satisfied:
    \begin{enumerate}
        \item \emph{No deviation:} Along $\rho$, assume that Challenger always accepts the edge suggested by Prover. Then $\infocc(S(\varnothing))$ holds and $\pi$ is the projection of $\rho$ in $\game$. By hypothesis, the payoff profile of $\pi$ is equal to $\boutelement_0 = (\outelement_1,\dots,\outelement_n)$, and is thus the profile announced by Prover. Hence, $\infocc(S_{ann(\outelement)})$ exactly holds for each $\outelement \in \{\outelement_1,\dots,\outelement_n\}$. As $\boutelement$ is the payoff profile of $\pi$, each formula $\varphi^i_{\outelement_i}$ holds, for each $\outelement_i$. We conclude that Formula~\ref{ne:prover-cond1} is satisfied.

        \item \emph{Deviation:} If eventually, Challenger makes some player~$j$ deviate, it follows that $\rho$ visits infinitely often states in $S(j)$, and $\infocc(S(j))$ holds. As $\bsigma$ is an NE and by definition of $\tau_\prover$, the projection $\pi'$ of $\rho$ is such that $\pay_j(\pi) \notR_j \pay_j(\pi')$. Recall that the payoff profile of $\pi$ is $\boutelement_0 = (\outelement_0,\dots,\outelement_n)$. Let $\outelement' \in \outset_j$ be such that $\pay_j(\pi') =  \outelement'$, it follows that have $\outelement_j \notR_j \outelement'$. We have that $\infocc(S(\outelement_j))$ holds only for $\outelement_j$ and $\varphi^i_{\outelement'}$ also holds. Hence, Formula \ref{ne:prover-cond2} is satisfied.
    \end{enumerate}
    Consequently, $\tau_{\prover}$ is a winning strategy for Prover.

    \medskip

    Conversely, suppose that Prover has a winning strategy $\tau_\prover$ from $s_0$ in $\game_{\prover,\challenger}$. From $\tau_\prover$, we construct a strategy profile $\bsigma = (\sigma_i)_{i \in \players}$ from $v_0$ in $\game$ as follows.

    For each history $hv$ in $\game$, there exists a unique history $g$ in $\game_{\prover,\challenger}$ whose projection is $hv$ and $g$ ends in a state $(v,\boutelement,\varnothing,\dots)$ or $(v,\outelement,j,\dots)$. In the first case, if $v \in V_i$ and $\tau_\prover(g) = ((v,v'),\boutelement,\varnothing,\dots)$, we define $\sigma_i(hv)=v'$. In the second case, if $v \in V_i$ with $i \neq j$ and $\tau_\prover(g) = (v',\outelement,j,\dots)$, we define $\sigma_i(hv)=v'$. We denote by $\pi$ the outcome $\outcomefrom{\bsigma}{v_0}$.

    Let $\tau_\challenger$ be the strategy of Challenger that always accepts the edges suggested by Prover. Consider the play $\rho = \outcomefrom{(\tau_\prover,\tau_\challenger)}{s_0}$. By construction of $\bsigma$, its projection on $\game$ is equal to $\pi$. As $\tau_\prover$ is winning, Formula~\ref{ne:prover-cond1} is satisfied along $\rho$, so the payoff profile of $\pi$ is $\boutelement_0 = (\outelement_1,\dots,\outelement_n)$ by construction.

    Let us show that $\bsigma$ is an NE from $v_0$ in $\game$.
    Let $\tau_j$ be a deviating strategy of player~$j$. Our goal is to show that $\pi' = \outcomefrom{\bsigma_{-j},\tau_j}{v_0}$ is not a profitable deviation for $j$, i.e., $\outelement_j = \pay_j(\pi) \notR_j \pay_j(\pi')$. There is a unique play $\rho'$ in $\game_{\prover,\challenger}$ whose projection onto $\game$ is $\pi'$. Note that since $\tau_\prover$ is winning, Formula~\ref{ne:prover-cond2} is satisfied along $\rho'$ by construction of $\bsigma$. As $\infocc(S_{\text{ann}}(\outelement_j))$ is satisfied, it must hold that the only payoff $\outelement \in \outset_j$ such that $\rho'$ satisfies $\varphi_{\outelement}^j$ also satisfies $\outelement_j \notR_j \outelement$.
    Consequently, $\bsigma$ is an NE.
\end{proof}

Thanks to the previous proposition, we are now able to prove \cref{prop:prover-game-ne-exptime}. We do not give details of the proof, as it is very close to the proof of \Cref{prop:membership}, with almost the same size for the arena and Emerson-Lei objectives. Moreover, as for SPEs, we get the corollary that the existence NE problem is in \exptime{}.

\section{Constrained SPE Existence Problem is EXPTIME-Hard (Fixed Number of Payoffs)}
\label{app:spe-exptime-hard-fixed-payoffs}

In this section, we provide another proof for the \exptime{}-hardness of the constrained SPE existence problem for \X{} games (\cref{theorem:main-second-result}). The reduction, from the membership problem for linear space ATMs, is similar to the one of \cref{section:exptime-hardness} for \cref{prop:SPE-exptime-hard-fixed-players} and, instead of using a fixed number of players, uses a fixed number of payoffs for each player.

\begin{proposition}\label{prop:SPE-exptime-hard-fixed-payoffs}
    The constrained SPE existence problem for \X{} games is \exptimeHard{}, already for three payoffs per player.
\end{proposition}

We refer to the notation and definitions from \cref{section:exptime-hardness}, and we explain here the main differences. Intuitively, the main difference is now that Adam and Eve are not verifying the validity the configurations and transitions: new external players called Referees, one for each cell, have this role. 

Given a linear space ATM $M$ and an input word $w = w_1 \ldots w_n$ of length~$n$, we construct a similar \X{} game $\game$ with $n+2$ players. The set of players is $\players = \{\text{Eve}, \text{Adam}, R_1, \dots, R_{n}\}$, where each $R_k$ is a referee. The two first phases, the configuration and transition phases are exactly the same. Only the referee phase changes: it consists in a sequence of vertices $r_1^\alpha, r_2^\alpha \dots r_{n}^\alpha$, where each $r_k^\alpha$ is controlled by Referee $R_k$, $k \in \{1,\ldots,n\}$, and $\alpha \in \{\text{acc},\text{rej},\text{go}\}$. Each Referee has two available edges: $(i)$ $R_k$ can choose to pass, i.e., move to $r_{k+1}^\alpha$ (or, when $k = n$, move to $v_{\alpha}$), or $(ii)$, $R_k$ can choose to flag the play as invalid, i.e., move to the sink vertex $\bot$. The game is illustrated in \cref{fig:exptime-hardness-arena-referees}. 

\begin{figure}
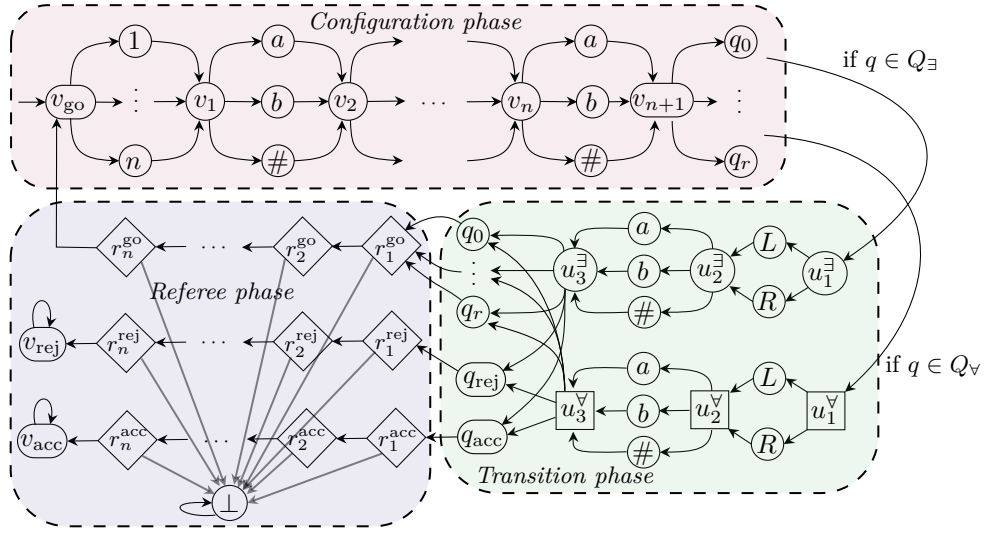

    \centering

    \caption{The game $\game$ used in the \exptime{} reduction, where circles are owned by Eve, squares by Adam, and diamonds by Referees.}
    \label{fig:exptime-hardness-arena-referees}
\end{figure}

Each Referee $R_k$ checks the move validity of \emph{both players} Eve and Adam at the same time. Therefore, we say that a play $\pi$ is \emph{$k$-valid} if:
    \begin{enumerate}
        \item For any configuration $\mathfrak{c} = (q,\ell,x)$ of $\pi$ such that $\ell = k$, let $\mathfrak{t} = (m,s',q')$ be the transition declared\footnote{Here, we no longer make a distinction based on which player declares the transition.} just after $\mathfrak{c}$, then $(q,x_k,q',s',m) \in \delta$.
        \item The $k$-projection of the first configuration of $\pi$ is equal to $(q_0,1,w_1)$ if $k = 1$, and to $(\cdot,\cdot,w_k)$ if $k \neq 1$.
        \item For any two successive configurations $\mathfrak{c} = (q,\ell,x)$ and $\mathfrak{c}'$ of $\pi$ separated by the transition $\mathfrak{t} = (m,s',q')$, we define $(q',\ell',t)$ such that $\ell' = \begin{cases}
        \ell - 1 & \text{if } m = L,\\
        \ell + 1 & \text{if } m = R,
        \end{cases}$ and $t = \begin{cases}
        s' & \text{if } \ell = k,\\
        s  & \text{if } \ell \neq k.
        \end{cases}$ Then the $k$-projection of $\mathfrak{c}'$ is equal to $(q',\ell',t)$ if $\ell' = k$, and to $(\cdot,\cdot,t)$ if $\ell' \neq k$.
    \end{enumerate}
Otherwise the play $\pi$ is \emph{$k$-invalid}. Moreover, a history is \emph{$k$-valid} if it is prefix of a $k$-valid play, otherwise it is \emph{$k$-invalid}. 

We keep the definitions of $C_{\text{acc}}$, $C_{\text{rej}}$, $C_{\text{go}}$, and $\overline{\plays(v_{\text{go}})}$, but change the other sets:
\begin{itemize}
    \item $C_\bot$ is the set of plays in $\plays(v_{\text{go}})$ that eventually loop on $\bot$;
    \item $C_k$ is the set of plays in $\plays(v_{\text{go}})$ that are $k$-valid;
    \item $C_k^{\text{\KO}}$ is the set of plays in $\plays(v_{\text{go}})$ that are $k$-invalid.
\end{itemize}

The preference relations $\R_E,\R_A,\R_1,\dots,\R_n$ for players $\text{Eve}, \text{Adam}, R_1, \dots, R_{n}$ respectively, are the strict partial orders defined as follows, where notation $C_{\text{acc},\text{rej}}$ means $C_{\text{acc}} \cup C_\text{rej}$ (similarly for $C_{\text{go},\bot}$ and $C_{\text{go},\text{acc},\text{rej}}$):
\begin{itemize}
    \item For Eve: $$\overline{\plays(v_\text{go})} \cup C_{\text{go}} \cup C_\bot ~\R_E~ C_{\text{rej}} ~\R_E~ C_{\text{acc}}.$$

    \item For Adam: $$\overline{\plays(v_\text{go})} \cup C_{\text{go}} \cup C_\bot ~\R_A~ C_{\text{acc}} ~\R_A~ C_{\text{rej}}.$$
    
    \item For $R_k$: $$\overline{\plays(v_\text{go})} \cup ((C_k \cap C_{\text{go},\bot}) \cup (C_k^{\text{\KO}} \cap C_{\text{go},\text{acc},\text{rej}})) ~\R_k~ (C_k^{\text{\KO}} \cap C_{\bot}) ~\R_k~ (C_k \cap C_{\text{acc},\text{rej}}).$$
\end{itemize}
Hence, only three payoffs are used for each player, and one can check the sets used to define each preference relation form a partition of $V^\omega$. To prove that $\game$ is an \X{} game, we still need to show that all these sets are $\omega$-regular. As \DBWs{} are closed under union and intersection with polynomial size constructions, it is enough to provide \DBWs{} for the sets $C_{\text{acc}}$, $C_{\text{rej}}$, $C_{\text{go}}$, $\overline{\plays(v_{\text{go}})}$, $C_\bot$, 
$C_k$, and $C_k^{\text{\KO}}$. We concentrate on the last two cases, the others being easy to treat. The \DBW{} $\aut{A}_k$ accepting $C_k$ is depicted in \Cref{fig:exptime-hard-dbw-valid}. It is close to the \DBW{} of \Cref{fig:exptime-hard-dbw-A_k-OK}. The main differences are that this automaton has to also check $A$-validity of plays, does not impose that accepting plays eventually loop on $\bot_k$, and has all its states final except $q_\sharp$. The \DBW{} accepting $C_k^{\text{\KO}}$ is similar to $\aut{A}_k$ except that its unique final state is $q_\sharp$.

\begin{figure}
    \centering
    \begin{tikzpicture}[x=0.75pt,y=0.75pt,yscale=-1]
        \draw   (99.85,112.82) .. controls (99.85,107.75) and (103.96,103.63) .. (109.04,103.63) -- (136.83,103.63) .. controls (141.91,103.63) and (146.02,107.75) .. (146.02,112.82) -- (146.02,112.82) .. controls (146.02,117.9) and (141.91,122.01) .. (136.83,122.01) -- (109.04,122.01) .. controls (103.96,122.01) and (99.85,117.9) .. (99.85,112.82) -- cycle ;
        \draw (117.92,101.57) .. controls (114.03,84.24) and (128.1,83.29) .. (125.95,98.43);
        \draw [shift={(125.41,101.23)}, rotate = 283.72,fill={rgb, 255:red, 0;green,0;blue,0},line width=0.08,draw opacity=0] (5.36,-2.57) -- (0,0) -- (5.36,2.57) -- (3.56,0) -- cycle;
        \draw (122.89,171.41) -- (123.05,124.95);
        \draw [shift={(122.88,174.41)}, rotate = 270.19,fill={rgb, 255:red, 0;green,0;blue,0},line width=0.08,draw opacity=0] (5.36,-2.57) -- (0,0) -- (5.36,2.57) -- (3.56,0) -- cycle;
        \draw   (113.6,183.69) .. controls (113.6,178.57) and (117.76,174.41) .. (122.88,174.41) .. controls (128,174.41) and (132.16,178.57) .. (132.16,183.69) .. controls (132.16,188.82) and (128,192.97) .. (122.88,192.97) .. controls (117.76,192.97) and (113.6,188.82) .. (113.6,183.69) -- cycle ;
        \draw (114.37,187.77) .. controls (96.98,191.77) and (94.38,177.08) .. (111.19,179.72);
        \draw [shift={(114,180.29)}, rotate = 193.62,fill={rgb, 255:red, 0;green,0;blue,0},line width=0.08,draw opacity=0] (5.36,-2.57) -- (0,0) -- (5.36,2.57) -- (3.56,0) -- cycle;
        \draw (197.62,112.82) -- (148.65,112.82);
        \draw [shift={(200.62,112.82)}, rotate = 180,fill={rgb, 255:red, 0;green,0;blue,0},line width=0.08,draw opacity=0] (5.36,-2.57) -- (0,0) -- (5.36,2.57) -- (3.56,0) -- cycle;
        \draw (312.62,112.82) -- (271.07,112.82);
        \draw [shift={(315.62,112.82)}, rotate = 180,fill={rgb, 255:red, 0;green,0;blue,0},line width=0.08,draw opacity=0] (5.36,-2.57) -- (0,0) -- (5.36,2.57) -- (3.56,0) -- cycle;
        \draw (360.45,124.55) -- (360.45,183.75) -- (135.16,183.69);
        \draw [shift={(132.16,183.69)}, rotate = 0.01,fill={rgb, 255:red, 0;green,0;blue,0},line width=0.08,draw opacity=0] (5.36,-2.57) -- (0,0) -- (5.36,2.57) -- (3.56,0) -- cycle;
        \draw (437.42,112.82) -- (403.47,112.82);
        \draw [shift={(440.42,112.82)}, rotate = 180,fill={rgb, 255:red, 0;green,0;blue,0},line width=0.08,draw opacity=0] (5.36,-2.57) -- (0,0) -- (5.36,2.57) -- (3.56,0) -- cycle;
        \draw   (97.22,112.82) .. controls (97.22,106.33) and (102.49,101.07) .. (108.98,101.07) -- (136.89,101.07) .. controls (143.38,101.07) and (148.65,106.33) .. (148.65,112.82) -- (148.65,112.82) .. controls (148.65,119.32) and (143.38,124.58) .. (136.89,124.58) -- (108.98,124.58) .. controls (102.49,124.58) and (97.22,119.32) .. (97.22,112.82) -- cycle ;
        \draw   (203.46,112.82) .. controls (203.46,107.75) and (207.58,103.63) .. (212.65,103.63) -- (259.04,103.63) .. controls (264.11,103.63) and (268.23,107.75) .. (268.23,112.82) -- (268.23,112.82) .. controls (268.23,117.9) and (264.11,122.01) .. (259.04,122.01) -- (212.65,122.01) .. controls (207.58,122.01) and (203.46,117.9) .. (203.46,112.82) -- cycle ;
        \draw   (200.62,112.82) .. controls (200.62,106.33) and (205.89,101.07) .. (212.38,101.07) -- (259.31,101.07) .. controls (265.8,101.07) and (271.07,106.33) .. (271.07,112.82) -- (271.07,112.82) .. controls (271.07,119.32) and (265.8,124.58) .. (259.31,124.58) -- (212.38,124.58) .. controls (205.89,124.58) and (200.62,119.32) .. (200.62,112.82) -- cycle ;
        \draw   (319.4,112.82) .. controls (319.4,107.75) and (323.52,103.63) .. (328.59,103.63) -- (390.5,103.63) .. controls (395.57,103.63) and (399.69,107.75) .. (399.69,112.82) -- (399.69,112.82) .. controls (399.69,117.9) and (395.57,122.01) .. (390.5,122.01) -- (328.59,122.01) .. controls (323.52,122.01) and (319.4,117.9) .. (319.4,112.82) -- cycle ;
        \draw   (315.62,112.82) .. controls (315.62,106.33) and (320.89,101.07) .. (327.38,101.07) -- (391.71,101.07) .. controls (398.2,101.07) and (403.47,106.33) .. (403.47,112.82) -- (403.47,112.82) .. controls (403.47,119.32) and (398.2,124.58) .. (391.71,124.58) -- (327.38,124.58) .. controls (320.89,124.58) and (315.62,119.32) .. (315.62,112.82) -- cycle ;
        \draw   (443.64,112.82) .. controls (443.64,107.75) and (447.75,103.63) .. (452.83,103.63) -- (491.06,103.63) .. controls (496.14,103.63) and (500.25,107.75) .. (500.25,112.82) -- (500.25,112.82) .. controls (500.25,117.9) and (496.14,122.01) .. (491.06,122.01) -- (452.83,122.01) .. controls (447.75,122.01) and (443.64,117.9) .. (443.64,112.82) -- cycle ;
        \draw   (440.42,112.82) .. controls (440.42,106.33) and (445.69,101.07) .. (452.18,101.07) -- (491.71,101.07) .. controls (498.2,101.07) and (503.47,106.33) .. (503.47,112.82) -- (503.47,112.82) .. controls (503.47,119.32) and (498.2,124.58) .. (491.71,124.58) -- (452.18,124.58) .. controls (445.69,124.58) and (440.42,119.32) .. (440.42,112.82) -- cycle ;

        \draw (99.67,105.27) node [anchor=north west,inner sep=0.75pt,font=\normalsize,align=left] {$(q,\ell ,s)$};
        \draw (116.13,177.53) node [anchor=north west,inner sep=0.75pt,font=\normalsize,align=left] {$q_{\sharp }$};
        \draw (127.1,126.23) node [anchor=north west,inner sep=0.75pt,font=\small,align=left] {$v_{\text{go}} \cdot \ell '$ for $\ell '\neq \ell $ (if $\ell =k$)};
        \draw (126.9,141.83) node [anchor=north west,inner sep=0.75pt,font=\small,align=left] {$v_{k} \cdot s'$ for $s'\neq s$};
        \draw (127.3,157.43) node [anchor=north west,inner sep=0.75pt,font=\small,align=left] {$v_{\text{state}} \cdot q'$ for $q'\neq q$ (if $\ell =k$)};
        \draw (203.07,104.47) node [anchor=north west,inner sep=0.75pt,font=\normalsize,align=left] {$(q,\ell ,s,m)$};
        \draw (154.3,96.03) node [anchor=north west,inner sep=0.75pt,font=\small,align=left] {$u_{1}^{\forall } \cdot m$};
        \draw (317.67,103.27) node [anchor=north west,inner sep=0.75pt,font=\normalsize,align=left] {$(q,\ell ,s,m,s')$};
        \draw (277.1,95.03) node [anchor=north west,inner sep=0.75pt,font=\small,align=left] {$u_{2}^{\forall } \cdot s'$};
        \draw (364.7,138.23) node [anchor=north west,inner sep=0.75pt,font=\small,align=left] {$u_{3}^{\exists } \cdot q'$, $u_{3}^{\forall } \cdot q'$ such that\\$\ell =k$ and $(q,s,q',s',m) \notin \delta $};
        \draw (404.1,77.23) node [anchor=north west,inner sep=0.75pt,font=\small,align=left] {$u_{3}^{\exists } \cdot q'$};
        \draw (117.3,76.83) node [anchor=north west,inner sep=0.75pt,font=\small,align=left] {$*$};
        \draw (443.27,104.47) node [anchor=north west,inner sep=0.75pt,font=\normalsize,align=left] {$(q',\ell ',s')$};
        \draw (154.4,78) node [anchor=north west,inner sep=0.75pt,font=\small,align=left] {$u_{1}^{\exists } \cdot m$};
        \draw (277,77.4) node [anchor=north west,inner sep=0.75pt,font=\small,align=left] {$u_{2}^{\exists } \cdot s'$};
        \draw (403.5,94.63) node [anchor=north west,inner sep=0.75pt,font=\small,align=left] {$u_{3}^{\forall } \cdot q'$};
        \draw (88.3,179.63) node [anchor=north west,inner sep=0.75pt,font=\small,align=left] {$*$};
        \draw (440.7,70.03) node [anchor=north west,inner sep=0.75pt,font=\footnotesize,align=left] {such that $\ell \neq k$ or\\$(q,s,q',s',m) \in \delta $};
    \end{tikzpicture}
    \caption{The \DBW{} $\aut{A}_k$ accepting $C_k$.}
    \label{fig:exptime-hard-dbw-valid}
\end{figure}
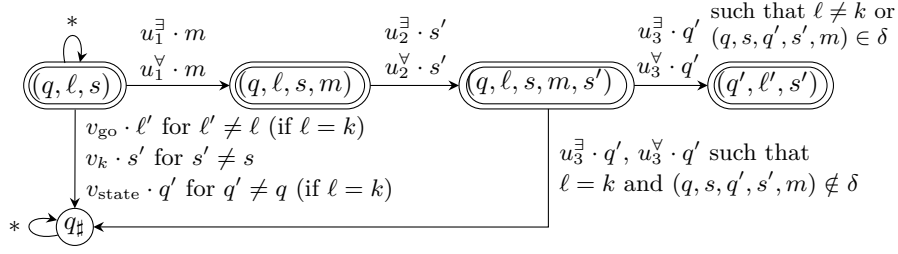

\subparagraph*{Correctness of the Reduction}
 We are now ready to prove the correctness of the reduction: the Turing machine $M$ accepts $w$ if and only if there exists an SPE from $v_{\text{go}}$ in $\game$ satisfying the constraint $C_{\text{acc}}$ for all players. In view of the preference relation $\R_E$, this is equivalent to the existence of an SPE whose outcome belongs to $C_{\text{acc}}$. The proof requires some preliminary lemmas.

For each $R_k$, we define the \emph{canonical strategy} $\varsigma_k$ as follows: for $\alpha \in \{\text{go},\text{rej},\text{acc}\}$, $\varsigma_k(h  r_k^{\alpha}) = \bot$ if $h$ is $k$-invalid, otherwise the other successor is chosen. These strategies are the most natural strategies to get the highest payoff for each Referee.

\begin{lemma}\label{lem:spe-exptime-hard-invalid-history-referee}
    Let $\sigma$ be an SPE and $hv$ be an invalid history of $\game$. If $\outcomefrom{\sigma_{\|h}}{v}$ eventually enters a referee phase, then $h\outcomefrom{\sigma_{\|h}}{v} \in C_\bot$.
\end{lemma}

\begin{proof}
    If $hv$ is invalid, then the play $h\outcomefrom{\sigma_{\|h}}{v}$ is also invalid. Hence, there exists some $k$ such that $h\outcomefrom{\sigma_{\|h}}{v} \in C_k^{\text{\KO}}$. By contradiction, assume that $h\outcomefrom{\sigma_{\|h}}{v} \not\in C_\bot$. As $\outcomefrom{\sigma_{\|h}}{v}$ eventually enters a referee phase, Referee $R_k$ has a profitable deviation in the subgame from the next visited vertex $r_k^{\alpha}$, with $\alpha \in \{\text{go},\text{rej},\text{acc}\}$, by playing their canonical strategy $\varsigma_k$. This is impossible as $\sigma$ is an SPE.
\end{proof}

\begin{lemma}\label{lem:spe-exptime-hard-valid-history}
    Let $\sigma$ be an SPE and $h v_{\text{go}}$ be a valid history of $\game$ ending in $v_{\text{go}}$. Then, $h \outcomefrom{\sigma_{\|h}}{v_{\text{go}}}$ is a valid play and does not belong to $C_\bot$.
\end{lemma}

\begin{proof}
    As the computation of $M$ on $w$ is finite, there exists some constant $K$ such that the number of visits to $v_{\text{go}}$ of any valid history is bounded by $K$. Let $h v_{\text{go}}$ be a valid history. The proof is done by induction on its number of visits to $v_{\text{go}}$.

    We begin with the base case: a maximal number of such visits. Hence, in the subgame $\game_{\|h}$, starting from  $v_{\text{go}}$, if Eve and Adam make valid choices in the configuration phase and in the following transition phase,  the play $\outcomefrom{\sigma_{\|h}}{v_{\text{go}}}$ reaches  $q_{\text{acc}}$ or $q_{\text{rej}}$ at the end of the transition phase, and then $v_{\text{acc}}$, $v_{\text{rej}}$ or $\bot$ at the end of the referee phase, without reaching $v_{\text{go}}$ again. We consider the following three cases:

    \begin{enumerate}
        \item Suppose first that Eve and Adam make valid choices. Each Referee $R_k$ has the choice between going to $\bot$ or continuing in a way to reach $v_{\text{acc}}$ or $v_{\text{rej}}$. The first case is not profitable as the play $h \outcomefrom{\sigma_{\|h}}{v_{\text{go}}}$ belongs to $C_k \cap C_{\bot}$ (instead of  $C_k \cap C_{\text{acc,rej}}$). Therefore, $h\outcomefrom{\sigma_{\|h}}{v_{\text{go}}}$ is a valid play that belongs to  $C_{\text{acc,rej}}$.
        \item  Suppose now that starting at the initial vertex $v_{\text{go}}$ of $\game_{\|h}$, Eve makes valid choices in the configuration phase, but Eve or Adam makes invalid choices in the following transition phase. By Lemma~\ref{lem:spe-exptime-hard-invalid-history-referee} (with $h$ extended with the choices made during both configuration and transition phases), it follows that the play $h\outcomefrom{\sigma_{\|h}}{v_{\text{go}}}$ belongs to $C_\bot$. Hence, in the transition phase, Eve or Adam has a profitable deviation: to make valid choices so that the resulting play belongs $C_{\text{acc},\text{rej}}$ by Case~1 above. Consequently, Case~2 never occurs.

        \item Finally, suppose that Eve makes invalid choices during the configuration phase. By Lemma \ref{lem:spe-exptime-hard-invalid-history-referee} (with $h$ extended with Eve's choices), we get $h\outcomefrom{\sigma_{\|h}}{v_{\text{go}}} \in C_\bot$. Hence, by Case 1, Eve has a profitable deviation by making valid choices in the configuration phase and thus producing a play in $C_{\text{acc},\text{rej}}$. Therefore, Case~3 also does not occur.
    \end{enumerate}
    We have shown that among the three cases, only the first one holds and $h\outcomefrom{\sigma_{\|h}}{v_{\text{go}}}$ is a valid play that belongs to  $C_{\text{acc,rej}}$.

    Consider the general case and a valid history $hv_{\text{go}}$. In the subgame $\game_{\|h}$, starting from  $v_{\text{go}}$, if Eve and Adam make valid choices in the first configuration and transition phases and in the following transition phase, then the play $\outcomefrom{\sigma_{\|h}}{v_{\text{go}}}$ reaches  $v_{\text{acc}}$, $v_{\text{rej}}$ or $\bot$, but also possibly $v_{\text{go}}$, at the end of the next referee phase. In the last case, let $h'v_{\text{go}}$ be the history that extends $hv_{\text{go}}$ to this visit of $v_{\text{go}}$. By induction hypothesis, $h'\outcomefrom{\sigma_{\|h'}}{v_{\text{go}}}$ is a valid play that belongs to  $C_{\text{acc,rej}}$. Using this property, we can repeat the previous arguments in all three cases and obtain that $h\outcomefrom{\sigma_{\|h}}{v_{\text{go}}}$ is also a valid play that belongs to  $C_{\text{acc,rej}}$.
\end{proof}

\begin{proof}[Proof of \cref{prop:SPE-exptime-hard-fixed-payoffs}] Let us prove that $M$ accepts $w$ if and only if there exists an SPE from $v_{\text{go}}$ with outcome in $C_{\text{acc}}$. 
Let us first assume that $M$ accepts $w$. We take the finite computation tree of $M$ on $w$ from the initial configuration.
At each node of the tree, either $\exists$ wins, or $\forall$ wins. Therefore, there exist two positional strategies $\sigma_\exists$ and $\sigma_{\forall}$, winning from vertices in their respective winning regions. These strategies correspond to $\sigma_E$ and $\sigma_A$, strategies of Eve and Adam in $\game$ producing valid choices. We extend them arbitrarily after any invalid history. Let us consider the following strategy profile: $\bsigma = (\sigma_E,\sigma_A,\varsigma_1,\dots,\varsigma_n)$, defined with the canonical strategies of Referees. As $M$ accepts $w$, the root of the computation tree on $w$ is winning for $\exists$, and then by construction, the outcome of $\bsigma$ from $v_{\text{go}}$ eventually loops in $v_{\text{acc}}$. It remains to show that $\bsigma$ is an SPE, that is, for each history $hv$, $\bsigma_{\| h}$ is an NE in the subgame $\game_{\| h}$ from $v$. If $hv$ is not valid, by construction of the canonical strategies, $h \outcomefrom{\bsigma_{\|h}}{v} \in C_\bot$. Thus, in $\game_{\| h}$, no player has an incentive to deviate: Referees already obtain the best payoff they can, and no matter what Adam or Eve do, the resulting play still belongs to $C_\bot$. If $hv$ is valid and $v$ is winning for Eve (resp.\ winning for Adam), then $h \outcomefrom{\sigma_{\|h}}{v} \in C_{\text{acc}}$ (resp.\ $C_{\text{rej}}$). Adam (resp.\ Eve) cannot get a better payoff by deviating as $\sigma_E$ (resp.\ $\sigma_A$) is a winning strategy by construction, and the Referees obtain the best payoff they can.

Let us now assume that there exists an SPE $\bsigma = (\sigma_E,\sigma_A,\sigma_{1},\dots,\sigma_n)$ such that $\rho = \outcomefrom{\bsigma}{v_{\text{go}}} \in C_{\text{acc}}$. We construct a strategy $\sigma_\exists$ resolving non-determinism in $M$, by looking at all valid plays consistent with $\sigma_E$.
Let $\tau_{\forall}$ be a strategy resolving universality in $M$. It corresponds to some strategy $\sigma'_A$ of Adam in $\game$. Let us denote by $\rho'$ the play $\outcomefrom{\bsigma_{-A},\sigma'_A}{v_{\text{go}}}$. As $\sigma$ is an NE from the initial vertex $v_{\text{go}}$, we know that $\rho \notR_A \rho'$, that is, $\rho' \in C_{\text{loop},\text{acc},\bot}$ as $\rho \in C_{\text{acc}}$. Let us show that $\rho'$ is valid. If not, decompose $\rho'$ as $h'\pi'$ where $h'$ is the longest valid history ending in $v_{\text{go}}$. We get a contradiction with Lemma~\ref{lem:spe-exptime-hard-valid-history}: in $\game_{\|h}$, Eve starts making valid choices. Hence as any valid play visits $v_{\text{go}}$ finitely many times, we get $\rho' \in C_{\text{acc},\bot}$. Coming back to the decomposition $\rho' = h'\pi'$, we get $\rho' \not\in C_\bot$ by Lemma~\ref{lem:spe-exptime-hard-valid-history}. This shows that $\rho' \in C_{\text{acc}}$. Therefore, the Turing machine $M$ accepts $w$.
\end{proof}

\section{Constrained NE Existence Problem is PSPACE-Complete (Fixed Number of Payoffs)}
\label{app:ne-recognizable-pspace}

This section is devoted to prove the following \pspaceComplete{}ness result for the constrained NE existence problem, when the number of payoffs of each player is fixed.

\begin{theorem}\label{theorem:ne-pspace-complete}
    The constrained NE existence problem is \pspaceComplete{} for \X{} games where the number of payoffs of each player is fixed.
\end{theorem}

We begin with the proof of \pspace{}-membership. We use an algorithm which is similar to the one proposed in \cite[Section 4]{BruyereFiliotGrandmontRaskin26-arxiv}, where an alternating parity automaton (\APW{}) is used to accept all plays that are NE outcomes.

\begin{proof}[Membership proof of \cref{theorem:ne-pspace-complete}]
Let us first recall the approach of~\cite{BruyereFiliotGrandmontRaskin26-arxiv}.\footnote{We suppose that the reader is familiar with alternating automata on infinite words.} There, each player~$i \in \players$ has a preference relation $\R_i'$ over $V^\omega$ in a way to compare plays (no payoff function is used). This relation is $\omega$-automatic by hypothesis, and thus accepted by a \DPW{}. It is proved that the following set, called \emph{value of player~$-i$}, where $-i$ is the coalition of all players playing against player~$i$, is accepted by an \APW{} of polynomial size:
    \[
    \val_{-i}(v) = \{\pi \in \plays(v) \mid \exists \sigma_{-i} \in \Gamma_{-i}, \forall \rho \in \plays_{\sigma_{-i}}(v),\, \pi \notR_i' \rho\}.
    \]
    ($\Gamma_{-i}$ denotes the set of strategies of player $-i$, and $\plays_{\sigma_{-i}}$ the set of plays consistent with $\sigma_{-i}$.)
    Then, it is shown that the set of all NE outcomes from $v$ is equal to $\bigcap_{i\in\players} \val_{-i}(v)$, and thus also accepted by an \APW{} of polynomial size (see \cite[Lemma 12]{BruyereFiliotGrandmontRaskin26-arxiv}).

    In our setting, recall that each player~$i$ has a fixed number $\ell_i$ of payoffs in the set $\outset_i$, a preference relation $\R_i$ over $\outset_i$, and a \DPW{} accepting the set $\pay_i^{-1}(\outelement_i)$ for each $\outelement_i \in \outset_i$. As explained in \cref{section:preliminaries}, it amounts to have an $\omega$-recognizable (thus $\omega$-automatic) relation $\R_i'$. We can thus show that
    \[
    \val_{-i}(v) = \{\pi \in \plays(v) \mid \exists \sigma_{-i} \in \Gamma_{-i}, \forall \rho \in \plays_{\sigma_{-i}}(v),\, \pay_i(\pi) \notR_i \pay_i(\rho)\}.
    \]
    We are going to construct an alternating Rabin automaton (\ARW{}) $\aut{B}_i$ accepting $\val_{-i}(v)$, for all $i \in \players$. In this way, the set $\bigcap_{i\in\players} \val_{-i}(v)$ of all NE outcomes from $v$ will be also accepted by an \ARW{} \aut{B}, as the \ARWs{} are closed under intersection with a linear size construction \cite{handbook-automata-kupferman18}. We will explain at the end of proof how $\aut{B}$ is used to solve in \pspace{} the constrained NE existence problem.   

    Let us fix $i \in \players$. Consider a word $\pi \in V^\omega$. We define $\aut{B}_i$ by describing a run over $\pi$. Such a run starts by guessing a payoff $\outelement \in \outset_i$ while reading an $\varepsilon$-transition\footnote{We can get rid of $\varepsilon$-transitions by retaining the input letters in the states.}. Then, it universally branches out with an $\varepsilon$-transition to three states:
    \begin{itemize}
        \item To $q_0^i$, the initial state of the \DPW{} $\aut{A}^i_{\outelement}$ for the guessed $\outelement$. From $q_0^i$, the \ARW{} simulates $\aut{A}^i_{\outelement}$ on the input word $\pi$, ensuring that $\pay_i(\pi) = \outelement$. The accepting condition is a parity condition, that can be trivially translated into a Rabin condition.
        \item To $v$, the starting vertex of $\pi$. From $v$, the \ARW{} simulates the arena $\arena$ to ensure that $\pi \in \plays(v)$. The accepting condition is a Büchi condition (all states are final), i.e., a Rabin condition.
        \item To $(v,\overline{q_0}^i)$ where $\overline{q_0}^i$ denotes the tuple of initial states of $(\aut{A}^i_{\outelement_j})_{j \in J}$, with $J = \{j \in \{1,\dots,\ell_i\} \mid \outelement \notR_i \outelement_j\}$. From the state $(v,\overline{q_0}^i)$, we ignore the input word $\pi$ and use the alternation in the same way as in~\cite[Lemma 6]{BruyereFiliotGrandmontRaskin26-arxiv} to generate the tree of all plays $\rho \in \plays(v)$ consistent with a built strategy $\sigma_{-i}$: From any state $(u,\bar{q}^i)$, if $u \in V_i$ (resp.\ $u \not\in V_i$), the state has universal (resp.\ existential) transitions to every state $(u',\overline{q'}^i)$, where $(u,u') \in E$ and $\overline{q'}^i = \bar\delta_i(\bar q^i,u)$. The accepting condition is
        \[
        \bigvee_{j \in J} \Rabin{\aut{A}^i_{\outelement_j}},
        \]
        where $\Rabin{\aut{A}^i_{\outelement_j}}$ is the Rabin translation of the parity objective of $\aut{A}^i_{\outelement_j}$.
        \end{itemize}

    The automaton $\aut{B}_i$ accepts $val_{-i}(v)$. Indeed, a word $\pi$ is accepted by $\aut{B}_i$ if and only if $(i)$ from the first universal branching: there is a payoff $\outelement \in \outset_i$ such that $\pay_i(\pi) = \outelement$ and $\pi \in \plays(v)$, and $(ii)$ from the second universal branching: there exists a strategy $\sigma_{-i}$ (by the existential branching) such that all plays $\rho$ consistent with $\sigma_{-i}$ (by the universal branching), satisfy $\outelement \notR_i \pay_i(\rho)$ (by the Rabin condition).

    The \ARW{} $\aut{B}_i$ has polynomial size: let us define $Q^i$ as $\max \{|Q^i_\outelement| \mid \outelement\in \outset_i\}$ where $Q^i_\outelement$ is the set of states of $\aut{A}^i_\outelement$, and $d^i$ as $\max \{d^i_\outelement \mid \outelement\in \outset_i\}$ where $\{0, \dots, d^i_\outelement\}$ is the set of priorities of $\aut{A}^i_\outelement$. Then, as $|J|$ is bounded by $\ell_i$, $\aut{B}_i$ has a size of $\bigO(|V| \times |Q^i|^{\ell_i})$, with $\bigO(\ell_i\cdot d_i)$ Rabin pairs. As $\ell_i$ is constant, this is polynomial. Consequently, the \ARW{} $\aut{B}$  accepting all NE outcomes is itself of polynomial size.

    To conclude the proof, it remains to show how $\aut{B}$ allows to solve the constrained NE existence problem. Given the constraints $d_i \in P_i$, $i \in \players$, and an initial vertex $v$,  we have to decide whether there exists an NE outcome $\pi$ from $v$ such that $d_i = \pay_i(\pi)$ or $d_i \R_i \pay_i(\pi)$, for all $i$. We proceed as follows. First, for each $i$, we construct an \ARW{} $\aut{C}_i$ of polynomial size that accepts the set $L_i = \{\rho \in V^\omega \mid d_i = \pay_i(\rho) \text{ or } d_i \R_i \pay_i(\rho) \}$. This automaton guesses a payoff $\outelement$ such that $d_i = \outelement$ or $d_i \R_i \outelement$ and then simulates the \DPW{} $\aut{A}^i_{\outelement}$ for the guessed $\outelement$. Second, we construct the \ARW{} $\aut{C}$ equal to the intersection of $\aut{B}$ and each $\aut{C}_i$, that has polynomial size. Finally, to decide whether there exists an NE outcome satisfying each constraint $d_i$, we check whether the language accepted by $\aut{C}$ is nonempty. The nonemptiness problem for \ARWs{} is in \pspace{}: we build a nondeterministic parity automaton equivalent to $\aut{C}$, of single-exponential size~\cite{DBLP:journals/corr/abs-2002-07278}, and then perform the nonemptiness check on-the-fly. Checking nonemptiness of nondeterministic parity automata is in \nl{}.
\end{proof}

The \pspaceHard{}ness result of \cref{theorem:ne-pspace-complete} is inspired by a proof in Section 5.7.4 of~\cite{BouyerBMU15}. The hardness holds when each player has at most two payoffs.

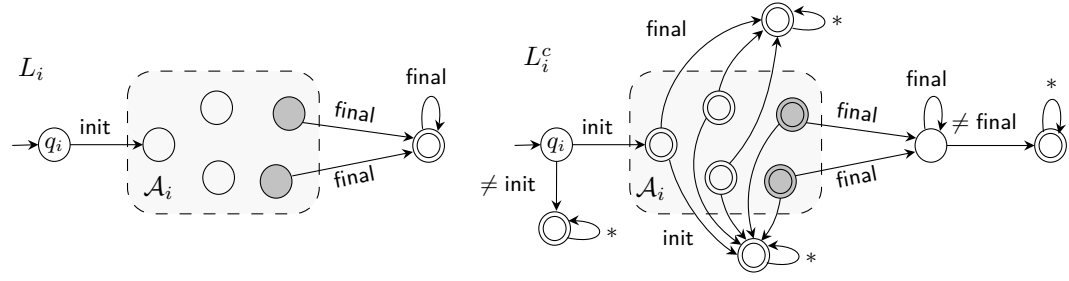
\begin{figure}
    \centering
    \begin{tikzpicture}[x=0.75pt,y=0.75pt,yscale=-1]
        \draw (122.29,127.11) -- (155.84,127.11);
        \draw [shift={(158.84,127.11)},rotate=180,fill={rgb,255:red,0;green,0;blue,0},line width=0.08, draw opacity=0] (5.36,-2.57) -- (0,0) -- (5.36,2.57) -- (3.56,0) -- cycle;
        \draw (106.41,127.11) .. controls (106.41,122.48) and (109.96,118.73) .. (114.35,118.73) .. controls (118.74,118.73) and (122.29,122.48) .. (122.29,127.11) .. controls (122.29,131.73) and (118.74,135.48) .. (114.35,135.48) .. controls (109.96,135.48) and (106.41,131.73) .. (106.41,127.11) -- cycle;
        \draw (158.84,127.11) .. controls (158.84,122.48) and (162.4,118.73) .. (166.79,118.73) .. controls (171.17,118.73) and (174.73,122.48) .. (174.73,127.11) .. controls (174.73,131.73) and (171.17,135.48) .. (166.79,135.48) .. controls (162.4,135.48) and (158.84,131.73) .. (158.84,127.11) -- cycle;
        \draw (294.18,127.52) .. controls (294.18,122.9) and (297.74,119.15) .. (302.13,119.15) .. controls (306.51,119.15) and (310.07,122.9) .. (310.07,127.52) .. controls (310.07,132.15) and (306.51,135.9) .. (302.13,135.9) .. controls (297.74,135.9) and (294.18,132.15) .. (294.18,127.52) -- cycle;
        \draw (296.39,127.52) .. controls (296.39,124.19) and (298.96,121.48) .. (302.13,121.48) .. controls (305.29,121.48) and (307.86,124.19) .. (307.86,127.52) .. controls (307.86,130.86) and (305.29,133.56) .. (302.13,133.56) .. controls (298.96,133.56) and (296.39,130.86) .. (296.39,127.52) -- cycle;
        \draw (239.07,116.02) -- (291.78,123.34);
        \draw [shift={(294.75,123.75)},rotate=187.9,fill={rgb,255:red,0;green,0;blue,0},line width=0.08, draw opacity=0] (5.36,-2.57) -- (0,0) -- (5.36,2.57) -- (3.56,0) -- cycle;
        \draw (299.12,119.19) .. controls (296.19,97.89) and (310.48,97.7) .. (307.29,117.32);
        \draw [shift={(306.72,120.23)},rotate=282.86,fill={rgb,255:red,0;green,0;blue,0},line width=0.08, draw opacity=0] (5.36,-2.57) -- (0,0) -- (5.36,2.57) -- (3.56,0) -- cycle;
        \draw (93.26,127.48) -- (103.41,127.19);
        \draw [shift={(106.41,127.11)},rotate=178.35,fill={rgb,255:red,0;green,0;blue,0},line width=0.08, draw opacity=0] (5.36,-2.57) -- (0,0) -- (5.36,2.57) -- (3.56,0) -- cycle;
        \draw [fill={rgb,255:red,0;green,0;blue,0},fill opacity=0.03,dash pattern={on 4.5pt off 4.5pt}] (151.07,103.3) .. controls (151.07,96.09) and (156.91,90.25) .. (164.12,90.25) -- (234.2,90.25) .. controls (241.41,90.25) and (247.25,96.09) .. (247.25,103.3) -- (247.25,148.7) .. controls (247.25,155.91) and (241.41,161.75) .. (234.2,161.75) -- (164.12,161.75) .. controls (156.91,161.75) and (151.07,155.91) .. (151.07,148.7) -- cycle;
        \draw (187.84,108.61) .. controls (187.84,103.98) and (191.4,100.23) .. (195.79,100.23) .. controls (200.17,100.23) and (203.73,103.98) .. (203.73,108.61) .. controls (203.73,113.23) and (200.17,116.98) .. (195.79,116.98) .. controls (191.4,116.98) and (187.84,113.23) .. (187.84,108.61) -- cycle;
        \draw (188.84,143.61) .. controls (188.84,138.98) and (192.4,135.23) .. (196.79,135.23) .. controls (201.17,135.23) and (204.73,138.98) .. (204.73,143.61) .. controls (204.73,148.23) and (201.17,151.98) .. (196.79,151.98) .. controls (192.4,151.98) and (188.84,148.23) .. (188.84,143.61) -- cycle;
        \draw [fill={rgb, 255:red, 194; green, 194; blue, 194},fill opacity=1 ] (224.18,111.52) .. controls (224.18,106.9) and (227.74,103.15) .. (232.13,103.15) .. controls (236.51,103.15) and (240.07,106.9) .. (240.07,111.52) .. controls (240.07,116.15) and (236.51,119.9) .. (232.13,119.9) .. controls (227.74,119.9) and (224.18,116.15) .. (224.18,111.52) -- cycle;
        \draw [fill={rgb, 255:red, 194; green, 194; blue, 194},fill opacity=1 ] (217.68,145.52) .. controls (217.68,140.9) and (221.24,137.15) .. (225.63,137.15) .. controls (230.01,137.15) and (233.57,140.9) .. (233.57,145.52) .. controls (233.57,150.15) and (230.01,153.9) .. (225.63,153.9) .. controls (221.24,153.9) and (217.68,150.15) .. (217.68,145.52) -- cycle;
        \draw (233.75,142.75) -- (291.31,131.33);
        \draw [shift={(294.25,130.75)},rotate=168.78,fill={rgb,255:red,0;green,0;blue,0},line width=0.08, draw opacity=0] (5.36,-2.57) -- (0,0) -- (5.36,2.57) -- (3.56,0) -- cycle;
        \draw (374.29,127.11) -- (407.84,127.11);
        \draw [shift={(410.84,127.11)},rotate=180,fill={rgb,255:red,0;green,0;blue,0},line width=0.08, draw opacity=0] (5.36,-2.57) -- (0,0) -- (5.36,2.57) -- (3.56,0) -- cycle;
        \draw (358.41,127.11) .. controls (358.41,122.48) and (361.96,118.73) .. (366.35,118.73) .. controls (370.74,118.73) and (374.29,122.48) .. (374.29,127.11) .. controls (374.29,131.73) and (370.74,135.48) .. (366.35,135.48) .. controls (361.96,135.48) and (358.41,131.73) .. (358.41,127.11) -- cycle;
        \draw (410.84,127.11) .. controls (410.84,122.48) and (414.4,118.73) .. (418.79,118.73) .. controls (423.17,118.73) and (426.73,122.48) .. (426.73,127.11) .. controls (426.73,131.73) and (423.17,135.48) .. (418.79,135.48) .. controls (414.4,135.48) and (410.84,131.73) .. (410.84,127.11) -- cycle;
        \draw (546.18,127.52) .. controls (546.18,122.9) and (549.74,119.15) .. (554.13,119.15) .. controls (558.51,119.15) and (562.07,122.9) .. (562.07,127.52) .. controls (562.07,132.15) and (558.51,135.9) .. (554.13,135.9) .. controls (549.74,135.9) and (546.18,132.15) .. (546.18,127.52) -- cycle;
        \draw (491.07,116.02) -- (543.78,123.34);
        \draw [shift={(546.75,123.75)},rotate=187.9,fill={rgb,255:red,0;green,0;blue,0},line width=0.08, draw opacity=0] (5.36,-2.57) -- (0,0) -- (5.36,2.57) -- (3.56,0) -- cycle;
        \draw (551.12,119.19) .. controls (548.19,97.89) and (562.48,97.7) .. (559.29,117.32);
        \draw [shift={(558.72,120.23)},rotate=282.86,fill={rgb,255:red,0;green,0;blue,0},line width=0.08, draw opacity=0] (5.36,-2.57) -- (0,0) -- (5.36,2.57) -- (3.56,0) -- cycle;
        \draw (345.26,127.48) -- (355.41,127.19);
        \draw [shift={(358.41,127.11)},rotate=178.35,fill={rgb,255:red,0;green,0;blue,0},line width=0.08, draw opacity=0] (5.36,-2.57) -- (0,0) -- (5.36,2.57) -- (3.56,0) -- cycle;
        \draw [fill={rgb,255:red,0;green,0;blue,0},fill opacity=0.03,dash pattern={on 4.5pt off 4.5pt}] (403.07,103.3) .. controls (403.07,96.09) and (408.91,90.25) .. (416.12,90.25) -- (486.2,90.25) .. controls (493.41,90.25) and (499.25,96.09) .. (499.25,103.3) -- (499.25,148.7) .. controls (499.25,155.91) and (493.41,161.75) .. (486.2,161.75) -- (416.12,161.75) .. controls (408.91,161.75) and (403.07,155.91) .. (403.07,148.7) -- cycle;
        \draw (439.84,108.61) .. controls (439.84,103.98) and (443.4,100.23) .. (447.79,100.23) .. controls (452.17,100.23) and (455.73,103.98) .. (455.73,108.61) .. controls (455.73,113.23) and (452.17,116.98) .. (447.79,116.98) .. controls (443.4,116.98) and (439.84,113.23) .. (439.84,108.61) -- cycle;
        \draw (440.84,143.61) .. controls (440.84,138.98) and (444.4,135.23) .. (448.79,135.23) .. controls (453.17,135.23) and (456.73,138.98) .. (456.73,143.61) .. controls (456.73,148.23) and (453.17,151.98) .. (448.79,151.98) .. controls (444.4,151.98) and (440.84,148.23) .. (440.84,143.61) -- cycle;
        \draw (485.75,142.75) -- (543.31,131.33);
        \draw [shift={(546.25,130.75)},rotate=168.78,fill={rgb,255:red,0;green,0;blue,0},line width=0.08, draw opacity=0] (5.36,-2.57) -- (0,0) -- (5.36,2.57) -- (3.56,0) -- cycle;
        \draw (366.35,135.48) -- (366.26,157.25);
        \draw [shift={(366.25,160.25)},rotate=270.23,fill={rgb,255:red,0;green,0;blue,0},line width=0.08, draw opacity=0] (5.36,-2.57) -- (0,0) -- (5.36,2.57) -- (3.56,0) -- cycle;
        \draw (357.41,169.11) .. controls (357.41,164.48) and (360.96,160.73) .. (365.35,160.73) .. controls (369.74,160.73) and (373.29,164.48) .. (373.29,169.11) .. controls (373.29,173.73) and (369.74,177.48) .. (365.35,177.48) .. controls (360.96,177.48) and (357.41,173.73) .. (357.41,169.11) -- cycle;
        \draw (371.75,173.75) .. controls (391.81,178.05) and (393.62,165.47) .. (375.02,165.2);
        \draw [shift={(372.25,165.25)},rotate=357.34,fill={rgb,255:red,0;green,0;blue,0},line width=0.08, draw opacity=0] (5.36,-2.57) -- (0,0) -- (5.36,2.57) -- (3.56,0) -- cycle;
        \draw [fill={rgb, 255:red, 194; green, 194; blue, 194},fill opacity=1 ] (476.68,112.02) .. controls (476.68,107.4) and (480.24,103.65) .. (484.63,103.65) .. controls (489.01,103.65) and (492.57,107.4) .. (492.57,112.02) .. controls (492.57,116.65) and (489.01,120.4) .. (484.63,120.4) .. controls (480.24,120.4) and (476.68,116.65) .. (476.68,112.02) -- cycle;
        \draw (478.89,112.02) .. controls (478.89,108.69) and (481.46,105.98) .. (484.63,105.98) .. controls (487.79,105.98) and (490.36,108.69) .. (490.36,112.02) .. controls (490.36,115.36) and (487.79,118.06) .. (484.63,118.06) .. controls (481.46,118.06) and (478.89,115.36) .. (478.89,112.02) -- cycle;
        \draw [fill={rgb, 255:red, 194; green, 194; blue, 194},fill opacity=1 ] (470.68,145.52) .. controls (470.68,140.9) and (474.24,137.15) .. (478.63,137.15) .. controls (483.01,137.15) and (486.57,140.9) .. (486.57,145.52) .. controls (486.57,150.15) and (483.01,153.9) .. (478.63,153.9) .. controls (474.24,153.9) and (470.68,150.15) .. (470.68,145.52) -- cycle;
        \draw (472.89,145.52) .. controls (472.89,142.19) and (475.46,139.48) .. (478.63,139.48) .. controls (481.79,139.48) and (484.36,142.19) .. (484.36,145.52) .. controls (484.36,148.86) and (481.79,151.56) .. (478.63,151.56) .. controls (475.46,151.56) and (472.89,148.86) .. (472.89,145.52) -- cycle;
        \draw (606.18,127.52) .. controls (606.18,122.9) and (609.74,119.15) .. (614.13,119.15) .. controls (618.51,119.15) and (622.07,122.9) .. (622.07,127.52) .. controls (622.07,132.15) and (618.51,135.9) .. (614.13,135.9) .. controls (609.74,135.9) and (606.18,132.15) .. (606.18,127.52) -- cycle;
        \draw (608.39,127.52) .. controls (608.39,124.19) and (610.96,121.48) .. (614.13,121.48) .. controls (617.29,121.48) and (619.86,124.19) .. (619.86,127.52) .. controls (619.86,130.86) and (617.29,133.56) .. (614.13,133.56) .. controls (610.96,133.56) and (608.39,130.86) .. (608.39,127.52) -- cycle;
        \draw (562.07,127.52) -- (603.18,127.52);
        \draw [shift={(606.18,127.52)},rotate=180,fill={rgb,255:red,0;green,0;blue,0},line width=0.08, draw opacity=0] (5.36,-2.57) -- (0,0) -- (5.36,2.57) -- (3.56,0) -- cycle;
        \draw (452.85,136.83) .. controls (465.82,122.34) and (473.27,105.91) .. (476.9,75.74);
        \draw [shift={(477.22,72.9)},rotate=96.24,fill={rgb,255:red,0;green,0;blue,0},line width=0.08, draw opacity=0] (5.36,-2.57) -- (0,0) -- (5.36,2.57) -- (3.56,0) -- cycle;
        \draw (418.79,118.73) .. controls (418.67,95.63) and (440.4,69.64) .. (466.43,64.95);
        \draw [shift={(469.28,64.52)},rotate=173.25,fill={rgb,255:red,0;green,0;blue,0},line width=0.08, draw opacity=0] (5.36,-2.57) -- (0,0) -- (5.36,2.57) -- (3.56,0) -- cycle;
        \draw (447.79,100.23) .. controls (451.2,85.99) and (459.52,78.92) .. (467.83,71.66);
        \draw [shift={(470,69.75)},rotate=138.37,fill={rgb,255:red,0;green,0;blue,0},line width=0.08, draw opacity=0] (5.36,-2.57) -- (0,0) -- (5.36,2.57) -- (3.56,0) -- cycle;
        \draw (469.28,64.52) .. controls (469.28,59.9) and (472.84,56.15) .. (477.22,56.15) .. controls (481.61,56.15) and (485.17,59.9) .. (485.17,64.52) .. controls (485.17,69.15) and (481.61,72.9) .. (477.22,72.9) .. controls (472.84,72.9) and (469.28,69.15) .. (469.28,64.52) -- cycle;
        \draw (471.49,64.52) .. controls (471.49,61.19) and (474.06,58.48) .. (477.22,58.48) .. controls (480.39,58.48) and (482.95,61.19) .. (482.95,64.52) .. controls (482.95,67.86) and (480.39,70.56) .. (477.22,70.56) .. controls (474.06,70.56) and (471.49,67.86) .. (471.49,64.52) -- cycle;
        \draw (442.06,108.61) .. controls (442.06,105.27) and (444.62,102.56) .. (447.79,102.56) .. controls (450.95,102.56) and (453.52,105.27) .. (453.52,108.61) .. controls (453.52,111.94) and (450.95,114.65) .. (447.79,114.65) .. controls (444.62,114.65) and (442.06,111.94) .. (442.06,108.61) -- cycle;
        \draw (443.06,143.61) .. controls (443.06,140.27) and (445.62,137.56) .. (448.79,137.56) .. controls (451.95,137.56) and (454.52,140.27) .. (454.52,143.61) .. controls (454.52,146.94) and (451.95,149.65) .. (448.79,149.65) .. controls (445.62,149.65) and (443.06,146.94) .. (443.06,143.61) -- cycle;
        \draw (413.06,127.11) .. controls (413.06,123.77) and (415.62,121.06) .. (418.79,121.06) .. controls (421.95,121.06) and (424.52,123.77) .. (424.52,127.11) .. controls (424.52,130.44) and (421.95,133.15) .. (418.79,133.15) .. controls (415.62,133.15) and (413.06,130.44) .. (413.06,127.11) -- cycle;
        \draw (484.75,68.75) .. controls (504.81,73.05) and (506.62,60.47) .. (488.02,60.2);
        \draw [shift={(485.25,60.25)},rotate=357.34,fill={rgb,255:red,0;green,0;blue,0},line width=0.08, draw opacity=0] (5.36,-2.57) -- (0,0) -- (5.36,2.57) -- (3.56,0) -- cycle;
        \draw (423,134.75) .. controls (426.78,151.76) and (440.83,171.9) .. (454.6,180.81);
        \draw [shift={(457,182.25)},rotate=208.89,fill={rgb,255:red,0;green,0;blue,0},line width=0.08, draw opacity=0] (5.36,-2.57) -- (0,0) -- (5.36,2.57) -- (3.56,0) -- cycle;
        \draw (441.5,113.75) .. controls (430.35,132.67) and (432.83,159.11) .. (456.72,175.74);
        \draw [shift={(459,177.25)},rotate=212.4,fill={rgb,255:red,0;green,0;blue,0},line width=0.08, draw opacity=0] (5.36,-2.57) -- (0,0) -- (5.36,2.57) -- (3.56,0) -- cycle;
        \draw (448.79,151.98) .. controls (451.26,161.78) and (453.36,165.23) .. (459.54,172.48);
        \draw [shift={(461.5,174.75)},rotate=229.04,fill={rgb,255:red,0;green,0;blue,0},line width=0.08, draw opacity=0] (5.36,-2.57) -- (0,0) -- (5.36,2.57) -- (3.56,0) -- cycle;
        \draw (457.28,182.52) .. controls (457.28,177.9) and (460.84,174.15) .. (465.22,174.15) .. controls (469.61,174.15) and (473.17,177.9) .. (473.17,182.52) .. controls (473.17,187.15) and (469.61,190.9) .. (465.22,190.9) .. controls (460.84,190.9) and (457.28,187.15) .. (457.28,182.52) -- cycle;
        \draw (459.49,182.52) .. controls (459.49,179.19) and (462.06,176.48) .. (465.22,176.48) .. controls (468.39,176.48) and (470.95,179.19) .. (470.95,182.52) .. controls (470.95,185.86) and (468.39,188.56) .. (465.22,188.56) .. controls (462.06,188.56) and (459.49,185.86) .. (459.49,182.52) -- cycle;
        \draw (471.75,186.25) .. controls (491.81,190.55) and (493.62,177.97) .. (475.02,177.7);
        \draw [shift={(472.25,177.75)},rotate=357.34,fill={rgb,255:red,0;green,0;blue,0},line width=0.08, draw opacity=0] (5.36,-2.57) -- (0,0) -- (5.36,2.57) -- (3.56,0) -- cycle;
        \draw (610.62,120.69) .. controls (607.69,99.39) and (621.98,99.2) .. (618.79,118.82);
        \draw [shift={(618.22,121.73)},rotate=282.86,fill={rgb,255:red,0;green,0;blue,0},line width=0.08, draw opacity=0] (5.36,-2.57) -- (0,0) -- (5.36,2.57) -- (3.56,0) -- cycle;
        \draw (359.62,169.11) .. controls (359.62,165.77) and (362.18,163.06) .. (365.35,163.06) .. controls (368.51,163.06) and (371.08,165.77) .. (371.08,169.11) .. controls (371.08,172.44) and (368.51,175.15) .. (365.35,175.15) .. controls (362.18,175.15) and (359.62,172.44) .. (359.62,169.11) -- cycle;
        \draw (478,116.75) .. controls (464.15,132.03) and (459.43,152.33) .. (464.45,171.45);
        \draw [shift={(465.22,174.15)},rotate=252.63,fill={rgb,255:red,0;green,0;blue,0},line width=0.08, draw opacity=0] (5.36,-2.57) -- (0,0) -- (5.36,2.57) -- (3.56,0) -- cycle;
        \draw (478.63,153.9) .. controls (477.18,162.22) and (475.44,167.1) .. (470.84,173);
        \draw [shift={(469,175.25)},rotate=310.6,fill={rgb,255:red,0;green,0;blue,0},line width=0.08, draw opacity=0] (5.36,-2.57) -- (0,0) -- (5.36,2.57) -- (3.56,0) -- cycle;
        \draw (125.7,113.23) node [anchor=north west,inner sep=0.75pt, font=\footnotesize,align=left] {$\mathsf{\mathsf{init}}$};
        \draw (289.46,87.35) node [anchor=north west,inner sep=0.75pt, font=\footnotesize,align=left] {$\mathsf{final}$};
        \draw (94.89,82.5) node [anchor=north west,inner sep=0.75pt,align=left] {$L_{i}$};
        \draw (108.18,122.4) node [anchor=north west,inner sep=0.75pt, font=\small,align=left] {$q_{i}$};
        \draw (157.89,142) node [anchor=north west,inner sep=0.75pt,align=left] {$\aut{A}_{i}$};
        \draw (254.67,103.58) node [anchor=north west,inner sep=0.75pt, font=\footnotesize,rotate=-9.27,align=left] {$\mathsf{final}$};
        \draw (252.47,139.87) node [anchor=north west,inner sep=0.75pt, font=\footnotesize,rotate=-349.11,align=left] {$\mathsf{final}$};
        \draw (377.7,113.23) node [anchor=north west,inner sep=0.75pt, font=\footnotesize,align=left] {$\mathsf{init}$};
        \draw (540.46,88.35) node [anchor=north west,inner sep=0.75pt, font=\footnotesize,align=left] {$\mathsf{final}$};
        \draw (348.39,77) node [anchor=north west,inner sep=0.75pt,align=left] {$L_{i}^{c}$};
        \draw (360.18,122.4) node [anchor=north west,inner sep=0.75pt, font=\small,align=left] {$q_{i}$};
        \draw (404.89,142.5) node [anchor=north west,inner sep=0.75pt,align=left] {$\aut{A}_{i}$};
        \draw (506.5,103.86) node [anchor=north west,inner sep=0.75pt, font=\footnotesize,rotate=-8.14,align=left] {$\mathsf{final}$};
        \draw (504.47,139.87) node [anchor=north west,inner sep=0.75pt, font=\footnotesize,rotate=-349.11,align=left] {$\mathsf{final}$};
        \draw (326.7,141.23) node [anchor=north west,inner sep=0.75pt, font=\footnotesize,align=left] {$\neq \mathsf{init}$};
        \draw (563.2,110.23) node [anchor=north west,inner sep=0.75pt, font=\footnotesize,align=left] {$\neq \mathsf{final}$};
        \draw (409.96,63.35) node [anchor=north west,inner sep=0.75pt, font=\footnotesize,align=left] {$\mathsf{final}$};
        \draw (502.96,61.35) node [anchor=north west,inner sep=0.75pt, font=\footnotesize,align=left] {$*$};
        \draw (389.96,167.35) node [anchor=north west,inner sep=0.75pt, font=\footnotesize,align=left] {$*$};
        \draw (418.2,168.73) node [anchor=north west,inner sep=0.75pt, font=\footnotesize,align=left] {$\mathsf{init}$};
        \draw (489.96,179.85) node [anchor=north west,inner sep=0.75pt, font=\footnotesize,align=left] {$*$};
        \draw (610.46,92.85) node [anchor=north west,inner sep=0.75pt, font=\footnotesize,align=left] {$*$};
    \end{tikzpicture}
    \caption{\DBWs{} for $L_i$ and $L_i^c$ used in the reduction for \cref{theorem:ne-pspace-complete}.}
    \label{fig:reduction-constrained-nash-existence-pspace}
\end{figure}
\begin{proof}[Hardness proof of \cref{theorem:ne-pspace-complete}]
    We use a reduction from the non-emptiness problem of the intersection of $n$ \DFAs{},\footnote{Deterministic finite automata, accepting a set of finite words.} which is \pspaceComplete{}~\cite{Kozen77}, to the constrained NE existence problem. Consider $n$ \DFAs{} $\aut{A}_i$ over the alphabet $\Sigma$ (that are complete). We can suppose that for each $i$, there exists some finite word $w_i \in \Sigma^*$ of polynomial length such that $w_i \in \lang(\aut{A}_i)$ as otherwise, the intersection of the $n$ \DFAs{} is empty. This can be checked in polynomial time.

    We are going to construct an \X{} game $\game = (\arena,(\pay_i)_{i \in \{1,\dots, n+1\}})$ with $n+1$ players, each using two payoffs except player $n+1$ that uses one payoff, in the following way. Let $\Sigma' = \Sigma \cup \{\mathsf{init},\mathsf{final}\}$. For each $i \in \{1,\ldots,n\} $, we define the language $L_i = \{\mathsf{init} \cdot w \cdot \mathsf{final}^\omega \mid w \in \lang(\aut{A}_i)\}$ and denote by $L_i^c$ its complement $(\Sigma')^\omega \ssetminus L_i$. Then, we define $\outset_i = \{\outelement_i,\outelement_i^c\}$,  $\mathord{\R_i} = \{(\outelement_i^c,\outelement_i)\}$, and $\pay_i$ as $\pay_i(x) = \outelement_i$ if $x \in L_i$, and $\pay_i(x) = \outelement_i^c$ otherwise. Intuitively, player~$i$ prefers words of $L_i$ to those of $L_i^c$. For player~$n+1$, we define only one payoff $\outelement_{n+1}$ such that $\mathord{\R_{n+1}} = \varnothing$ and $\pay_{n+1}$ is constant. The set of vertices of the arena $\arena$ is equal to $\Sigma'$ and all vertices belong to player~$n+1$. Its set of edges is $\Sigma' \times \Sigma'$ and the initial vertex is $\mathsf{init}$. Hence, all plays from $\mathsf{init}$ are NE outcomes due to the preference relation $\R_{n+1}$ of player~$n+1$. We finally define  the constraint $d_i = \outelement_i$ for each $i \in \{1, \dots,n \}$, and the constraint $d_{n+1} = \outelement_{n+1}$.
    So, by the definition of $\R_i$, the only way to satisfy the threshold $d_i$ for each player~$i$ is to have a play $\rho$ such that $\pay_i(\rho) = \outelement_i$.

    Let us show that the intersection of the $n$ \DFAs{} is non-empty if and only if there exists an NE from $\mathsf{init}$ whose outcome $\rho$ is such that $\pay_i(\rho) = \outelement_i$ for all $i \in \{1,\dots, n+1\}$. Suppose that $w \in \Sigma^*$ belongs to the intersection $\bigcap_{i=1}^n \lang(\aut{A}_i)$. Then for all $i \in \{1,\dots, n\}$, $\rho = \mathsf{init} \cdot w \cdot \mathsf{final}^\omega$ belongs to $L_i$ and satisfies $\pay_i(\rho) = \outelement_i$. We also have that $\rho$ is an NE outcome. Conversely, suppose that there exists an NE outcome $\rho$ from $\mathsf{init}$ such that $\pay_i(\rho) = \outelement_i$ for each player~$i$. It follows that $\rho \in L_i$ for all $i \in \{1,\dots, n\}$, by definition of $\pay_i$ and $d_i$. That is, $\rho = \mathsf{init} \cdot w \cdot \mathsf{final}^\omega$ for some finite word $w$ belonging to each $\lang(\aut{A}_i)$, thus to the intersection $\bigcap_{i=1}^n \lang(\aut{A}_i)$.

    It remains to show that $L_i$ and $L_i^c$ are both accepted by a \DPW{} (it is trivial for $\pay_{n+1}^{-1}(\outelement_{n+1})$). A \DBW{} is given for $L_i$ (resp.\ $L_i^c$) on the left (resp.\ right) part of \cref{fig:reduction-constrained-nash-existence-pspace}. A gray state represents an accepting state of the \DFA{} $\aut{A}_i$ while double circled states are accepting states of the \DBW{}. Note that a word $x$ belongs to $L_i^c$ if, and only if,
    \begin{itemize}
        \item $x$ begins with a symbol $a$ different from $\mathsf{init}$,
        \item $x = \mathsf{init} \cdot x'$ with $x' \in \Sigma^\omega$ ($x$ does not leave $\aut{A}_i$),
        \item $x = \mathsf{init} \cdot w \cdot \mathsf{init} \cdot x'$ with $w \in \Sigma^*$ and $x' \in \Sigma'^\omega$ ($x$ leaves $\aut{A}_i$ with symbol $\mathsf{init}$),
        \item $x = \mathsf{init} \cdot w \cdot \mathsf{final} \cdot x'$ with $w \not\in \lang(\aut{A}_i)$ and $x' \in \Sigma'^\omega$ ($x$ leaves $\aut{A}_i$ from one of its non-final states with symbol $\mathsf{final}$),
        \item $x = \mathsf{init} \cdot w \cdot \mathsf{final}^+ \cdot a \cdot x'$ with $w \in \lang(\aut{A}_i)$, $a \ne \mathsf{final}$, and $x' \in \Sigma'^\omega$, ($x$ leaves $\aut{A}_i$ from one of its final states with symbol $\mathsf{final}$ and not with suffix $\mathsf{final}^\omega$).
    \end{itemize}
    This concludes the proof.
\end{proof}
\end{document}